\documentclass[nopreprintline, preprint, authoryear]{elsarticle}
\biboptions{longnamesfirst}
\usepackage{graphicx} 
\usepackage{xcolor}
\usepackage{bm}
\usepackage{amsmath}
\usepackage{amsfonts}
\usepackage{amsthm}
\usepackage{tikz}
\usetikzlibrary{trees, positioning}
\usepackage{hyperref}
\usepackage{booktabs}
\usepackage{overpic}
\usepackage{multirow}
\usepackage{float}
\usepackage[skip=2pt]{caption}
\usepackage{subcaption}
\usepackage[margin=1in]{geometry}

\usepackage{changes}
\usepackage{soul}
\soulregister\citep7
\soulregister\cite7
\soulregister\textit7
\soulregister\ref7
\setdeletedmarkup{\st{#1}}
\setaddedmarkup{\textcolor{codered}{#1}}
\setdeletedmarkup{\textcolor{codeblue}{\st{#1}}}

\newcommand{\E}{\mathbb{E}}
\newcommand{\rr}{\mathbb{R}}
\newcommand{\nn}{\mathcal{N}}

\newcommand{\iw}{\operatorname{IW}}
\newcommand{\mt}{\operatorname{mt}}
\newcommand{\simIID}{\,\overset{\text{IID}}{\sim}\,}

\newcommand{\xbf}{\mathbf{x}}
\newcommand{\ybf}{\mathbf{y}}

\newcommand{\ubf}{\mathbf{u}}
\newcommand{\bbf}{\mathbf{b}}
\newcommand{\rbf}{\mathbf{r}}
\newcommand{\Abf}{\mathbf{A}}
\newcommand{\Rbf}{\mathbf{R}}
\newcommand{\Xbf}{\mathbf{X}}

\newcommand{\Vbf}{\mathbf{V}}
\newcommand{\Qbf}{\mathbf{Q}}
\newcommand{\Zbf}{\mathbf{Z}}
\newcommand{\Tbf}{\mathbf{T}}
\newcommand{\Ibf}{\mathbf{I}}
\newcommand{\Sbf}{\mathbf{S}}
\newcommand{\Wbf}{\mathbf{W}}

\newcommand{\Pbf}{\mathbf{P}}
\newcommand{\mubf}{\bm{\mu}}
\newcommand{\epsilonbf}{\bm{\varepsilon}}
\newcommand{\nubf}{\bm{\nu}}
\newcommand{\etabf}{\bm{\eta}}
\newcommand{\omegabf}{\bm{\omega}}
\newcommand{\gammabf}{\bm{\gamma}}
\newcommand{\zetabf}{\bm{\zeta}}
\newcommand{\Sigmabf}{\bm{\Sigma}}
\newcommand{\Psibf}{\bm{\Psi}}

\newcommand{\BH}{\mathbf{\widehat{B}}}
\newcommand{\UH}{\mathbf{\widehat{U}}}
\newcommand{\YH}{\mathbf{\widehat{Y}}}
\newcommand{\WH}{\mathbf{\widehat{W}}}
\newcommand{\SH}{\mathbf{\widehat{S}}}
\newcommand{\yH}{\mathbf{\widehat{y}}}
\newcommand{\uH}{\mathbf{\widehat{u}}}
\newcommand{\bH}{\mathbf{\widehat{b}}}
\newcommand{\PsiH}{\bm{\widehat{\Psi}}}

\newcommand{\BT}{\mathbf{\widetilde{B}}}
\newcommand{\YT}{\mathbf{\widetilde{Y}}}
\newcommand{\WT}{\mathbf{\widetilde{W}}}
\newcommand{\yT}{\mathbf{\widetilde{y}}}

\newcommand{\bT}{\mathbf{\widetilde{b}}}
\newcommand{\SigmaT}{\bm{\widetilde{\Sigma}}}

\newcommand{\uhat}{\widehat{u}}
\newcommand{\bhat}{\widehat{b}}
\newcommand{\ghat}{\widehat{g}}
\newcommand{\Qhat}{\widehat{Q}}

\newcommand{\pihat}{\widehat{\pi}}
\newcommand{\nuhat}{\widehat{\nu}}

\newcommand{\Util}{\widetilde{U}}
\newcommand{\util}{\widetilde{u}}

\newcommand{\pitil}{\widetilde{\pi}}
\newcommand{\nutil}{\widetilde{\nu}}

\newcommand{\sigmatil}{\widetilde{\sigma}}

\newcommand{\Scal}{\mathcal{S}}

\newcommand{\tRec}{t-Rec}
\newcommand{\train}{training length}
\newcommand{\trains}{training lengths}

\definecolor{codegreen}{rgb}{0,0.6,0}
\definecolor{codeorange}{RGB}{250,100,0}
\definecolor{codepurple}{rgb}{0.58,0,0.82}
\definecolor{codeblue}{RGB}{50,50,250}

\definecolor{codered}{RGB}{255,40,40}    

\newtheorem{proposition}{Proposition}
\newtheorem{theorem}{Theorem}
\newtheorem{lemma}{Lemma}

\DeclareMathOperator*{\argmax}{arg\,max}

\begin{document}

\begin{frontmatter}

\title{Modeling uncertainty in the covariance matrix for probabilistic forecast reconciliation}

\author[cc,lz]{Chiara Carrara\corref{cor}}
\address[cc]{University of Pavia, Italy}
\ead{chiara.carrara9802@gmail.com}

\author[lz]{Dario Azzimonti}

\author[lz]{Giorgio Corani}

\author[lz]{Lorenzo Zambon}
\address[lz]{SUPSI, Istituto Dalle Molle di Studi sull'Intelligenza Artificiale (IDSIA), Switzerland}

\cortext[cor]{Corresponding author}


\begin{abstract} 
In minimum trace (MinT) forecast reconciliation, the covariance matrix of the base forecast errors plays a crucial role.
Typically, this matrix is estimated and then treated as known. 
This can lead to an underestimation of the variance of the predictive distribution. To address the problem, we propose a Bayesian reconciliation model that accounts for uncertainty in the estimation of the covariance matrix.
By adopting an Inverse-Wishart prior and assuming Gaussian residuals, the
reconciled predictive distribution follows a multivariate t-distribution, obtained in closed form, rather than a multivariate Gaussian distribution.  
We evaluate our method on three tourism-related datasets, including a new publicly available
dataset.
Empirical results show that our approach consistently improves prediction intervals compared to MinT reconciliation.
\end{abstract}


\end{frontmatter}

\section{Introduction}

Hierarchical time series are collections of time series that adhere to a set of linear constraints. 
For example, the sales of individual items (the lower level of the hierarchy) are summed up to obtain the sales of different categories or stores (the upper levels of the  hierarchy).
Forecasts for hierarchical  time series need to be \textit{coherent}, meaning that the lower-level forecasts sum to the forecasts of the higher levels.
The \textit{base forecasts}, which are
independently produced  for each time series,  are generally incoherent. 
\textit{Forecast reconciliation} adjusts the base forecasts to produce  coherent forecasts for the entire hierarchy.

While early work focused on the reconciliation of point forecasts \citep{hyndman2011optimal, wickramasuriya2019optimal}, recent literature has increasingly turned toward \textit{probabilistic reconciliation}, which quantifies the uncertainty of the reconciled predictions.  
One of the most established approaches is MinT \citep{wickramasuriya2019optimal}, which relies on the covariance matrix of the base forecast errors. 
There exist alternative techniques that avoid the covariance matrix by learning optimal projections through score minimization \citep{panagiotelis2022reconc}, by learning linear regression models for reconciliation weights \citep{moller2024optimal}, or by employing end-to-end models that directly produce coherent forecasts \citep{rangapuram2021end, olivares_probabilistic_2023, bertani2025joint}. Nonetheless, covariance-based methods are a standard in the field.
Indeed, besides being computationally efficient, MinT has proven to be optimal with respect to both the expected mean squared error of point forecasts \citep{wickramasuriya2019optimal} and, in the case of Gaussian base forecasts, the logarithmic score of the reconciled distribution \citep{wickramasuriya2024probabilistic}.


However, this optimality holds only under the assumption that the covariance matrix of the base forecast errors is known.
In practice, this matrix is instead estimated from in-sample residuals. 
The performance of both reconciled point forecasts \citep{panagiotelis2021geometric} and distributions \citep{girolimetto2024cross} can vary significantly, depending on how the matrix is estimated.
Since the number of parameters to be estimated grows quadratically with the number of time series in the hierarchy, the estimation uncertainty can be substantial \citep{pritularga2021stochastic}. 
If the covariance matrix is treated as known, this uncertainty is ignored, 
leading to an underestimation of the predictive variance. 
Indeed, MinT reconciliation is proven to reduce the variance of the predictive distribution \citep{zambon2024properties}. While such a reduction is optimal if the covariance matrix is known, ignoring the estimation error results in reconciled intervals that are systematically too narrow.
Recent works combine reconciliation with conformal prediction 
so that the prediction intervals achieve  nominal coverage \citep{principato2025conformalpredictionhierarchicaldata}. 
However, to date, no reconciliation framework accounts for the uncertainty of the covariance matrix.

In this paper, we propose a Bayesian reconciliation approach that  models the uncertainty of the covariance matrix. 
By adopting a Bayesian framework, we  integrate parameter uncertainty into the predictive distribution \citep[Chap.~1.3]{bda2013}, yielding a more principled quantification of the uncertainty.
We assume an Inverse-Wishart (IW) prior distribution for the covariance matrix $\Wbf$ and we update it using the in-sample residuals, which we assume to be jointly Gaussian, conditionally on $\Wbf$.
From the resulting posterior, we analytically derive the incoherent posterior predictive distribution,
which is a multivariate $t$.
We then apply reconciliation via conditioning \citep{corani_reconc_ecml, zambon2024efficient}.
We obtain in closed form the reconciled distribution, which is still a multivariate $t$ with updated parameters; we thus call our method \textit{\tRec}.
This represents the only known analytical solution for reconciliation via conditioning besides the Gaussian case \citep{zambon2024properties}, 
compared to which it requires only minimal computational overhead.
%
%
We evaluate our method on three distinct datasets, including a novel dataset on Swiss tourism.
In all settings, \textit{\tRec} performs similarly to Gaussian reconciliation \citep{wickramasuriya2024probabilistic, zambon2024efficient} on the point forecasts, but consistently improves the probabilistic forecasts.

The paper is structured as follows. In Section~\ref{sec: prob reconc} we introduce the notation and recall the concepts of probabilistic reconciliation. In Section~\ref{sec: t-rec} we analytically derive \textit{\tRec} and discuss its specifications. In Section~\ref{sec: simulations} we demonstrate the model's behavior on a minimal hierarchy. Section~\ref{sec: experiment} reports empirical results on real datasets. Section~\ref{sec: conclusion} presents our conclusions.

\section{Probabilistic reconciliation}
\label{sec: prob reconc}

\begin{figure}[h!]
    \begin{center}
        \begin{tikzpicture}[every node/.style={rectangle,draw,align=center,rounded corners=4,font=\large},level distance=1.cm,
  level 1/.style={sibling distance=5.cm, edge from parent fork down},
  level 2/.style={sibling distance=2.1cm, edge from parent fork down},
  level 3/.style={sibling distance=1.2cm, edge from parent fork down}
]
  \node[fill=black!20,minimum width=1cm] {T}
    child {node[fill=black!10,minimum width=1cm] {A}
      child {node[fill=black!0,minimum width=1cm] {AA}}
      child {node[fill=black!0,minimum width=1cm] {AB}}
    }
    child {node[fill=black!10,minimum width=1cm] {B}
      child {node[fill=black!0,minimum width=1cm] {BA}}
      child {node[fill=black!0,minimum width=1cm] {BB}}
    };
\end{tikzpicture}
\vspace{2mm}
    \caption{Hierarchy with 4 bottom and 3 upper time series.}
    \label{fig: simple hier}
    \end{center}
\end{figure}

Fig.~\ref{fig: simple hier} is an example of a time series hierarchy, which can, for instance, represent the total visitors of a country (T), disaggregated by zones (A, B) and regions (AA, AB, BA, BB), so that:
\begin{equation*}
T = A+B, \quad A = AA+AB, \quad B=BA+BB.
\end{equation*}
At any time $t$, the hierarchical constraints between time series are given by: 
\begin{equation}\label{eq: u=Ab}
\ubf_t = \Abf \bbf_t,    
\end{equation}
where
$\mathbf{A} \in \rr^{n_u \times n_b}$ is the \textit{aggregation matrix}, which is made of zeros and ones and specifies how the \textit{bottom} time series $\bbf_t\in \rr^{n_b}$ aggregate to the \textit{upper} time series $\ubf_t \in \rr^{n_u}$.
Denoting $\ybf_t = \left[\ubf_t^T, \bbf_t^T\right]^T \in \rr^n$,
we can equivalently encode the hierarchy as $\ybf_t = \Sbf \bbf_t$, 
where   the \textit{summing matrix} $\Sbf \in \rr^{n \times n_b}$ is:
\begin{equation*}
    \Sbf = \begin{bmatrix}
    \Abf\\
    \mathbf{I}_{n_b \times n_b}
\end{bmatrix}.
\end{equation*}
\citet{Girolimetto2024} introduce a more general framework for linearly constrained time series that does not rely on the summing matrix. 
Our method is also applicable to that case; however, 
in the following, we adopt the standard setting based on $\mathbf{S}$.

We consider probabilistic forecasts in the form of a joint predictive distribution $\pihat$ over $\rr^n$.
A  forecast distribution is \textit{coherent} if it  gives positive probability only to regions of points that satisfy the hierarchical constraints. 
Thus the support of a coherent distribution lies on  the \textit{coherent subspace} $\Scal$, defined as $\Scal := \{\ybf \in \rr^n \; \text{\small such that} \; \ybf = \Sbf \bbf \}$.
Probabilistic reconciliation \citep{panagiotelis2022reconc, zambon2024efficient}  adjusts the initially incoherent base forecast distribution $\pihat$ to obtain a coherent reconciled forecast distribution $\pitil$.

\paragraph{MinT reconciliation}
The minimum trace (MinT) reconciliation method was  introduced for the reconciliation of point forecasts \citep{wickramasuriya2019optimal}. 
Given the $h$-step ahead base point forecasts $\yH_{t+h|t}$, computed at time $t$ for time $t+h$, the reconciled point forecasts $\yT_{t+h|t}$ are obtained by projecting $\yH_{t+h|t}$ on the coherent subspace $\Scal$:
\begin{equation}\label{eq: minT point forecast}
\yT_{t+h|t} = \textbf{S} \Pbf_h \yH_{t+h|t},
\end{equation}
where $\Pbf_h = (\textbf{S}^T \Wbf_h^{-1} \textbf{S})^{-1} \textbf{S}^T \Wbf_h^{-1}$, and $\Wbf_h$ is the covariance matrix of the $h$-step ahead base forecast errors.
Assuming that the base forecasts $\yH_{t+h|t}$ are unbiased and that $\Wbf_h$ is known, 
the reconciled point forecasts $\yT_{t+h|t}$ 
minimize the expected mean squared error (MSE).

MinT was later extended to the probabilistic case by projecting on $\Scal$ the entire forecast distribution $\pihat$ \citep{panagiotelis2022reconc}.
If the joint base forecast distribution is Gaussian, i.e., $\pihat = \nn(\yH_{t+h|t},\, \Wbf_h)$,
the reconciled distribution is also Gaussian:
\[
\pitil = \nn(\yT_{t+h|t},\, \WT_h),
\]
where $\yT_{t+h|t}$ is given by Eq.~\eqref{eq: minT point forecast} and $\WT_h = \Sbf\Pbf_h\Wbf_h\Pbf_h^T\Sbf^T$.
In this case, $\pitil$ is optimal with respect to the log-score \citep{wickramasuriya2024probabilistic}.

\paragraph{Probabilistic reconciliation via conditioning}

\begin{figure}[h!]
\centering
\begin{subfigure}{.5\textwidth}
  \includegraphics[width=0.9\linewidth]{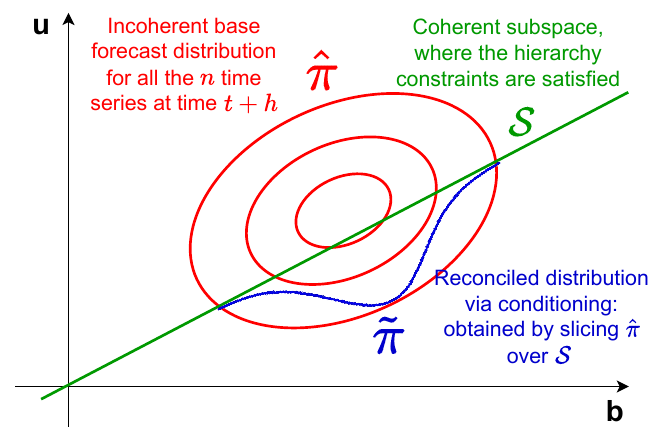}
\end{subfigure}%
\begin{subfigure}{.5\textwidth}
  \hspace{2mm}
  \includegraphics[width=0.9\linewidth]{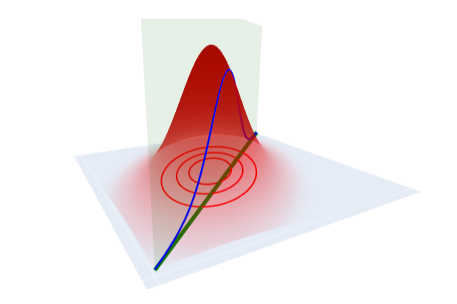}
\end{subfigure}
\caption{
Reconciliation via conditioning: the reconciled distribution $\pitil$ is obtained by 
slicing $\pihat$ on the coherent subspace $\Scal$.
} 
\label{fig: rec via cond}
\end{figure}

An alternative approach to probabilistic reconciliation derives the reconciled distribution by conditioning the joint base forecast distribution on the hierarchical constraints
\citep{corani_reconc_ecml, zambon2024efficient}.
Given \mbox{$\YH_{t+h|t} = \big[\UH^T_{t+h|t},\; \BH^T_{t+h|t} \big]^T$}, 
a random vector distributed as the $h$-step ahead base forecast distribution $\pihat$, the reconciled distribution of the bottom series is specified by 
\begin{equation}\label{eq:rec_via_cond}
    \BT_{t+h|t} \;\;\sim\;\; \BH_{t+h|t} \;\Big|\; \big(\UH_{t+h|t} - \Abf\BH_{t+h|t} = 0 \big), 
\end{equation}
and it can be extended to the entire hierarchy via
\begin{equation}\label{eq:rec_via_cond_Y}
\YT_{t+h|t} = \Sbf \BT_{t+h|t}.
\end{equation}
Eq.~\eqref{eq:rec_via_cond} allows the computation of the reconciled density up to a normalizing constant.
Generally, the reconciled distribution is not available in closed form, thus requiring sampling approaches \citep{corani_probabilistic_2023, zambon2024efficient}.
To date, the only exception is when the base forecast distribution is jointly Gaussian; in this case,
reconciliation via conditioning can be solved analytically and coincides with MinT \citep{zambon2024properties}.
Reconciliation via conditioning can be interpreted as slicing the incoherent joint density over the coherent subspace $\mathcal{S}$ (Fig.~\ref{fig: rec via cond}). When the base forecast distribution is Gaussian, both reconciliation via conditioning and MinT depend on the covariance matrix $\Wbf_h$, which must be estimated.

\paragraph{Estimation of $\Wbf_1$}
For $h=1$, $\Wbf_1$ is  estimated as the covariance matrix of the in-sample residuals \citep{wickramasuriya2019optimal, panagiotelis2021geometric, di2022forecast}, 
usually adopting the shrinkage estimator \citep{schafer2005shrinkage}:
\begin{equation}\label{eq: shrinkage}
    \WH_1 = (1 - \lambda_{\text{shrink}}) \SH_1 + \lambda_{\text{shrink}} \widehat{\mathbf{D}}_1,
\end{equation}
where $\SH_1$ is the sample covariance matrix of the one-step ahead residuals, $\widehat{\mathbf{D}}_1$ is the diagonal matrix of the sample variances, and $\lambda_{\text{shrink}} \in [0,1]$ is the  shrinkage parameter.
In what follows, we denote by \textit{MinT} the MinT reconciliation method with the shrinkage estimator $\WH_1$.
Double shrinkage approaches (\citealp[Chap. 4.3.3]{wickramasuriya2017}; \citealp{Carrara2025}) shrink $\SH_1$  towards two targets simultaneously; however, they provide only limited improvements over traditional shrinkage or incur in significant computational cost.

\paragraph{Estimation of $\Wbf_h$}
There is no standard method of estimating $\Wbf_h$ for $h>1$.
\cite{wickramasuriya2019optimal} assume $\Wbf_h = k_h \Wbf_1$ for some constant $k_h$, which can be left unspecified for computing the reconciled point forecasts. Some heuristics for  setting $k_h$  in the probabilistic framework
have been proposed by \cite{corani_reconc_ecml}.
\cite{girolimetto2024cross} instead estimate $\Wbf_h$ using multistep residuals. 
While our method can be applied with any choice of $\Wbf_h$,  in our experiments, for simplicity, we only consider  $h=1$.
In the following, we drop the $h$ subscript for ease of notation.

\section{Probabilistic reconciliation modeling the uncertainty on $\Wbf$} \label{sec: t-rec}
We propose a Bayesian approach  which treats  $\Wbf$
as a random variable, rather than as a fixed matrix. 
We set a prior distribution on $\Wbf$ and update it using the information from the in-sample residuals. 
We account for the uncertainty on $\Wbf$ by modeling its entire posterior distribution, rather than taking a point estimate $\WH$ as in previous works.
By adopting an Inverse-Wishart prior and assuming Gaussian residuals, the resulting incoherent predictive distribution is a multivariate t.
This predictive distribution is the result of the Bayesian treatment of $\mathbf{W}$ under Gaussian errors; 
we do not make any heavy-tailed assumption on the base forecasting model. 
Finally, we reconcile the incoherent predictive distribution via conditioning, 
obtaining a coherent multivariate t in closed form.
We call our method \textit{\tRec}. 

\paragraph{Prior of $\Wbf$}
We adopt an Inverse-Wishart ($\iw$) prior distribution for $\Wbf$:
\begin{equation}\label{eq: prior on W}
\Wbf \;\sim\; \iw(\Psibf_0, \nu_0).
\end{equation}
The Inverse-Wishart is a distribution defined on the space of positive-definite matrices \citep{gupta2018matrix};
$\nu_0 \in \rr$ represents the degrees of freedom and $\Psibf_0 \in \rr^{n \times n}$ is the scale matrix, which is proportional to the mean.

\paragraph{Likelihood of the residuals}

We assume that the in-sample residuals 
$\rbf_1, \dots, \rbf_T$, defined as $\rbf_j = \yH_j - \ybf_j$ for $j=1,\dots,T$, are independent and identically distributed (IID).
This is a tenable assumption if the models that generate the base forecasts are selected within a wide set of competing models
(\citealp[Chap.~1.4.1]{svetunkov2023forecasting}; \citealp[Chap.~5.4]{hyndman2021forecasting}) via statistical  model selection. 
Moreover, as in \cite{corani_reconc_ecml, wickramasuriya2024probabilistic}, 
we assume that, conditionally on $\Wbf$, the residuals are distributed as a multivariate Gaussian with zero mean and covariance matrix $\Wbf$:
\begin{equation}\label{eq: likelihood}
\rbf_1, \dots, \rbf_T \;|\; \Wbf \;\simIID\; \nn(\bm{0}, \Wbf). 
\end{equation} 

\paragraph{Posterior of $\Wbf$}
Given the in-sample residuals $\Rbf = \begin{bmatrix} \rbf_1 & \dots & \rbf_T \end{bmatrix}$, we
compute the posterior distribution of $\Wbf$ via Bayes' rule \citep[Chap.~1.3]{bda2013}:
\[
\pi\left(\Wbf\,|\,\Rbf\right) = \frac{\pi\left(\Wbf\right)\pi\left(\Rbf\,|\,\Wbf\right)}{\pi\left(\Rbf\right)}.
\]
Generally, in Bayesian models, the posterior distribution  is not available in a closed parametric form, thus requiring sampling techniques (e.g., MCMC).
However, since the Inverse-Wishart prior distribution is conjugate to the Gaussian likelihood \citep[Chap.~3.6]{bda2013}, the posterior distribution remains Inverse-Wishart and can be computed analytically:
\begin{equation}\label{eq: posterior on W}
\Wbf\;|\;\Rbf \;\sim\;  \iw(\Wbf;\, \Psi', \nu'),
\end{equation}
with parameters $\nu' = \nu_0 + T$ and $\Psibf' = \Psibf_0 + \Rbf \Rbf^T \in \rr^{n \times n}$.


\paragraph{Multivariate t incoherent predictive distribution}
From Eq.~\eqref{eq: likelihood} follows that, given $\Wbf$, the conditional predictive distribution $\pihat\left(\ybf\,|\,\Wbf\right)$ is Gaussian.
To obtain the posterior predictive distribution $\pihat\left(\ybf\,|\,\Rbf\right)$, we compute the posterior joint distribution $\pihat\left(\ybf, \Wbf\,|\, \Rbf\right)$ and then we marginalize out $\Wbf$; this yields a multivariate t-distribution \citep[Chap.~3.6]{bda2013}:
\begin{align}
\pihat\left(\ybf\,|\,\Rbf\right)
&= \int \pihat\left(\ybf, \Wbf\,|\, \Rbf\right) \, d\Wbf \nonumber\\
&= \int \pihat\left(\ybf\,|\,\Wbf\right) \; \pi\left(\Wbf\,|\,\Rbf\right) \, d\Wbf \nonumber\\
&= \int \nn\left(\ybf;\,\yH,\,\Wbf\right) \; \iw(\Wbf;\, \Psi', \nu') \, d\Wbf \nonumber\\
&= \mt\left(\ybf;\, \yH,\, \frac{1}{\nu'-n+1} \Psibf',\, \nu'-n+1\right). \label{eq: posterior predictive mt}
\end{align}

\paragraph{Reconciliation}
From Eq.~\eqref{eq: posterior predictive mt}, the incoherent posterior predictive distribution is a multivariate t-distribution:
\[
\big[\UH^T,\; \BH^T \big]^T \;\sim\;  \mt\left(\yH,\, \frac{1}{\nu'-n+1} \Psibf',\, \nu'-n+1\right),
\]
where 
$\yH = \big[\uH^T,\; \bH^T \big]^T$, with $\uH \in \rr^{n_u}$ and $\bH \in \rr^{n_b}$, and $\Psibf'$ can be decomposed block-by-block into
$\Psibf'_U \in \rr^{n_u \times n_u}$, $\Psibf'_B \in \rr^{n_b \times n_b}$, $\Psibf'_{UB} \in \rr^{n_u \times n_b}$.
%
In the following theorem, we provide the  closed-form expression of the distribution reconciled  via conditioning, i.e.,  $\BT \,\sim\, \BH \,\Big|\, \big(\UH - \Abf\BH = 0 \big)$.
The proof, reported in \ref{app: rec via cond t},  exploits the closure of the multivariate t-distribution  under affine transformations \citep{Kotz_Nadarajah_2004} and conditioning \citep{ding2016conditional}.

\begin{theorem}
\label{theorem: t reconc via conditioning}
Under the previous assumptions on prior (Eq.~\eqref{eq: prior on W}) and likelihood (Eq.~\eqref{eq: likelihood}), 
the reconciled distribution via conditioning of the bottom time series is a multivariate t-distribution:


\begin{equation} \label{eq: reconciled bottom distribution posterior}
\BT \;\sim\; \mt\left(\bT,\, \SigmaT_B,\, \nutil \right),
\end{equation}
where
\begin{align}
\bT &= \bH + \left({\Psibf'_{UB}}^T - \Psibf'_B \Abf^T\right) \Qbf^{-1} (\Abf \bH - \uH), \label{eq: reconciled t mean bottom} \\
\SigmaT_B &= 
C
\left[\Psibf'_B - \left({\Psibf'_{UB}}^T - \Psibf'_B \Abf^T\right) \Qbf^{-1} \left({\Psibf'_{UB}}^T - \Psibf'_B \Abf^T\right)^T \right], \label{eq: reconciled t variance bottom} \\
\nutil &= \nu' - n_b + 1, \label{eq: reconciled t nu bottom}
\end{align}
%
and 
\begin{align}
&C =  \frac{1 +  (\Abf \bH - \uH)^T \Qbf^{-1} (\Abf \bH - \uH)}{\nutil}, \nonumber \\  
&\Qbf = \Psibf'_U - \Psibf'_{UB} \Abf^T - \Abf {\Psibf'_{UB}}^T + \Abf \Psibf'_B \Abf^T. \nonumber
\end{align}
\end{theorem}

Theorem~\ref{theorem: t reconc via conditioning} provides the reconciled distribution for the bottom-level time series. 
From Eq.~\eqref{eq:rec_via_cond_Y}, the reconciled distribution for the entire hierarchy
is a multivariate t with mean $\yT = \Sbf \bT$, scale matrix $\Sbf \SigmaT_B \Sbf^T$, and degrees of freedom $\nutil$.

\subsection{Parameters of the prior distribution}
\label{sec: prior_params}

We now discuss our choices for the scale matrix $\Psibf_0$ and the degrees of freedom $\nu_0$ of the Inverse-Wishart prior.

\paragraph{Scale matrix}
$\Psibf_0 \in \rr^{n \times n}$ is a positive definite matrix, proportional to the prior mean $\Psibf$ of the IW distribution via $\Psibf_0 = (\nu_0 - n - 1) \Psibf$.
To set $\Psibf$, we look at the covariance of the residuals of two simple methods \citep[Ch.~5.2]{hyndman2021forecasting}:  the naive and seasonal naive.
The naive method forecasts  all future values as equal to the last observation,
while the seasonal naive method sets each forecast equal to the last observed value of the same season (e.g., for a monthly time series, the forecast for all future Februarys is the last observed February value).
For each time series we select between naive and seasonal naive method using the Kwiatkowski-Phillips-Schmidt-Shin (KPSS) seasonality test \citep[Ch.~9.1]{hyndman2021forecasting}. 
We set  $\Psibf$ as the covariance matrix of the residuals of such simple methods,  estimated using the  shrinkage estimator \citep{schafer2005shrinkage} as in \textit{MinT}.

\paragraph{Degrees of freedom} The parameter $\nu_0$ controls the strength of the prior: a higher value indicates stronger confidence in the prior. 
We set $\nu_0 \in \rr$ by maximizing the out-of-sample log-score, estimated via Bayesian leave-one-out (LOO) cross-validation \citep{vehtari2017practical}:

\begin{align}
\nu_0 :&= \argmax_{\nu \in [n + 2, 5n]}\; \log\left(\prod_{i = 1}^T \pi(\mathbf{r}_i\,|\,\Rbf_{-i})\right) \nonumber \\
&= \argmax_{\nu \in [n + 2, 5n]}\; \sum_{i = 1}^T \log\left(\mt\Big(\mathbf{r}_i;\, 0,\, \frac{1}{\nu + T - n} (\Psibf_0 + \Rbf_{-i} \Rbf_{-i}^T),\, \nu + T - n\Big)\right), \label{eq: LOO}
\end{align}

where $\Rbf_{-i}\in \rr^{n \times (T - 1)}$ is the matrix of the in-sample residuals without the column $\rbf_i$,
and $\pi(\mathbf{r}_i\,|\,\Rbf_{-i})$ is the posterior predictive density of the residuals $\rbf_i$, given the observations $\Rbf_{-i}$; 
the derivation of Eq.~\eqref{eq: LOO} is analogous to that of Eq.~\eqref{eq: posterior predictive mt}.
The LOO approach is tenable for model selection in time series \citep{bergmeir2018note} if the residuals $\rbf_1, \dots, \rbf_T$ are approximately IID.
We restrict the search for the optimal value of $\nu_0$ to the interval $[n + 2, 5n]$. Indeed,  the expected value of the IW prior is not defined for $\nu_0 < n + 2$; on the other hand,
the optimal $\nu_0$ has been consistently below 
the upper bound $\nu_0 = 5n$.

A naive evaluation of the objective function in Eq.~\eqref{eq: LOO} is computationally expensive, as it requires  inverting $T$ matrices $\Psibf_0 + \Rbf_{-i} \Rbf_{-i}^T$, for $i=1,\dots,T$.
However, we achieve a significant speedup by noting that these matrices are rank-1 updates of the same matrix:
$
\Psibf_0 + \Rbf_{-i} \Rbf_{-i}^T = \Psibf_0 + \Rbf \Rbf^T - \rbf_i \rbf_i^T.
$
We thus use the Sherman-Morrison formula \citep{sherman1950adjustment} to obtain
\[
\left(\Psibf_0 + \Rbf_{-i} \Rbf_{-i}^T\right)^{-1}
= \left(\Psibf_0 + \Rbf \Rbf^T\right)^{-1}
+ \frac{\left(\Psibf_0 + \Rbf \Rbf^T\right)^{-1} \rbf_i \rbf_i^T \left(\Psibf_0 + \Rbf \Rbf^T\right)^{-1}}
{1 - \rbf_i^T \left(\Psibf_0 + \Rbf \Rbf^T\right)^{-1} \rbf_i}.
\]
This requires computing only \textit{once} the inverse of $\Psibf_0 + \Rbf \Rbf^T$.
We optimize $\nu$ using the \texttt{nloptr} R package \citep{johnsonNLoptr}. 
Even in our largest experiment ($n=111,\, T=60$), this optimization required less than $0.5$~s on a standard laptop,
representing the only significant computational overhead relative to \textit{MinT}.

\section{Reconciliation of  a minimal hierarchy}\label{sec: simulations}

\begin{figure}[ht]
\begin{center}
    \begin{tikzpicture}[every node/.style={rectangle,draw,align=center,rounded corners=4,font=\large},level distance=1.cm,
  level 1/.style={sibling distance=3.cm, edge from parent fork down}
]
  \node[fill=black!20,minimum width=1cm] {$U$}
    child {node[fill=black!5,minimum width=1cm] {$B_1$}}
    child {node[fill=black!5,minimum width=1cm] {$B_2$}};
\end{tikzpicture}
\end{center}
   \caption{Hierarchy with one upper and two bottom time series.}
\label{fig: min_hiearchy}
\end{figure}

We compare \textit{\tRec} and \textit{MinT} on the minimal hierarchy of Fig.~\ref{fig: min_hiearchy}, for which the formulas of the 
reconciled distributions are easier to interpret.
Let $\Psibf'$ be the scale matrix of the IW posterior distribution, specified by Eq.~\eqref{eq: posterior on W}, and let $\WH$ be the estimated covariance matrix used by MinT.
For simplicity, we focus here on the upper series, though analogous conclusions apply to the bottom series. 
The reconciled marginal distributions for \textit{MinT} (Gaussian) and \textit{\tRec} (univariate $t$) are characterized by the following mean and scale parameters:
\begin{align}
\util^{(MinT)} &= \Big( 1 - \frac{\ghat_u}{\Qhat}\Big)\, \uhat + \frac{\ghat_u}{\Qhat} \left(\bhat_1 + \bhat_2\right), 
\qquad \sigmatil_u^{(MinT)} = \sqrt{\widehat{W}_u - \frac{\ghat_u^2}{\Qhat}}, \label{eq:u_mint} \\
\util^{(t-Rec)} &= \Big( 1 - \frac{g'_u}{Q'}\Big)\, \uhat + \frac{g'_u}{Q'} \left(\bhat_1 + \bhat_2\right),
\qquad \sigmatil_u^{(t-Rec)} = \sqrt{C \left(\Psi_u' - \frac{{g'_u}^2}{Q'}\right)}, \label{eq:u_trec} 
\end{align}
where $g_u'$ and $Q'$ are scalar terms derived from $\Psibf'$, while $\ghat_u$ and $\Qhat$ are analogously derived from $\WH$; we denote by $\Psi'_u \in \rr$ and $\widehat{W}_u \in \rr$ the entries of $\Psibf'$ and $\WH$ corresponding to the upper series.
We report the complete expression of the reconciled distribution for \textit{\tRec} in~\ref{app:rec_minimal}, while for \textit{MinT} we refer to \citet{zambon2024properties}.

For both \textit{MinT} and \textit{\tRec}, the reconciled mean $\util$ is a convex combination of $\uhat$ and $\bhat_1+\bhat_2$.
However, there is an important  difference  in the prediction intervals. 
After \textit{MinT} reconciliation, the width of the prediction interval  always decreases, as the reconciled variance is smaller than the variance of the base forecasts \citep{zambon2024properties}.
In contrast, the prediction intervals of \textit{\tRec} can be wider than those of the base forecasts, since the $t$ distribution has generally heavier tails than the Gaussian.
Moreover, the scale of the reconciled distribution is proportional to the square root of the factor
\[
C   = \dfrac{1}{\nutil} \left(1 + \dfrac{(\uhat -\bhat_1 - \bhat_2)^2}{Q'}\right), 
\]
where $\uhat -\bhat_1 - \bhat_2$ and $Q'$ are respectively the mean and the variance of the random variable $\UH-\BH_1-\BH_2$, which represents the incoherence. 
A large \textit{scaled incoherence}, defined as $|\uhat -\bhat_1 - \bhat_2| \,/\, Q'^{1/2}$, 
results in a large value of $C$. Thus, the scale of the reconciled distribution (Eq.~\eqref{eq:u_trec}) could be higher than the scale of the base forecast. 
Hence, unlike \textit{MinT}, \textit{\tRec} 
widens the prediction intervals when faced with high incoherence.

\subsection{Simulations}
\label{sec:simulations}

We performed a simulation study in which
we  generated time series, computed the Gaussian base forecasts and reconciled them with \textit{MinT} and \textit{\tRec}. We considered the minimal hierarchy in Fig.~\ref{fig: min_hiearchy} and, as in \citet{wickramasuriya2019optimal}, we simulated the two bottom series using a structural time series model with trend, seasonality (which is set to 4), and correlated errors.
We obtained the upper series  as the sum of the bottom series. See~\ref{app:simulation} for more details on the data generating process. 
The base forecasts were purposely misspecified as we
computed them using ets  with additive noise from the R package \texttt{forecast} \citep{hyndman2008automatic}.
We introduce the \textit{relative width} of the prediction intervals (PIs) as:
\[
\text{relative width} = 
\frac{\text{width of the reconciled PI}}
{\text{width of the base forecast PI}}.
\]
In Fig.~\ref{fig:dist_PI_width_scatter}, we show  the distribution of the relative width  
(level 95\%) for \textit{MinT} and \textit{\tRec} with $T=12$. 
We report the results obtained with different values of $T$ in Fig.~\ref{fig:sim incoherence VS T}, \ref{app:simulation}.
As expected, the relative width of \textit{MinT} is consistently \textless  1 \citep{zambon2024properties}.
On the other hand, the relative width of \textit{\tRec} is sometimes above 1, and it is positively correlated to the scaled incoherence of the base forecasts.

\begin{figure}[h!]
    \centering
    \begin{overpic}[width=0.8\textwidth]{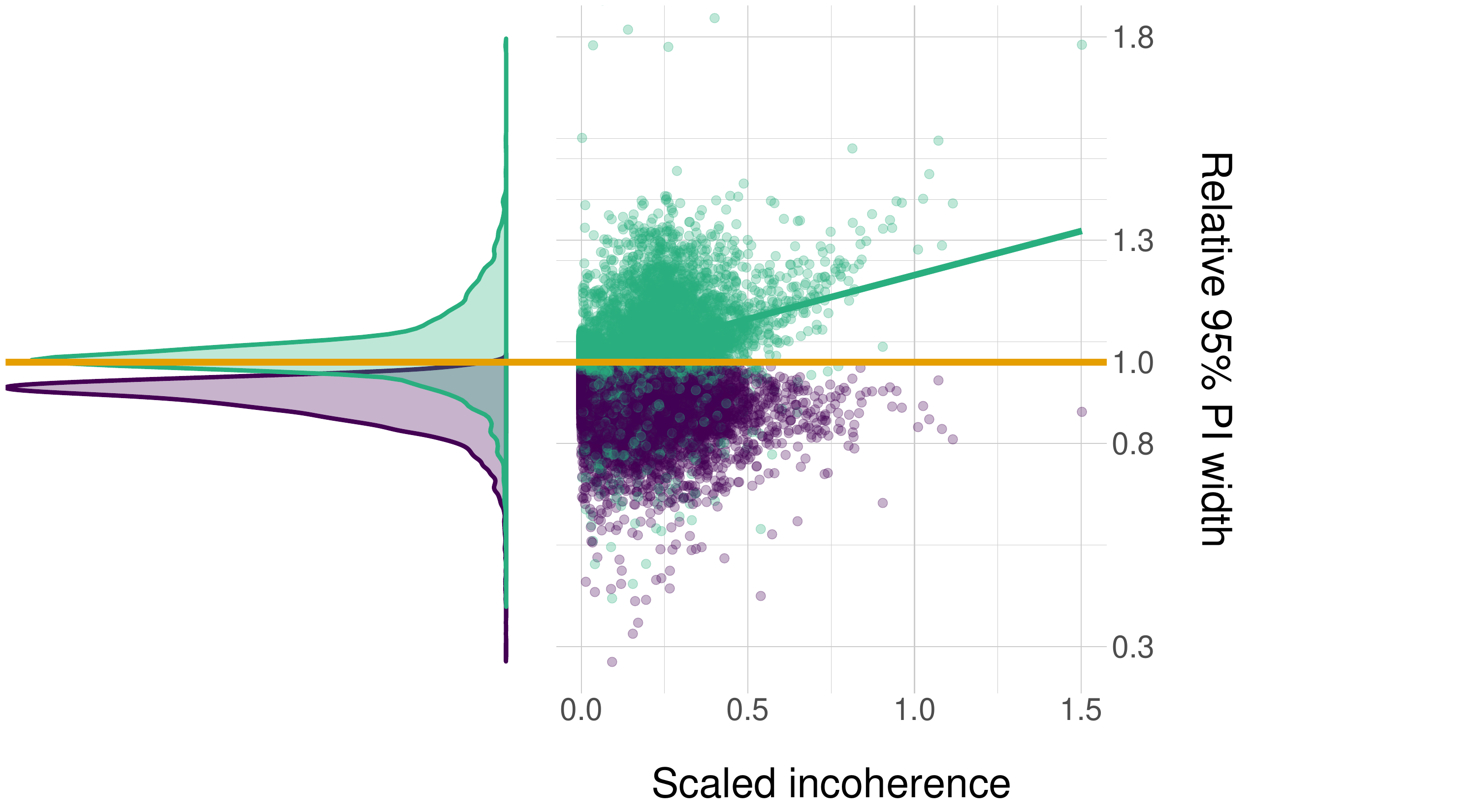}
            \put(5,10){\includegraphics[width=0.11\textwidth]{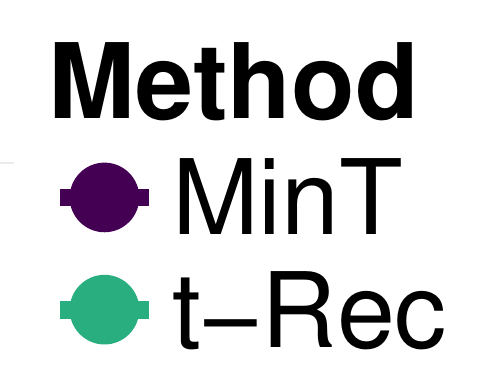}}
        \end{overpic}
    \caption{
    Left: distribution of the relative widths, i.e., reconciled PI width/base forecast PI width. We consider the 95\% prediction intervals for the upper series for \textit{MinT} (purple) and \textit{\tRec} (green), across 10,000 simulations with T=12.
    While \textit{MinT} always narrows the prediction intervals, \textit{\tRec} can also widen them.
    Right: for each simulation, the relative widths of \textit{MinT} and \textit{\tRec} are plotted against the scaled incoherence $|\uhat -\bhat_1 - \bhat_2| \,/\, \sqrt{Q'}$.
    Notably, \textit{\tRec} exhibits a positive correlation between interval widths and scaled incoherence.
    }
    \label{fig:dist_PI_width_scatter}
\end{figure}

In Table~\ref{tab:interval_width_simulations}, we report the geometric mean, over the simulations, of the relative width of the 80\% and 95\% PIs.
The relative width of \textit{MinT}
depends on the ratio of the standard deviations of the reconciled and the base forecasts; it is thus the 
same at any  confidence level.
Instead, the relative width of \textit{\tRec} varies with the confidence level. Due to the heavier tails of the t-distribution compared to the Gaussian, 
the relative width of
\textit{\tRec} is higher at the 95\% level than at the 80\%.

We also evaluate the accuracy of the forecasts obtained by both methods. 
The results, reported in \ref{app:simulation}, show that
\textit{\tRec} and \textit{MinT} perform similarly
on the point forecasts, but
\textit{\tRec}  outperforms \textit{MinT} on the probabilistic scores.

\begin{table}[h!]
\centering

\begin{tabular}{p{.7cm} cc @{\hskip 18pt} cc}
\toprule
& \multicolumn{4}{c}{\textbf{Relative width}}\\
 & \multicolumn{2}{c}{\textbf{80\%}} & \multicolumn{2}{c}{\textbf{95\%}} \\
{} & \textit{MinT} & \textit{\tRec} & \textit{MinT} & \textit{\tRec} \\
\cmidrule{2-5}
U  & 0.90 & 0.99 & 0.90 & 1.02 \\
B1 & 0.95 & 1.01 & 0.95 & 1.05 \\
B2 & 0.95 & 1.01 & 0.95 & 1.05 \\
\bottomrule
\end{tabular}
\caption{Geometric average over 10,000 experiments of the relative widths with $T = 12$.
}
\label{tab:interval_width_simulations}
\end{table}

Finally, we measure the effect of accounting for
the uncertainty on the covariance matrix $\Wbf$. 
We compare  \textit{\tRec} with
\textit{\tRec-MAP}, i.e., 
MinT equipped with a covariance matrix corresponding to the mode  of the Inverse-Wishart posterior distribution.
Thus \textit{\tRec-MAP} adopts the maximum a posteriori point estimate of $\Wbf$, while \textit{\tRec} accounts for the whole posterior distribution of $\Wbf$.
We run simulations with increasing training lengths $T$
(Fig.~\ref{fig:convergence}). 
For small values of $T$,
the PIs of \textit{\tRec} are wider than the PIs of \textit{\tRec-MAP}, because of the  uncertainty on $\Wbf$.
However, the posterior uncertainty on $\Wbf$ decreases with $T$ and eventually, for large $T$, the PIs of \textit{\tRec} and  \textit{\tRec-MAP} become equivalent.



\begin{figure}[h!]
    \centering
    \includegraphics[width=0.7\linewidth]{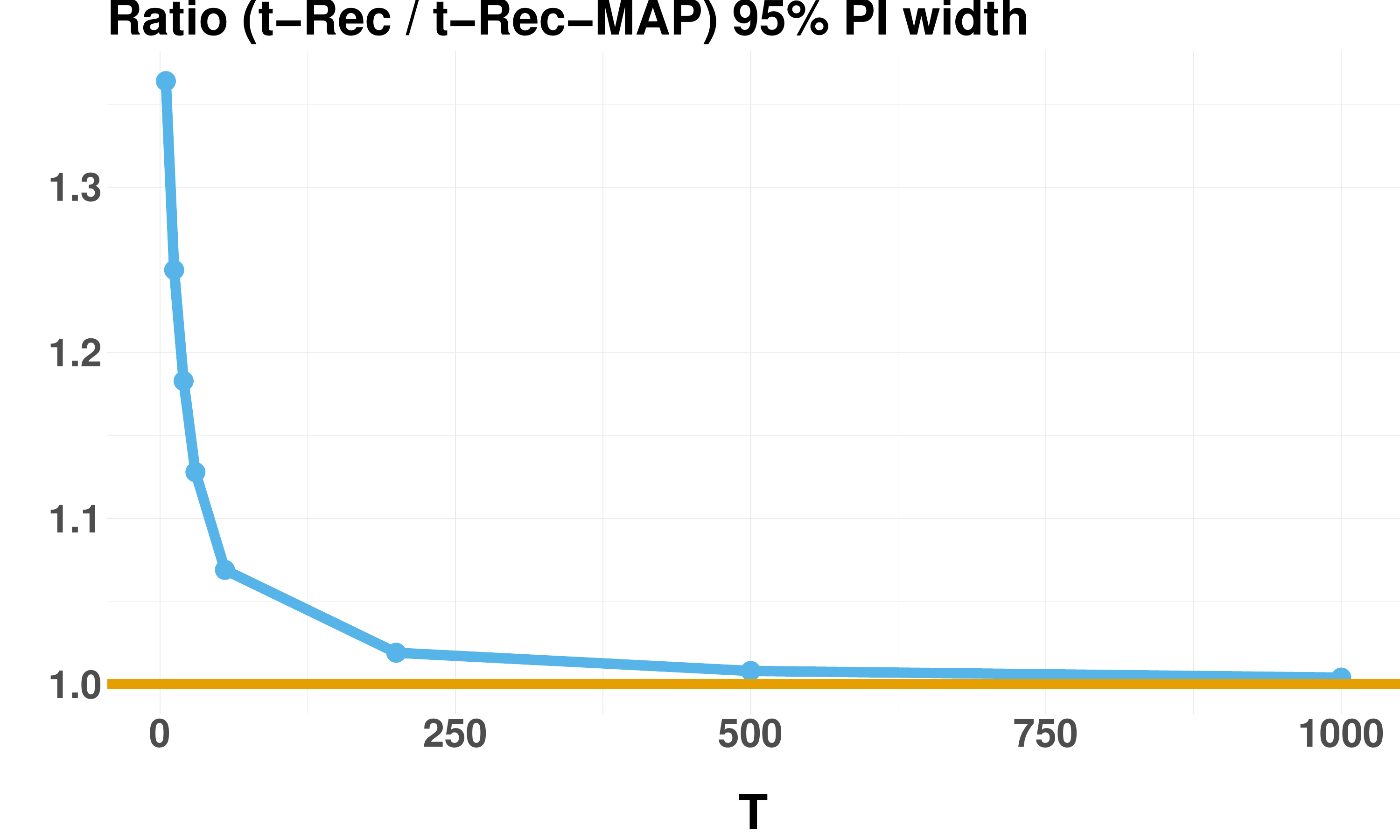}
    \caption{Geometric mean, computed over 10,000 independent experiments for each value of $T$, of the ratios between the 95\% prediction interval widths of \textit{\tRec} and \textit{\tRec-MAP}.
    The average ratio converges to 1 as the training length grows. 
    }
    \label{fig:convergence}
\end{figure}

\section{Experiments on real datasets}\label{sec: experiment}

We compare \textit{\tRec} and \textit{MinT} on three real-world datasets. The first is \textit{Swiss tourism}, which we  introduce  for the first time;  the other two
are taken from previous papers and regard Australian tourism.
We released the code to reproduce the results in a GitHub repository \footnote{\href{https://github.com/CarraraChiara/t-Rec}{Link to the GitHub repository}}.
The \textit{\tRec} method is implemented in the R packages \texttt{bayesRecon} \citep{bayesRecon} and \texttt{fable.bayesRecon} \citep{fablebayesRecon}.

\paragraph{Swiss tourism}
Switzerland is divided into 26 cantons, listed in Table~\ref{tab:cantons}, \ref{app:datasets}.
We consider monthly time series  of overnight stays, downloaded from the website of the Swiss Federal Statistical Office  \footnote{\href{https://www.pxweb.bfs.admin.ch/pxweb/en/px-x-1003020000_102/-/px-x-1003020000_102.px/table/tableViewLayout2/}{Swiss tourism data}}. 
The data are originally disaggregated by both canton and country of origin of the tourists.
However, 
we restrict our analysis to the canton level,
as further disaggregation yields time series with low counts and many zeros, for which 
the assumption of  Gaussianity of the residuals
is untenable.
The resulting hierarchy consists of 26 bottom-level series (one for each canton) and one top-level series (Swiss total). 
Each time series goes from January 2005 to January 2025, i.e.,
241 monthly observations.
We made the Swiss tourism dataset available in the R package \texttt{bayesRecon} \citep{bayesRecon}.

\paragraph{Australian tourism}
The other two datasets (\textit{Australian Tourism-M} and \textit{Australian Tourism-Q}) refer to Australian domestic overnight trips.
We obtained the
monthly dataset
(\textit{Australian Tourism-M})  from
Rob J. Hyndman's website \footnote{\href{https://robjhyndman.com/data/TourismData\_v3.csv}{https://robjhyndman.com/data/TourismData\_v3.csv}}.
The bottom level includes 76 series representing Australian regions, which are then aggregated into 21 zones, 7 states, and the total. This dataset spans from January 1998 to December 2016, covering 228 time points. For more details, see 
Table~6 in \cite{wickramasuriya2019optimal}. 

Some regions are, however, composed of a single zone.
In \cite{wickramasuriya2019optimal} they are in practice counted twice, both as region and as zone.
We instead do not include zones composed by a single region;
such regions are aggregated at the state level.

We take the quarterly dataset
\textit{Australian Tourism-Q} from the \texttt{tsibble} R package \citep{wang2020new}. 
This is a  two-level hierarchy, similar to the previous dataset but without the zone level. 
This dataset covers 20 years, from January 1998 to December 2017, for a total of 80 observations.
The main features of all datasets are summarized in Table~\ref{tab:dataset_summary}.

\begin{table}[ht]
\centering
\begin{tabular}{lcccccc}
\toprule
\textbf{Dataset} & Frequency & Levels & $n_b$ & $n_u$ &$n$ & Length \\
\midrule
Swiss Tourism       & Monthly & 1 & 26 & 1  & 27  & 241 \\
Australian Tourism-M & Monthly & 3 & 76 & 29 & 105 & 228 \\
Australian Tourism-Q & Quarterly & 2 & 76 & 8 & 84 & 80 \\
\bottomrule
\end{tabular}
\caption{Summary of the datasets.}
\label{tab:dataset_summary}
\end{table}

\paragraph{Experimental setting}
We compute the base forecasts for each time series using ets  from the \texttt{forecast} R package \cite{hyndman2008automatic}.
We only consider models with additive noise, as 
our reconciliation approaches assume a Gaussian
predictive distribution.
We perform reconciliation with both \textit{\tRec} and \textit{MinT}.

We evaluate performance across various \trains, using a rolling origin scheme for each.
We use 100 rolling origins for the monthly datasets (\textit{Swiss tourism} and \textit{Australian Tourism-M}).
For \textit{Australian Tourism-Q}, which has fewer observations, we use between 30 and 65 rolling origins, depending on the \train.

\paragraph{Evaluation metrics}
\label{subsec: eval_metrics}

We evaluate both point and probabilistic forecast accuracy. 
For the latter, to ensure a comprehensive evaluation, we employ univariate scores to measure marginal performance and multivariate scores to assess the joint predictive distribution.

We assess the point forecasts through the Mean Squared Error (MSE), averaged over the rolling origins: 
\begin{equation*}
    \text{MSE} = \frac{1}{R} \sum_{i = 1}^{R} \big\|\yT_i - \ybf_i\big\|_2^2,
\end{equation*}
where $||\cdot||_2$ is the Euclidean distance, $R$ is the number of rolling origins, $\yT_i \in \rr^n$ is the vector of the point forecasts for the $i$th rolling origin, and $\ybf_i \in \rr^n$ is the vector of the corresponding actual values. 
We report the relative MSE, i.e., the ratio between the MSE of the reconciled and the base forecasts:
\begin{equation}
    \text{RelMSE} = \frac{\text{MSE}}{\text{MSE}_{base}}. \label{eq: relMSE}
\end{equation}
For each series $j$, with $j=1,\dots,n$, we assess the predictive distribution using the Continuous Ranked Probability Score (CRPS)
\citep{gneiting2007strictly}:
\begin{equation*}
    \text{CRPS}_j = \frac{1}{R} \sum_{i = 1}^{R} \int_{-\infty}^{+\infty} \big(F_{i,j}(x) - \mathbb{I}(x \geq y_{i,j})\big)^2 \,dx,
\end{equation*}
where $F_{i,j}$ is the predictive CDF for the $i$th rolling origin and series $j$ and $y_{i,j}$ is the corresponding actual value.
We compute the CRPS for both the Gaussian and t-distributions in closed form using the \texttt{scoringRules} R package \citep{jordan_ScoringRules}.
We score the prediction intervals of each series using the Mean Interval Score (MIS) \citep{gneiting2007strictly}: 
\begin{equation*}
        \text{MIS}_j = \frac{1}{R} \sum_{i = 1}^{R} \left( (u^{\alpha}_{i,j} - l^{\alpha}_{i,j}) + \frac{2}{\alpha} \left[ (l^{\alpha}_{i,j} - y_{i,j}) \mathbb{I}(y_{i,j} < l^{\alpha}_{i,j}) + (y_{i,j} - u^{\alpha}_{i,j}) \mathbb{I}(y_{i,j} > u^{\alpha}_{i,j}) \right] \right),
\end{equation*}
where $l^{\alpha}_{i,j}$ and $u^{\alpha}_{i,j}$ are the lower and upper bounds of the $\alpha$-level prediction interval for the $i$th rolling origin and series $j$, and $y_{i,j}$ is the corresponding actual value.
We consider the $80\%$ and $95\%$ prediction intervals. 
The CRPS and MIS are univariate scale-dependent scoring rules; following \cite{girolimetto2024cross}, we aggregate the $n$ scores into a unique one by taking the geometric mean of the relative scores with respect to the base forecasts:
\begin{equation}
    RelCRPS = \left(\prod_{j = 1}^n \frac{CRPS_j}{CRPS_{j,base}}\right)^{\frac{1}{n}}, \label{eq: relCRPS}
\end{equation}
and we do the same with the relative MISs for the $80\%$ and $95\%$ prediction intervals.

\begin{figure}[b!]
    \centering

    \begin{subfigure}[b]{0.42\textwidth}
        \centering
        \begin{overpic}[width=\linewidth]{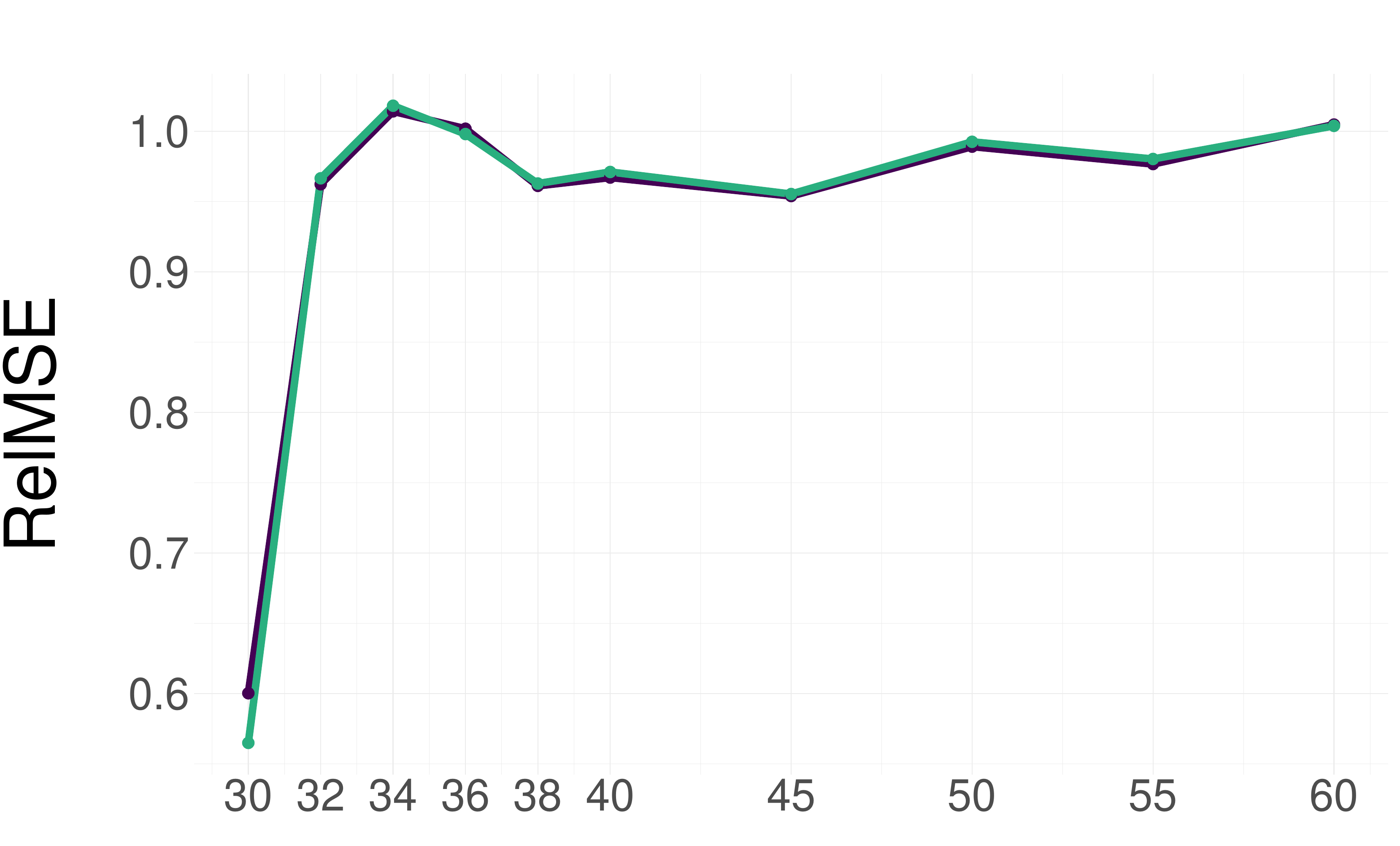}
            \put(80,10){\includegraphics[width=0.17\textwidth]{legend.pdf}}
        \end{overpic}
    \end{subfigure}
\hfill
    \begin{subfigure}[b]{0.42\textwidth}
        \centering
        \begin{overpic}[width=\linewidth]{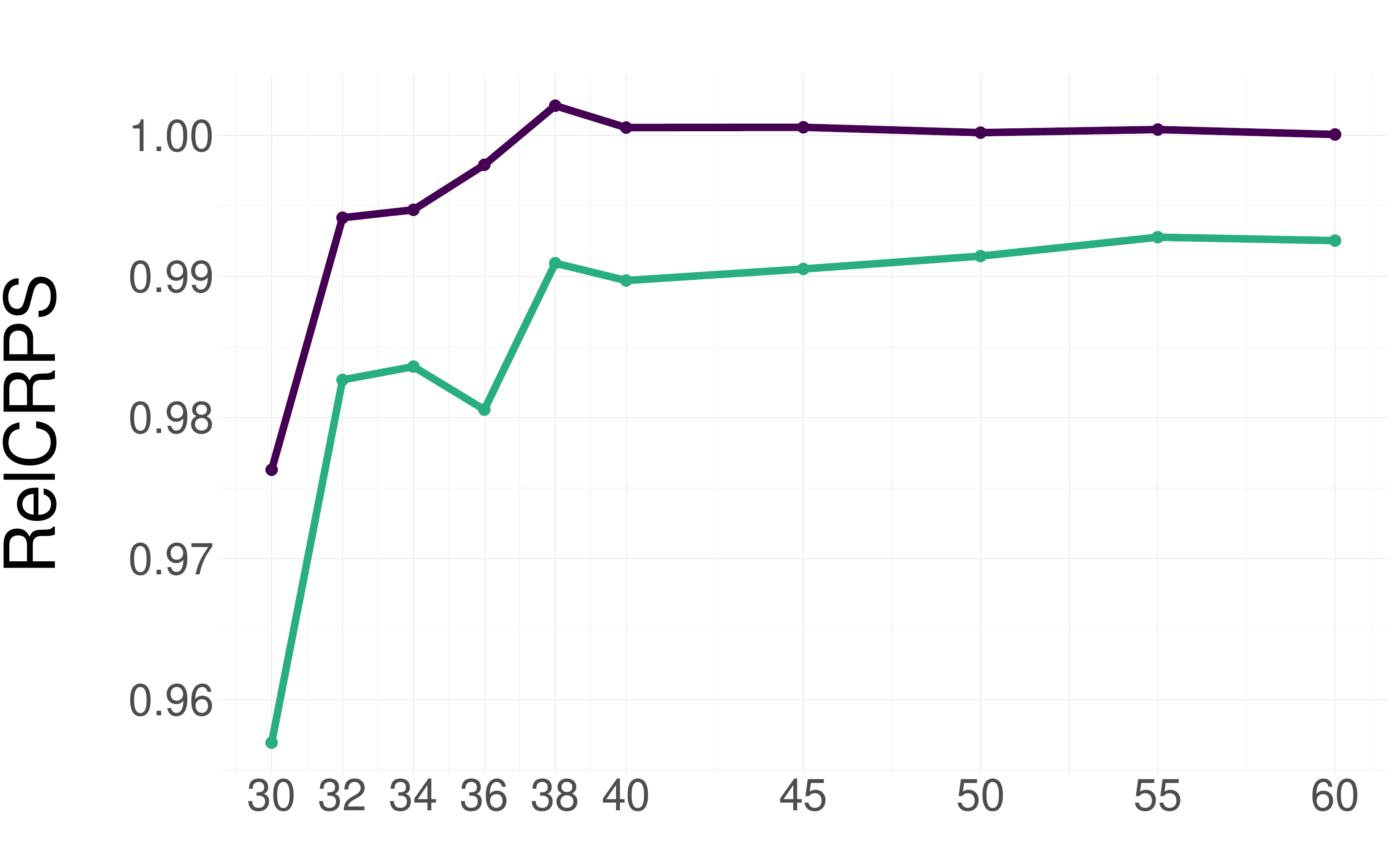}
            \put(80,10){\includegraphics[width=0.17\textwidth]{legend.pdf}}
        \end{overpic}
    \end{subfigure}

    \begin{subfigure}[b]{0.42\textwidth}
        \centering
        \begin{overpic}[width=\linewidth]{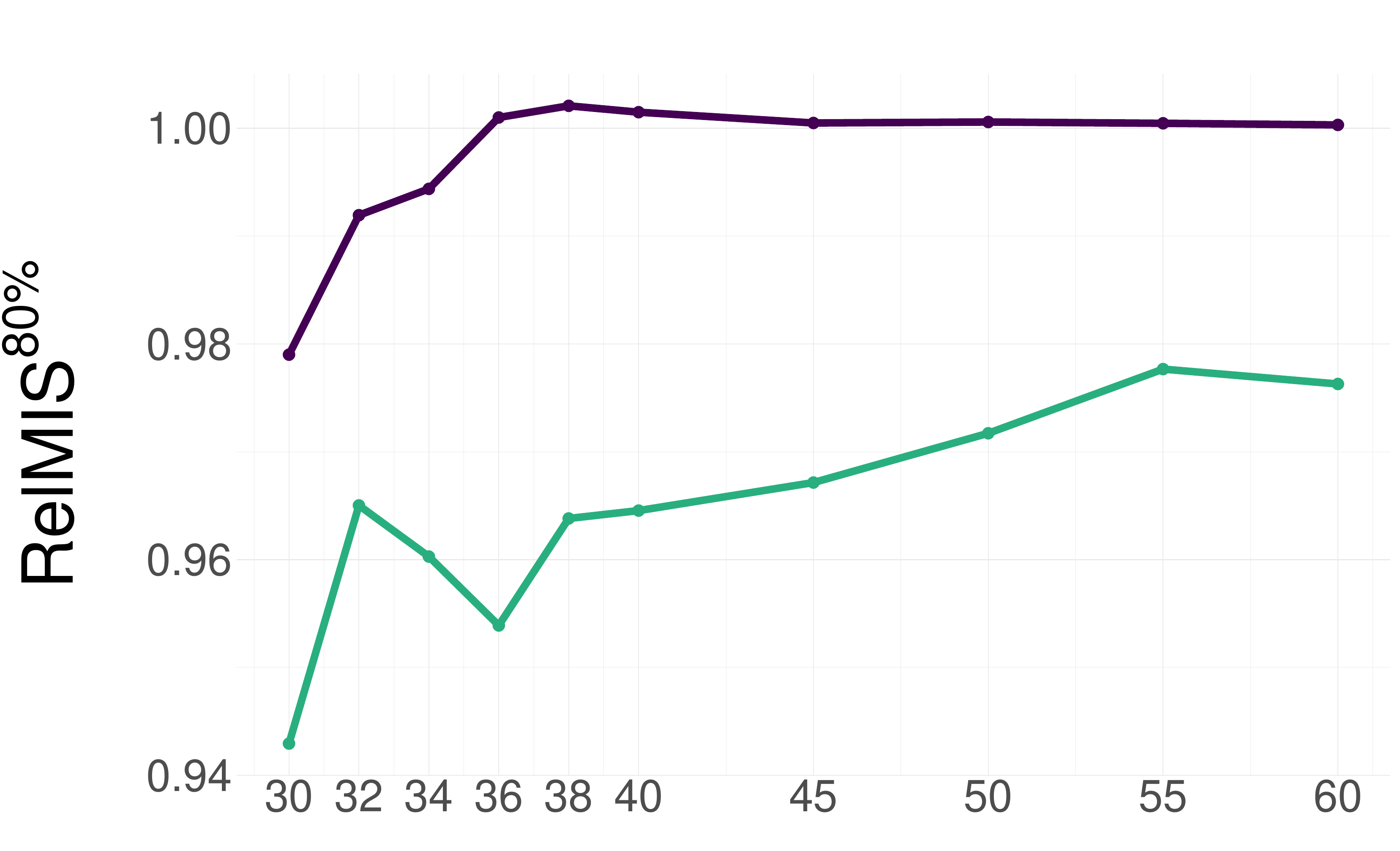}
            \put(80,15){\includegraphics[width=0.17\textwidth]{legend.pdf}}
        \end{overpic}
    \end{subfigure}
\hfill
    \begin{subfigure}[b]{0.42\textwidth}
        \centering
        \begin{overpic}[width=\linewidth]{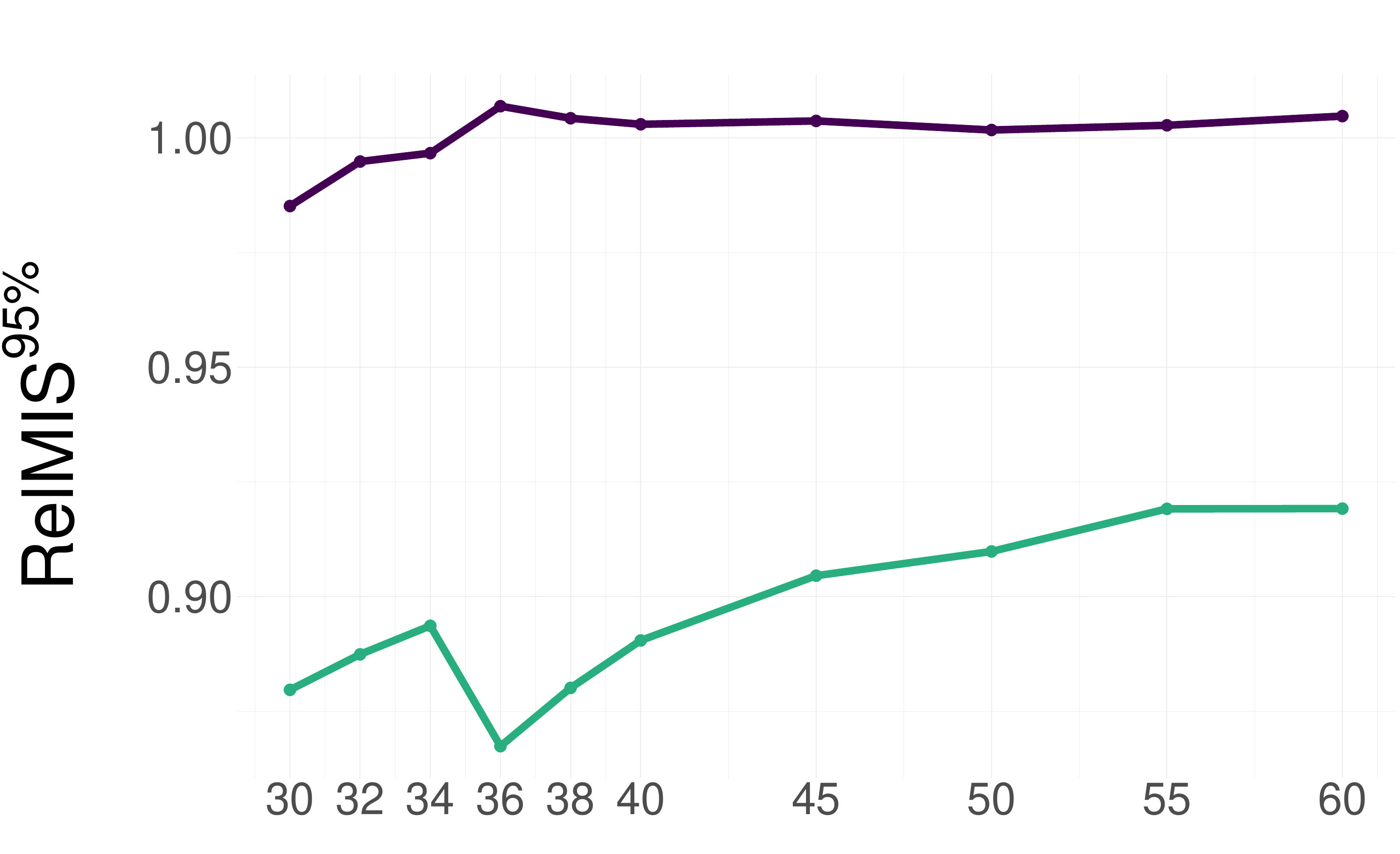}
            \put(80,12){\includegraphics[width=0.17\textwidth]{legend.pdf}}
        \end{overpic}
    \end{subfigure}

    \begin{subfigure}[b]{0.42\textwidth}
        \centering
        \begin{overpic}[width=\linewidth]{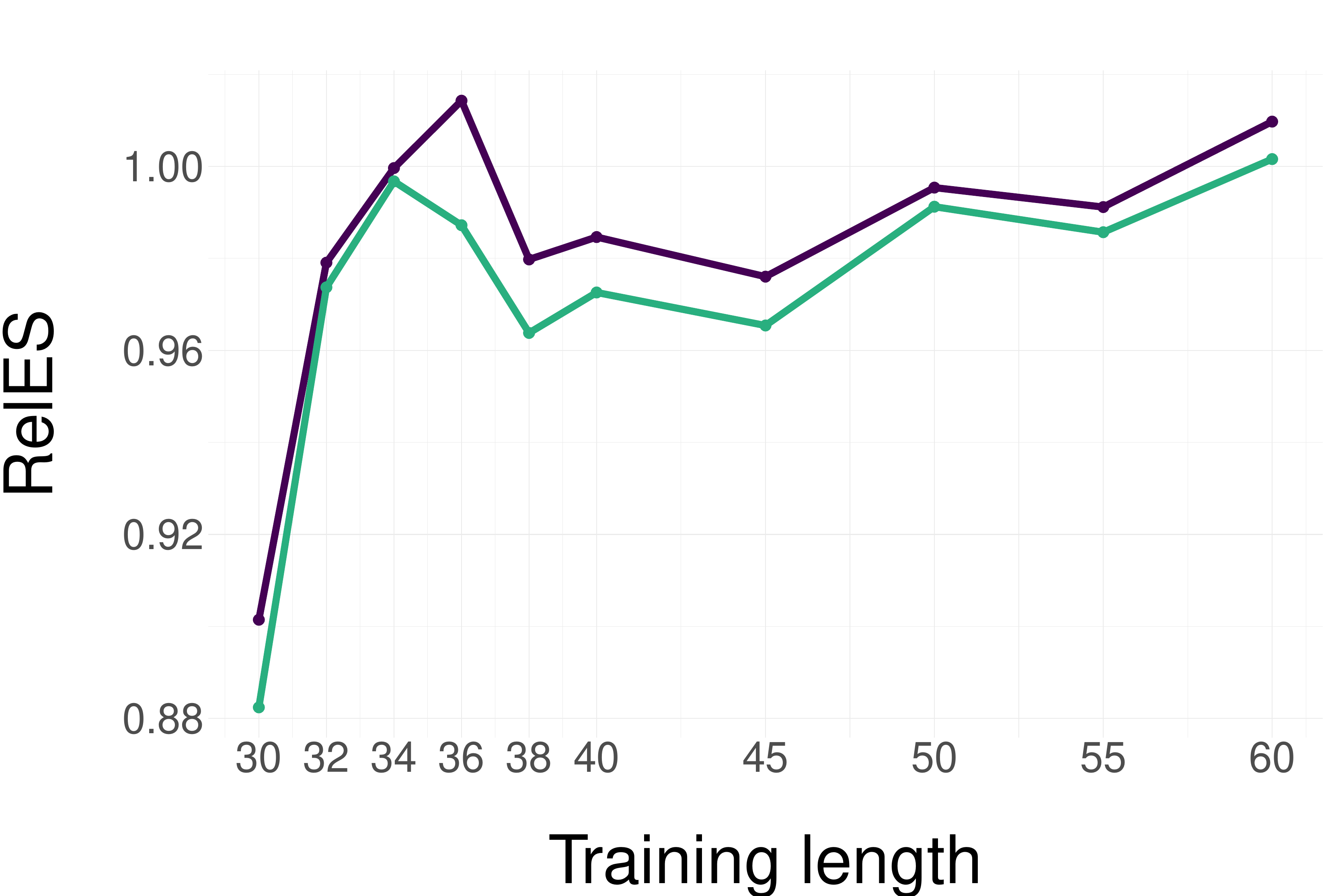}
             \put(80,18){\includegraphics[width=0.17\textwidth]{legend.pdf}}
        \end{overpic}
    \end{subfigure}

    \vspace{0.2cm}

    \caption{
    Results on the \textit{Swiss Tourism} dataset for \textit{MinT} (purple) and \textit{\tRec} (green),
    for the relative MSE (top-left), CRPS (top-right), MIS (middle-left: $80\%$, middle-right: $95\%$), and ES (bottom).
    A relative score $< 1$ means improvement over the base forecasts.}
    \label{fig:swiss score}
\end{figure}

We assess the joint predictive distribution using the energy score (ES) \citep{gneiting2007strictly}:
\begin{equation} \label{eq:def_ES}
\text{ES}_i = \E\|\Xbf_i - \ybf_i\|^\alpha - 
\frac{1}{2} \E\|\Xbf_i - \Xbf_i'\|^\alpha,
\end{equation}
where 
$\Xbf_i,\, \Xbf_i'$ are independent random vectors distributed as the joint forecast distribution for the $i$th rolling origin,
$\ybf_i \in \rr^n$ is the vector of the corresponding actual values,
and $\alpha \in (0,2)$; we set $\alpha = 1$, as it is common convention. 
The ES is a multivariate generalization of the CRPS \citep{gneiting2008assessing}.
We compute the ES by approximating the expectations in Eq.~\eqref{eq:def_ES} via sampling.
We then compute the relative score with respect to the base forecasts:
\begin{equation}\label{eq:rel_ES}
RelES = \frac{\sum_{i=1}^R\text{ES}_i}{\sum_{i=1}^R\text{ES}_{i,base}}.
\end{equation}

Finally, we check the statistical significance of the differences between methods.
For each univariate score (CRPS, MIS, and the univariate version of the MSE), we rank the models on each time series.
We then use the  Friedman test \citep{demvsar2006statistical} followed by the  Nemenyi post-hoc  to detect which models have a mean rank
significantly worse than the model with the best mean rank.
This procedure controls for multiple comparisons and it is also known as Multiple Comparison with the Best (MCB) (\citet{Koning_Franses_Hibon_Stekler_2005}). We use the implementation of the
\texttt{tsutils} R package \citep{tsutils}.

\begin{figure}[h!]
  \centering
  \begin{minipage}{0.35\textwidth}
    \resizebox{\textwidth}{!}{
      \includegraphics{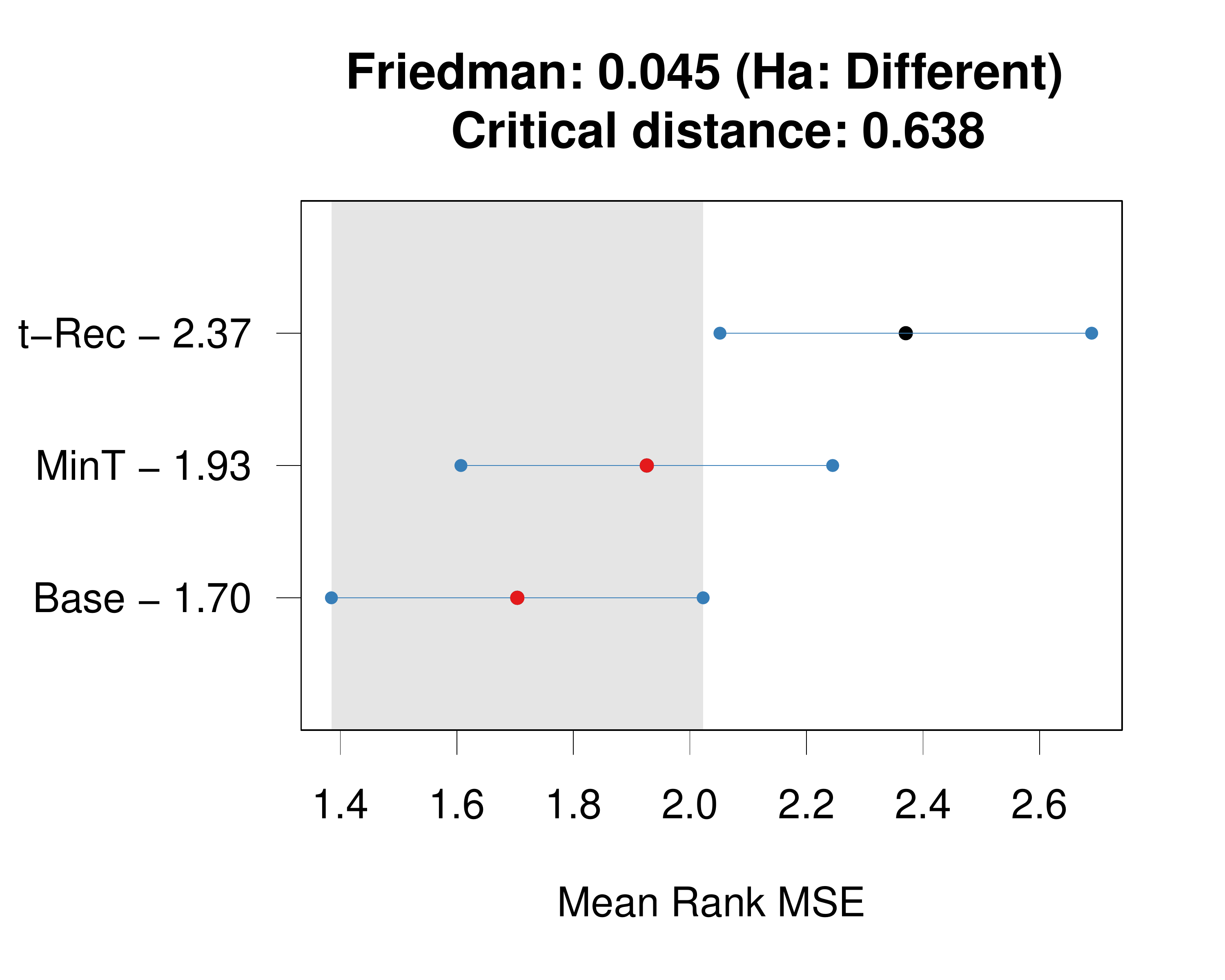}
    }
  \end{minipage}%
  \begin{minipage}{0.35\textwidth} 
    \resizebox{\textwidth}{!}{
      \includegraphics{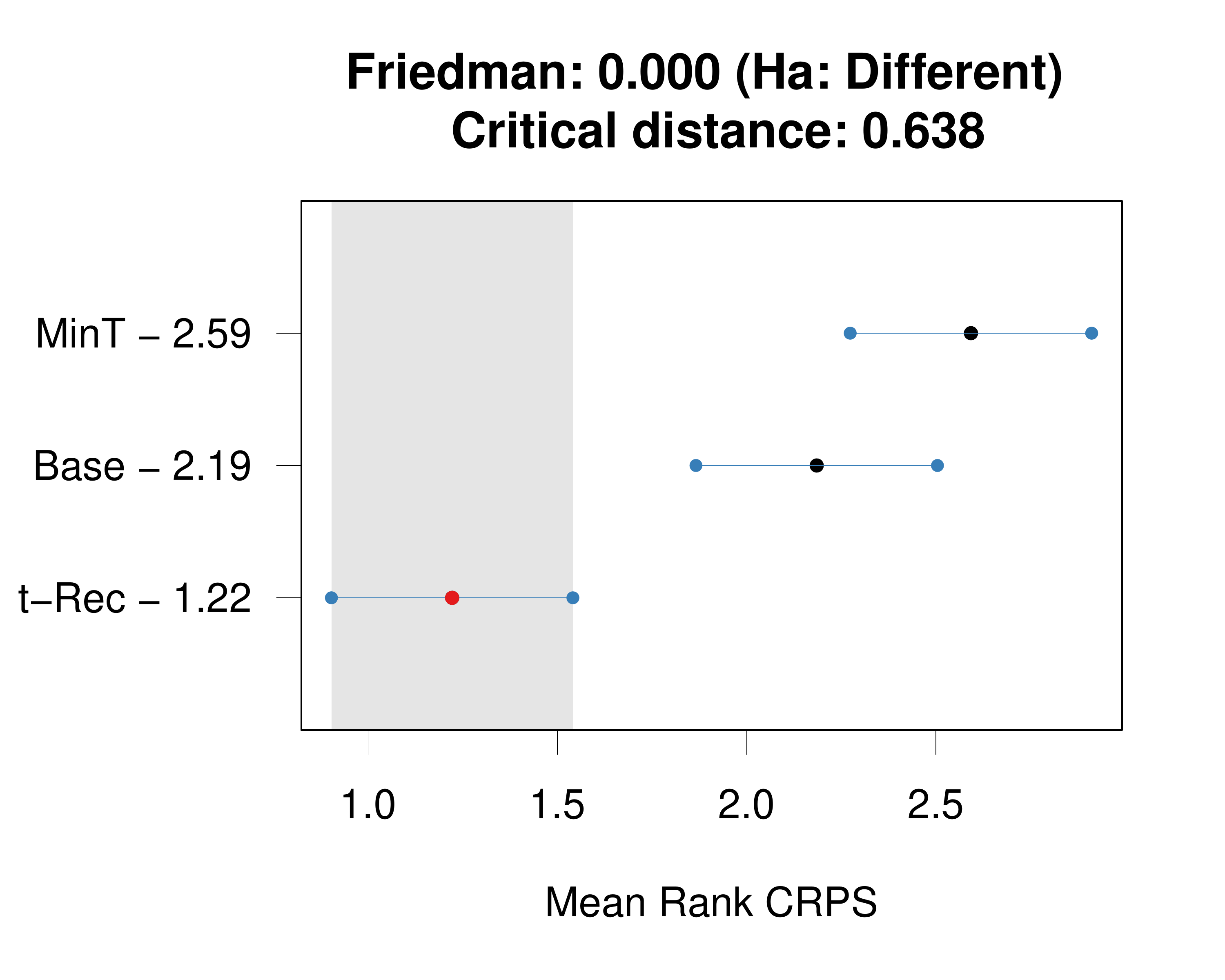}
    }
  \end{minipage}

  \begin{minipage}{0.35\textwidth}
    \resizebox{\textwidth}{!}{
      \includegraphics{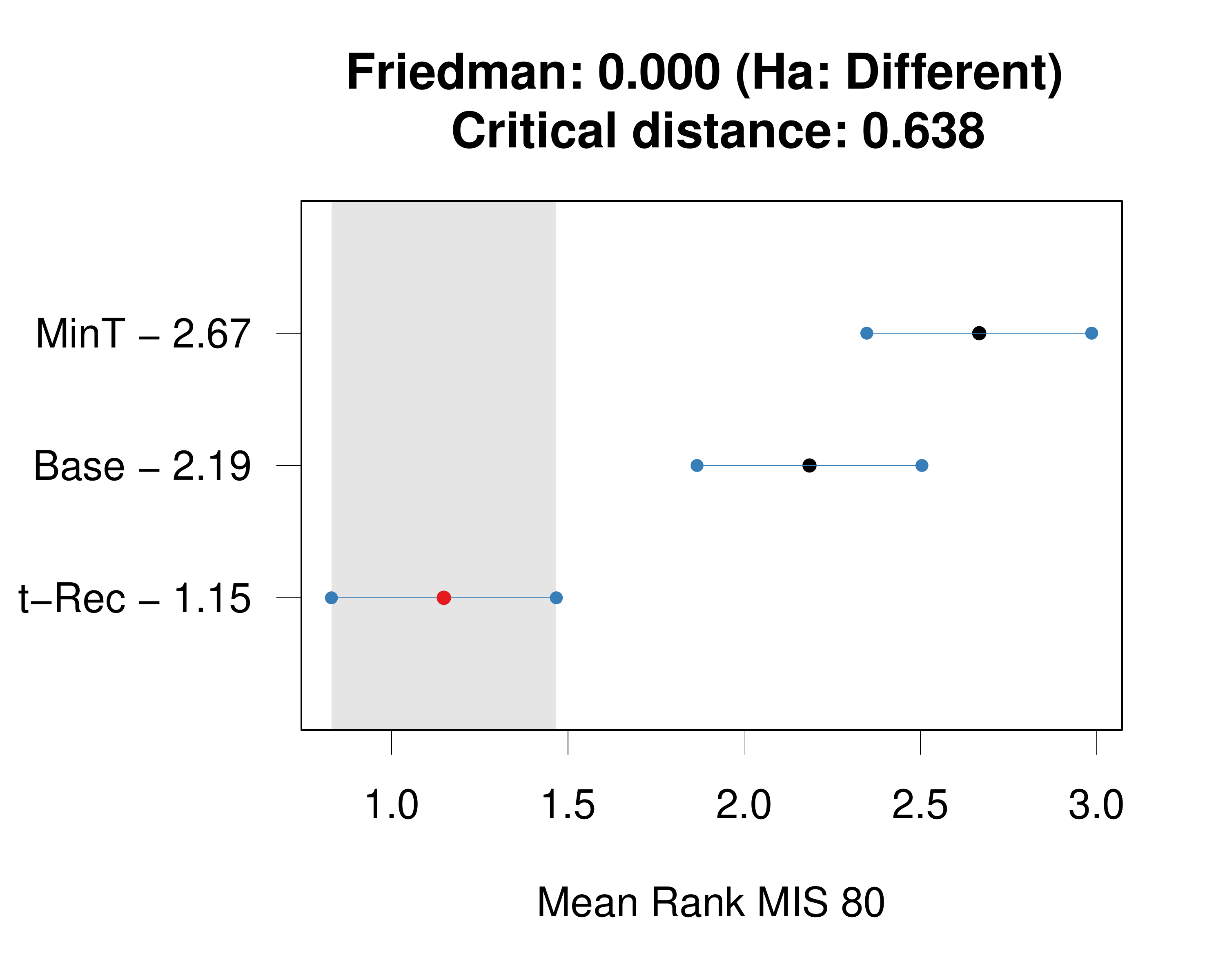}
    }
  \end{minipage}%
  \begin{minipage}{0.35\textwidth} 
    \resizebox{\textwidth}{!}{
      \includegraphics{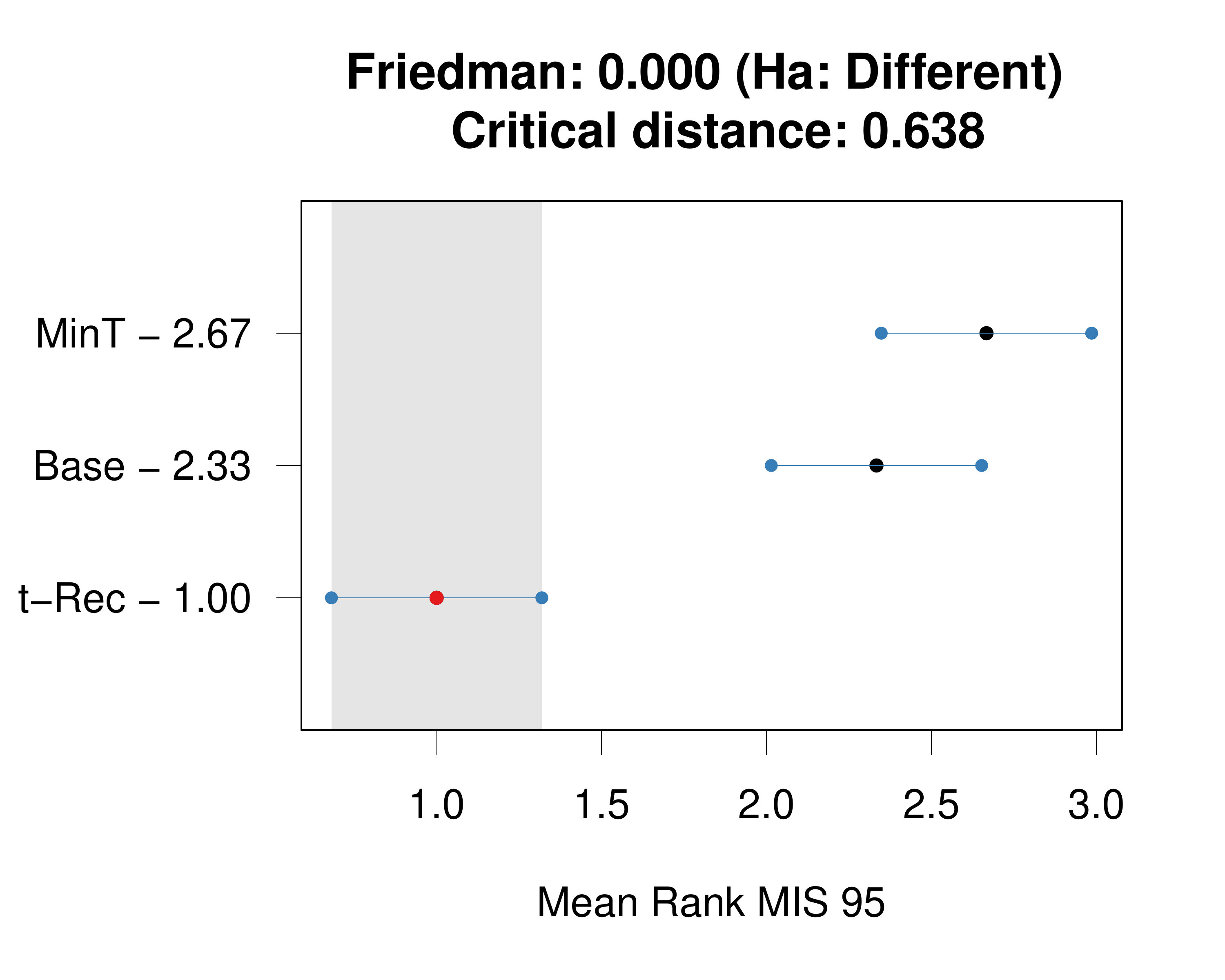}
    }
  \end{minipage}
  
 \caption{Friedman and post-hoc Nemenyi tests results for \train\ of 40 in the \textit{Swiss Tourism} dataset
 for the MSE (top-left), CRPS (top-right), and MIS (bottom-left: $80\%$, bottom-right: $95\%$).
 }
 \label{fig:swiss_MCB}
\end{figure}

\subsection{Swiss tourism results}\label{sec:swiss_tourism}

We consider \trains\ from 30 (two and a half years) to 60 (five years). The results are shown in  Fig.~\ref{fig:swiss score} as a function of 
the \train.


On the MSE, \textit{\tRec} and \textit{MinT} perform similarly across all \trains\ (top-left plot).
However, \textit{\tRec} shows a clear improvement in the predictive distribution, according to both univariate (CRPS, top-right, and MIS, middle) and multivariate scores (ES, bottom).
Since both the CRPS and ES average the quantile loss across all quantiles \citep[Chap.~5.9]{hyndman2021forecasting} and \textit{\tRec} particularly enhances the tails of the predictive distributions, the improvement is more pronounced for the MIS, especially for the $95\%$ intervals.
Indeed, the t-distribution has polynomial tails, which allocate more probability mass to the extremes compared to the exponential tails of the Gaussian distribution used by \textit{MinT}.
Moreover, \textit{\tRec} provides better coverage than \textit{MinT}; the average coverage of the $95\%$ prediction intervals is 88\% for \textit{\tRec} and 83\% for \textit{MinT}.
The improvement is consistent across all series, as shown in Fig.~\ref{fig:cov_swiss}, \ref{app:datasets}.

Fig.~\ref{fig:swiss_MCB} reports the results of the Friedman test with post-hoc for a \train\ of 40 observations. The titles report the p-value of the Friedman tests, the red dots represent the mean rank of the best method whose value is also reported on the y-axis label.
Methods that are outside of the shaded areas perform significantly worse than the best method.
Fig.~\ref{fig:MIS95_maps} reports the relative MIS for the 95\% prediction intervals at the canton level for a \train\ of 40 observations. These results further confirm that \textit{\tRec} yields higher-scoring prediction intervals than \textit{MinT}.

\begin{figure}[!ht]
    \begin{minipage}{0.46\textwidth}
    \centering
    \includegraphics[width=\linewidth]{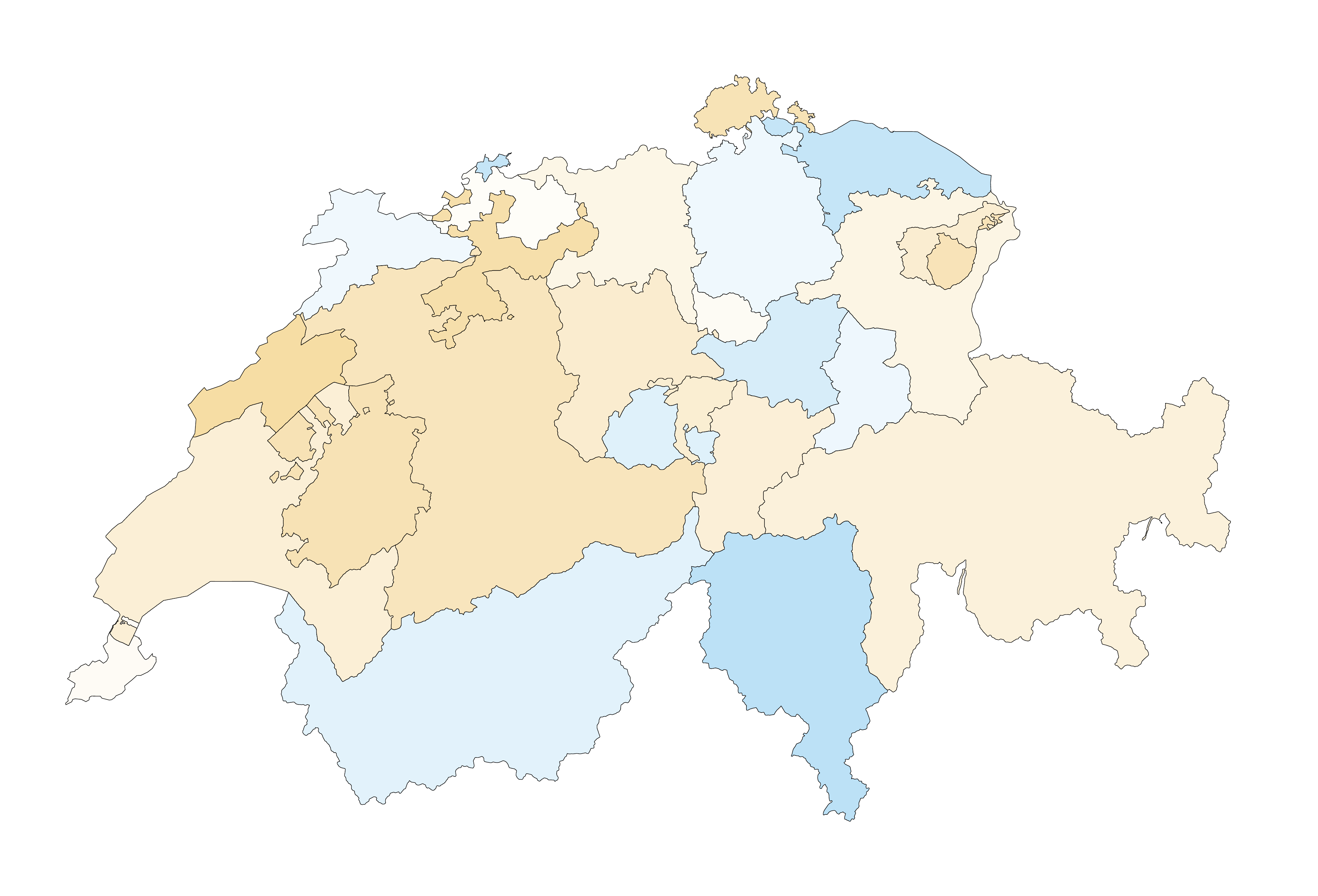}
\end{minipage}
\hfill
\begin{minipage}{0.49\textwidth}
    \centering
    \includegraphics[width=\linewidth]{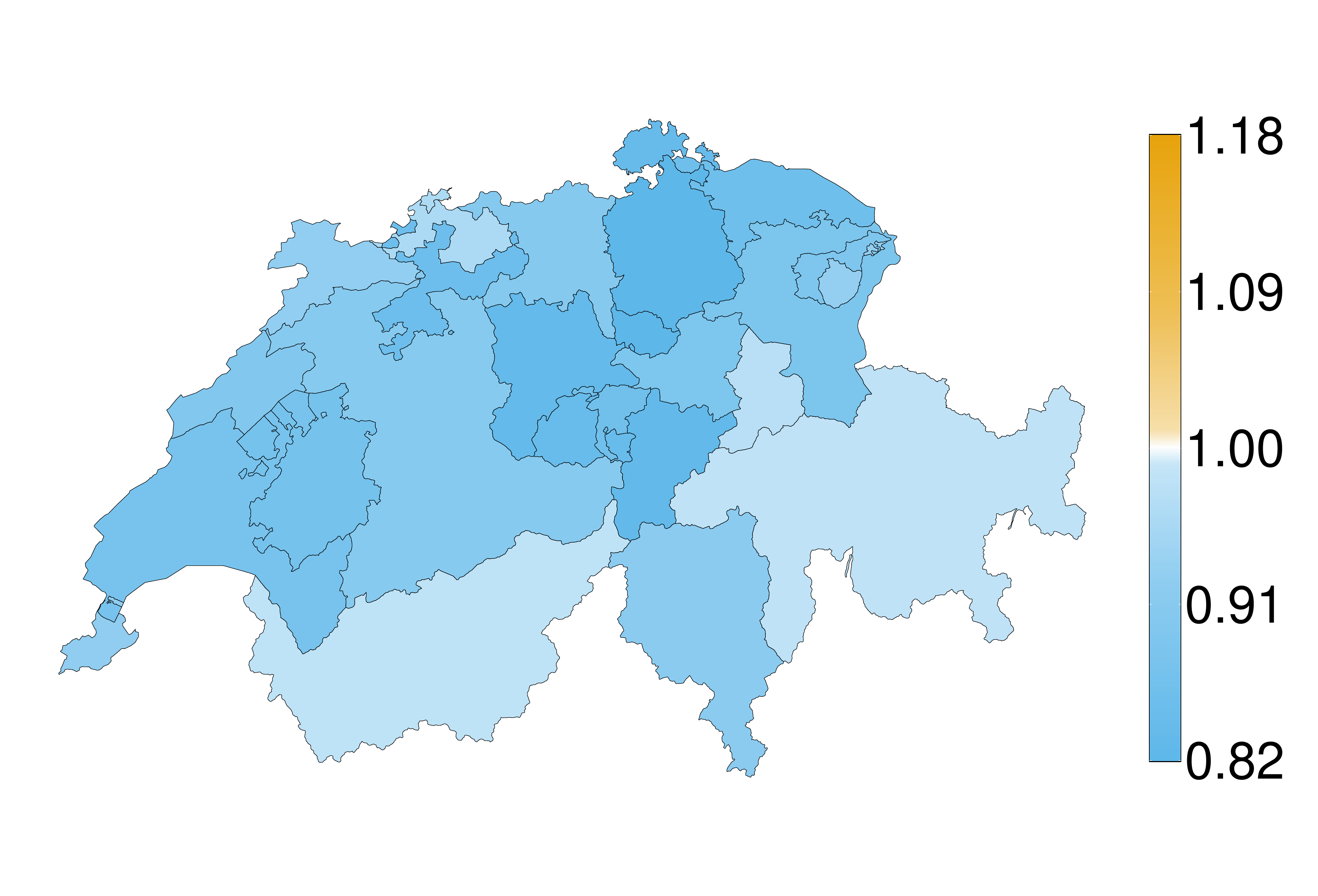}
\end{minipage}
    \caption{Relative MIS$^{95\%}$ of \textit{MinT} (left) and \textit{t‑Rec} (right) for each canton of the \textit{Swiss Tourism} dataset (\train $= 40$).
    Orange-colored cantons indicate worse performance compared to the base forecasts (ratio $> 1$). \textit{t‑Rec} consistently improves MIS$^{95\%}$ across all cantons and achieves greater absolute gains compared to \textit{MinT}.}
    \label{fig:MIS95_maps}
\end{figure}

\begin{figure}[H]
    \centering

    \makebox[\textwidth]{%
        \parbox[b]{0.42\textwidth}{\centering \textbf{\textit{Australian Tourism-M}}}%
        \hfill
        \parbox[b]{0.42\textwidth}{\centering \textbf{\textit{Australian Tourism-Q}}}%
    }
    
    \begin{subfigure}[b]{0.42\textwidth}
        \centering
        \begin{overpic}[width=\linewidth]{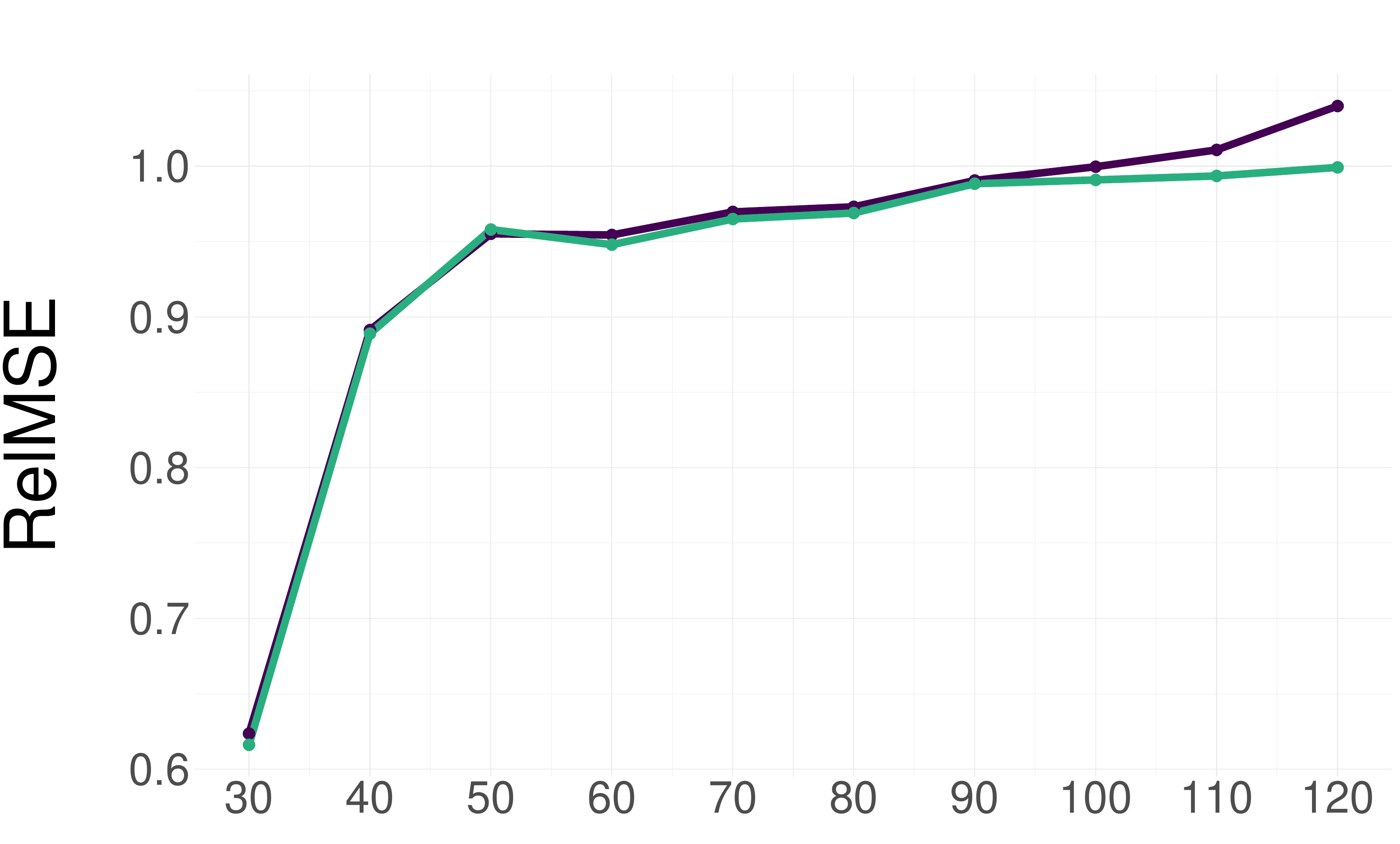}
            \put(80,10){\includegraphics[width=0.17\textwidth]{legend.pdf}}
        \end{overpic}
    \end{subfigure}
    \hfill
    \begin{subfigure}[b]{0.42\textwidth}
        \centering
        \begin{overpic}[width=\linewidth]{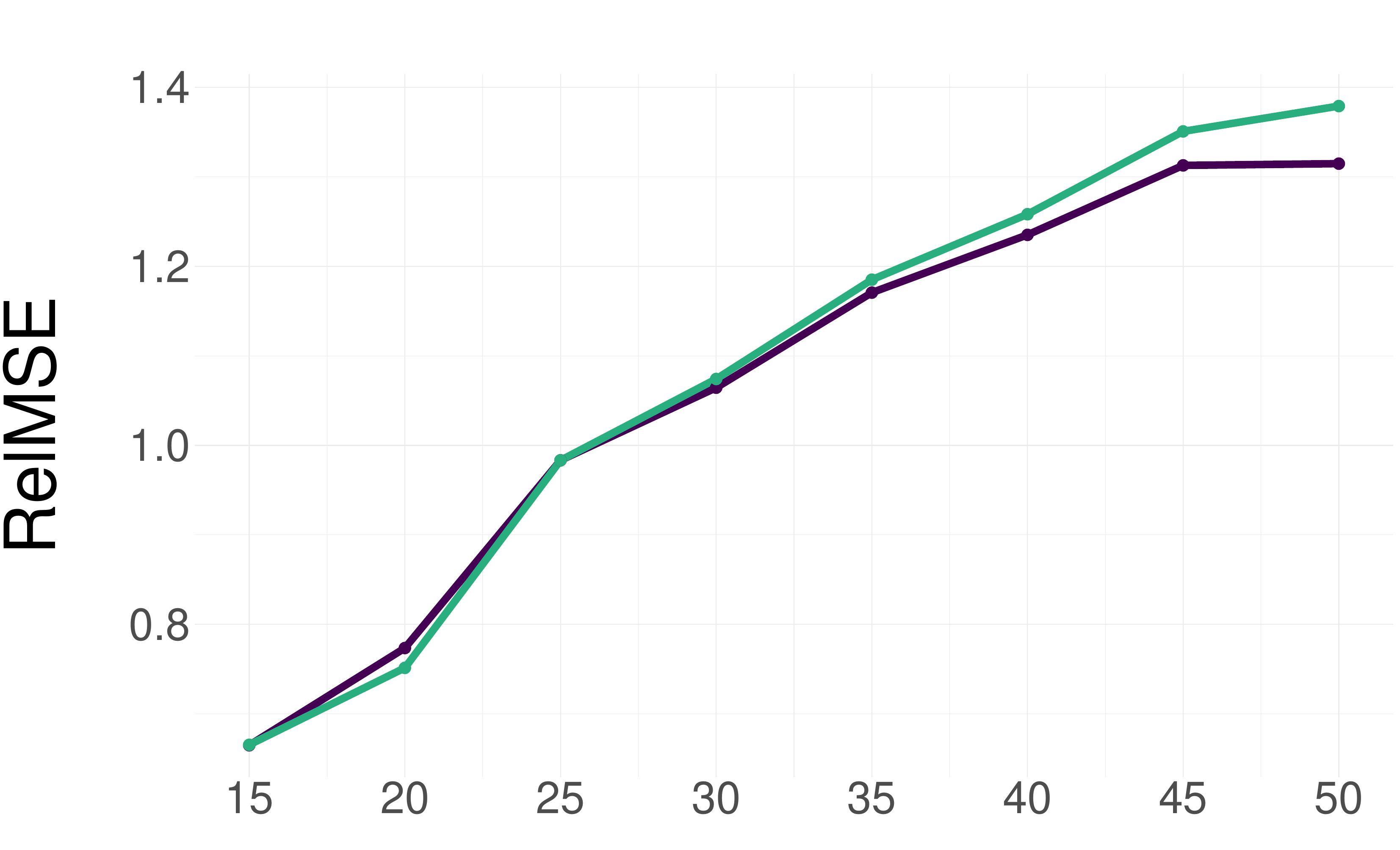}
            \put(80,10){\includegraphics[width=0.17\textwidth]{legend.pdf}}
        \end{overpic}
    \end{subfigure}
    
    \begin{subfigure}[b]{0.42\textwidth}
        \centering
        \begin{overpic}[width=\linewidth]{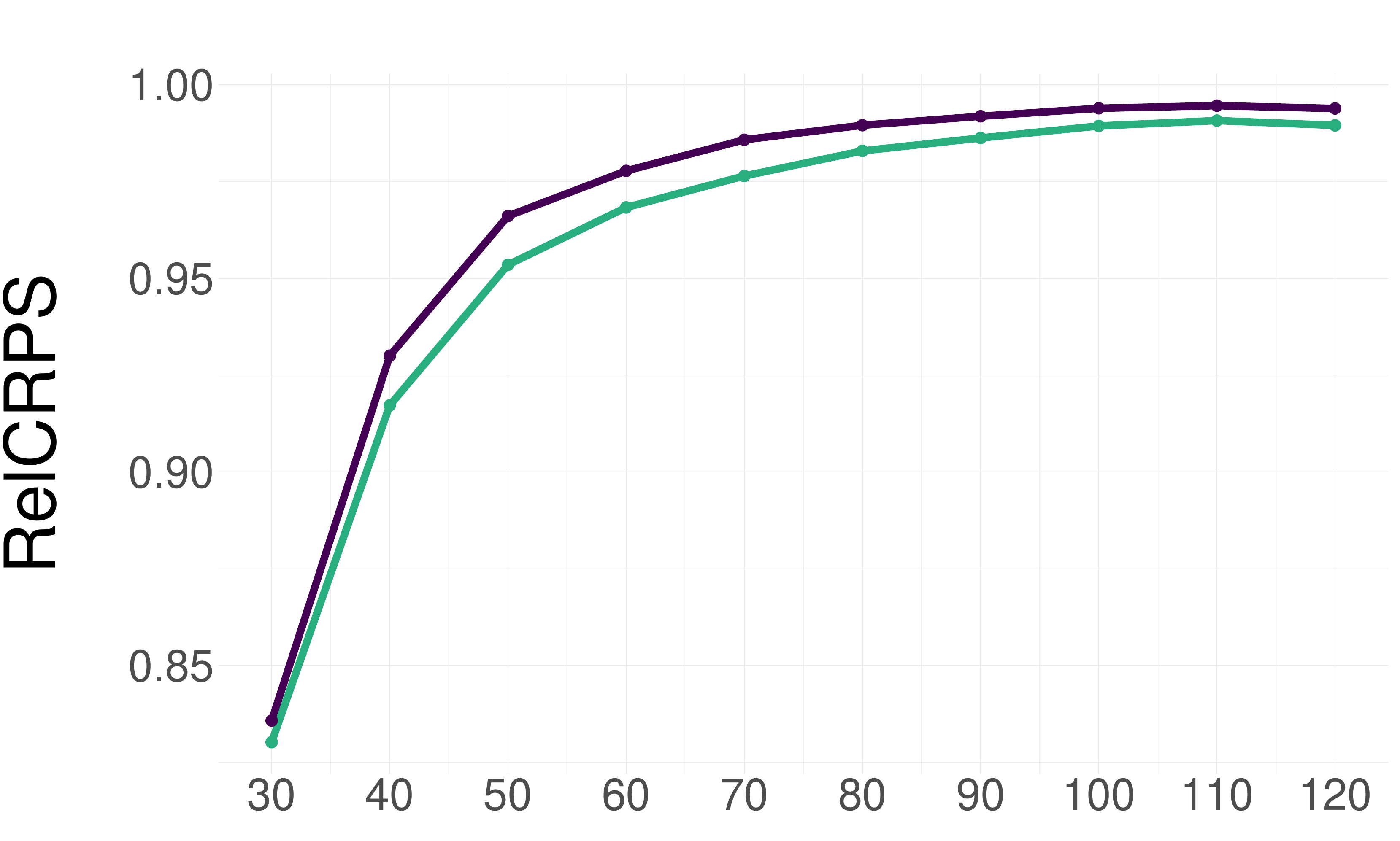}
            \put(80,10){\includegraphics[width=0.17\textwidth]{legend.pdf}}
        \end{overpic}
    \end{subfigure}
    \hfill
    \begin{subfigure}[b]{0.42\textwidth}
        \centering
        \begin{overpic}[width=\linewidth]{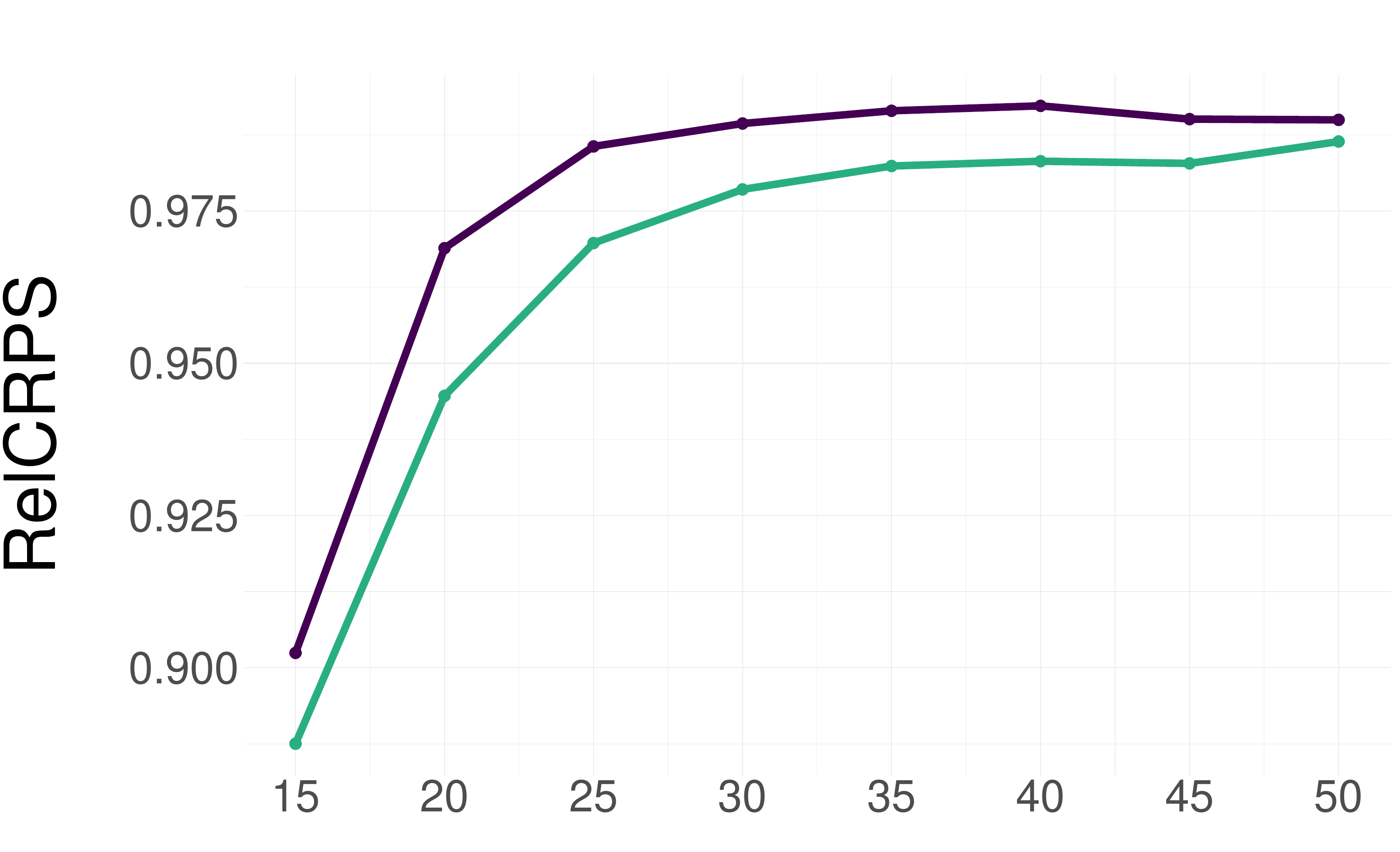}
            \put(80,10){\includegraphics[width=0.17\textwidth]{legend.pdf}}
        \end{overpic}
    \end{subfigure}

    \begin{subfigure}[b]{0.42\textwidth}
        \centering
        \begin{overpic}[width=\linewidth]{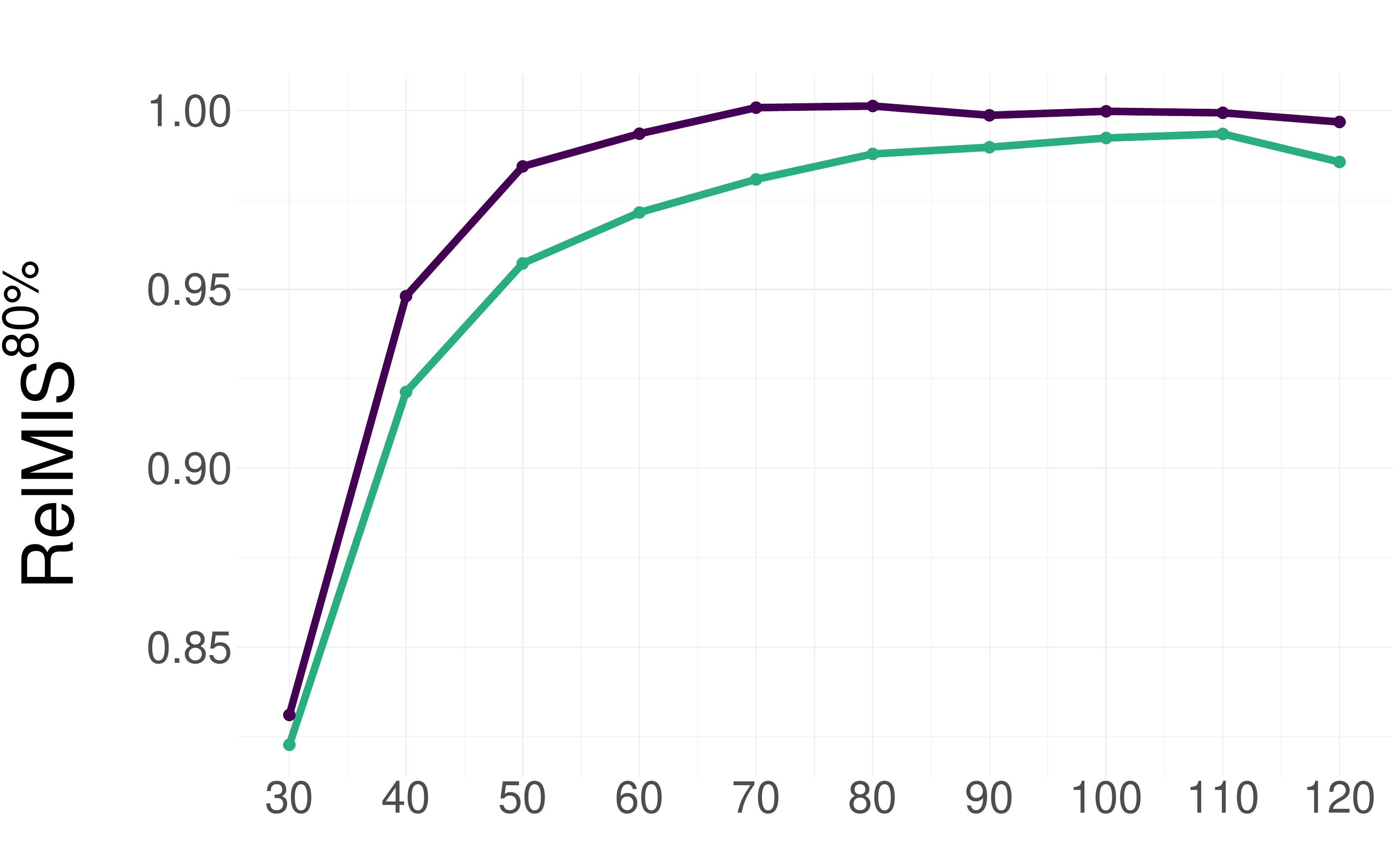}
            \put(80,10){\includegraphics[width=0.17\textwidth]{legend.pdf}}
        \end{overpic}
    \end{subfigure}
    \hfill
    \begin{subfigure}[b]{0.42\textwidth}
        \centering
        \begin{overpic}[width=\linewidth]{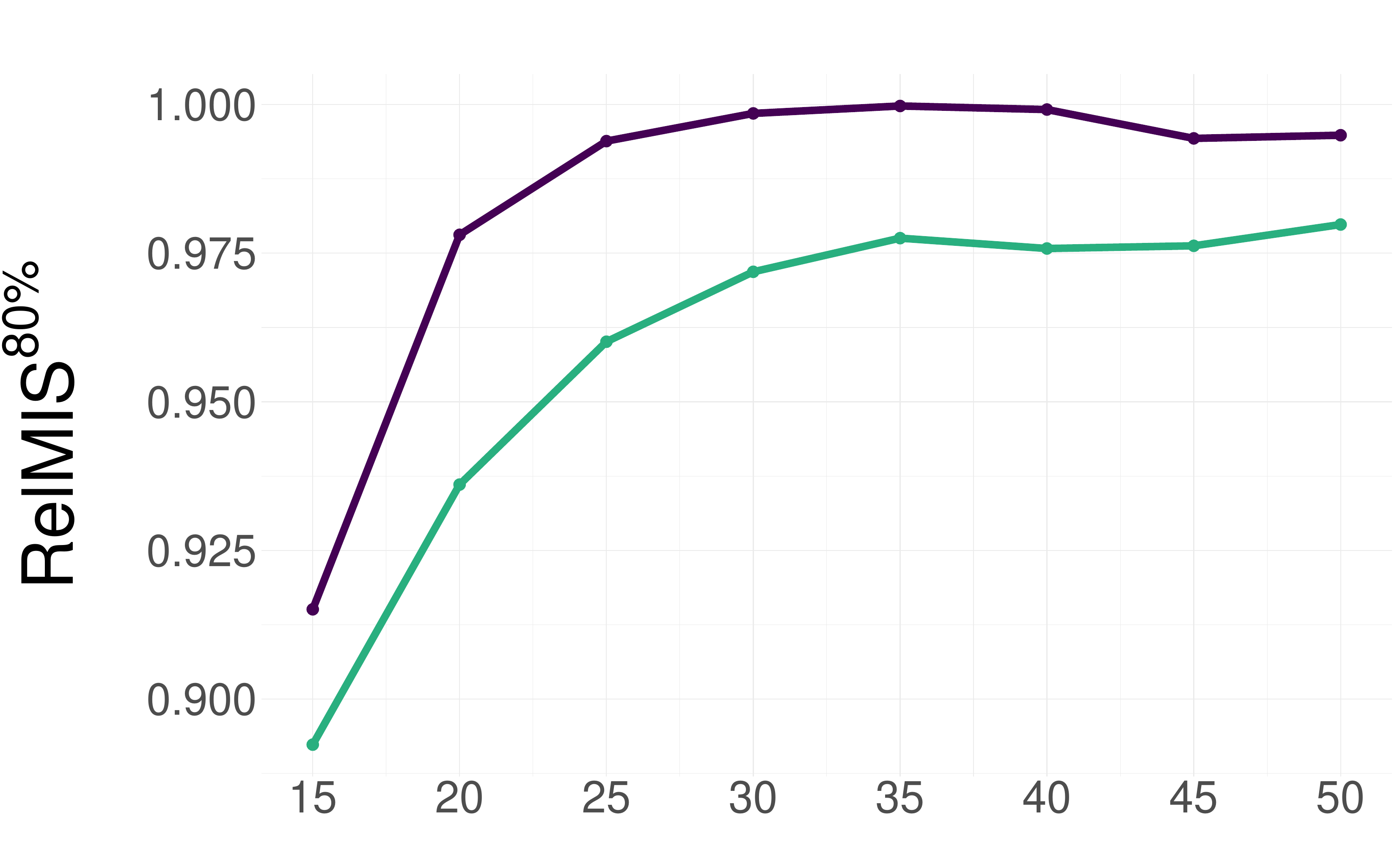}
            \put(80,10){\includegraphics[width=0.17\textwidth]{legend.pdf}}
        \end{overpic}
    \end{subfigure}

    \begin{subfigure}[b]{0.42\textwidth}
        \centering
        \begin{overpic}[width=\linewidth]{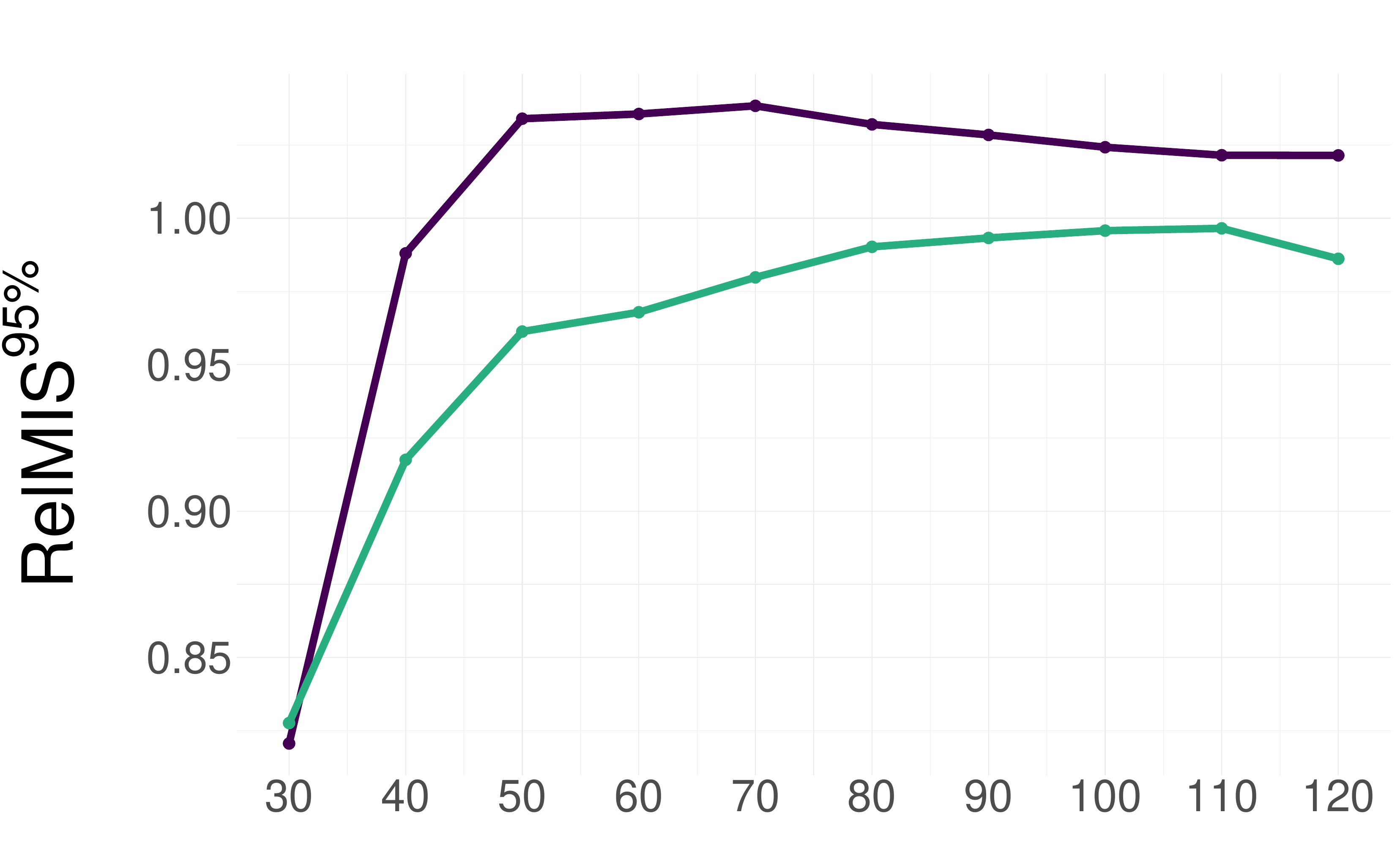}
            \put(80,15){\includegraphics[width=0.17\textwidth]{legend.pdf}}
        \end{overpic}
    \end{subfigure}
    \hfill
    \begin{subfigure}[b]{0.42\textwidth}
        \centering
        \begin{overpic}[width=\linewidth]{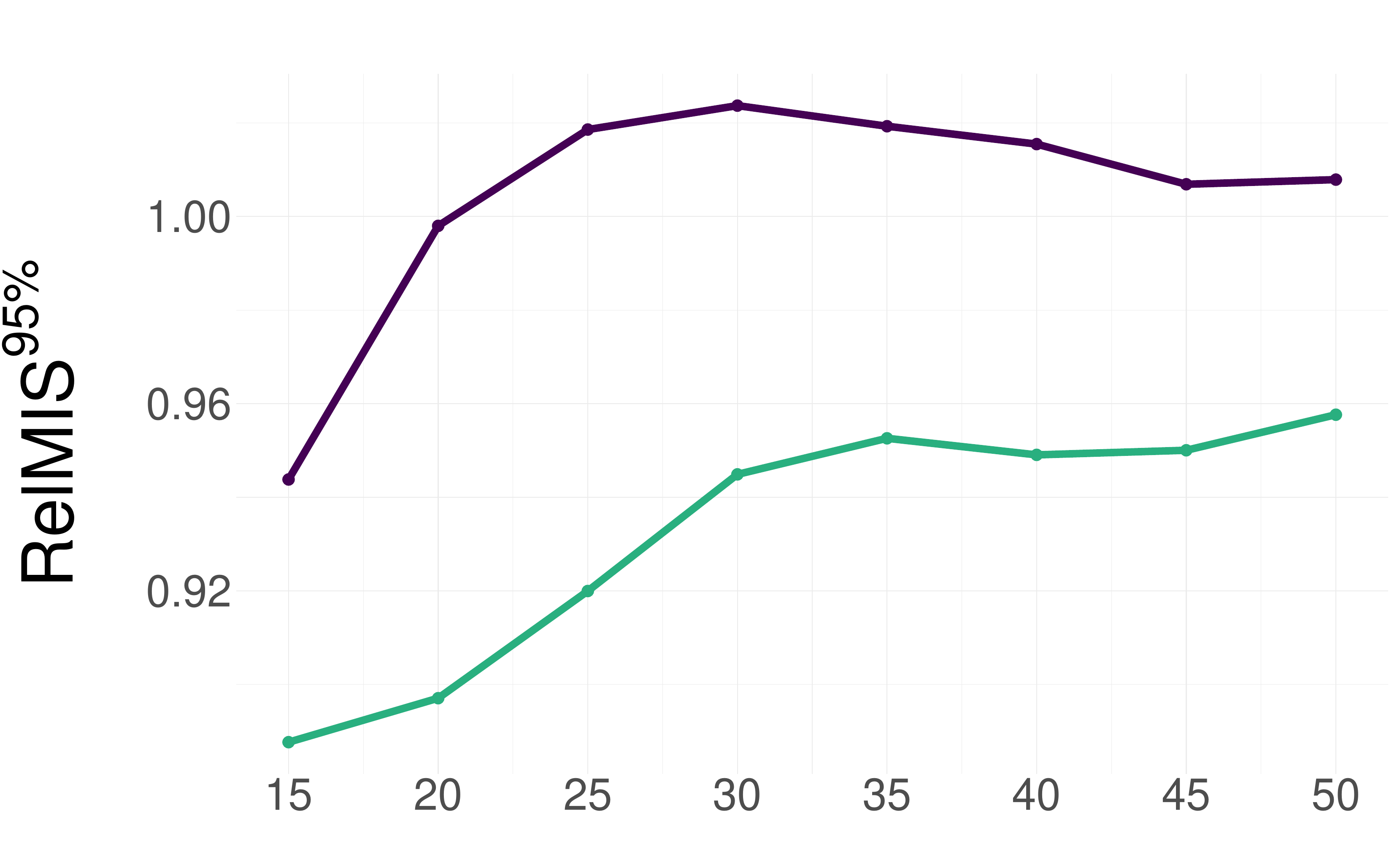}
            \put(80,15){\includegraphics[width=0.17\textwidth]{legend.pdf}}
        \end{overpic}
    \end{subfigure}
    
    \begin{subfigure}[b]{0.42\textwidth}
        \centering
        \begin{overpic}[width=\linewidth]{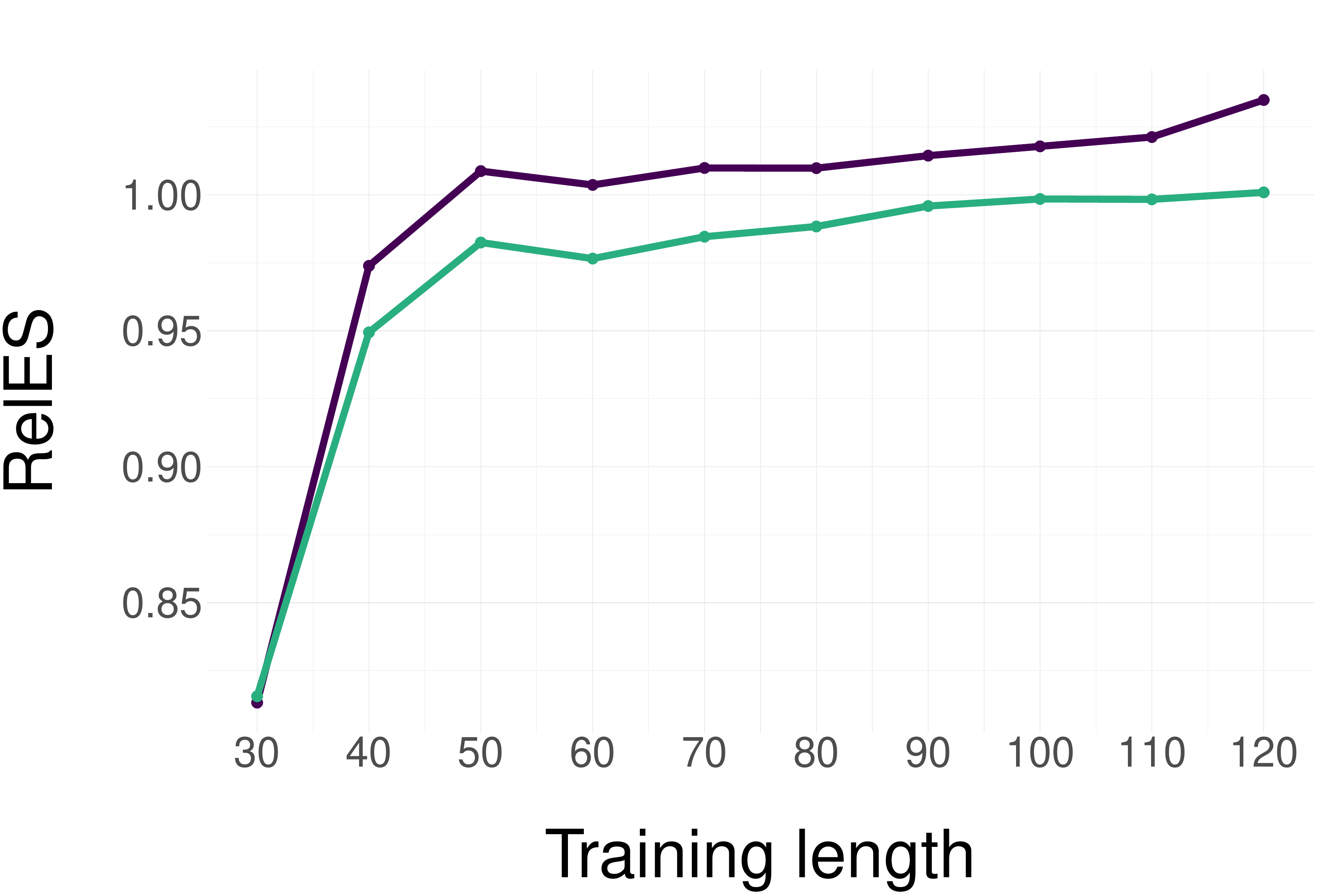}
            \put(80,15){\includegraphics[width=0.17\textwidth]{legend.pdf}}
        \end{overpic}
    \end{subfigure}
    \hfill
    \begin{subfigure}[b]{0.42\textwidth}
        \centering
        \begin{overpic}[width=\linewidth]{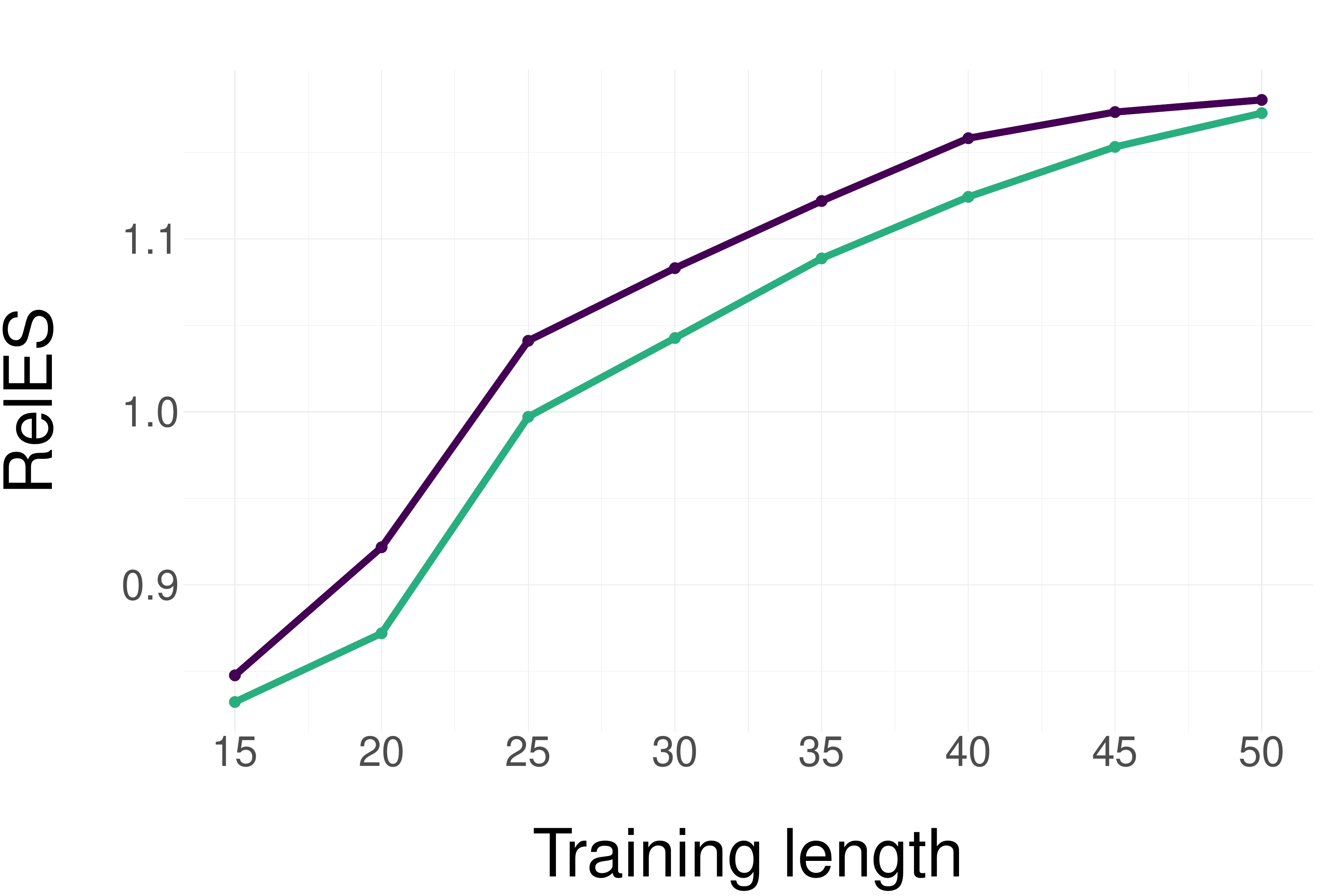}
            \put(80,15){\includegraphics[width=0.17\textwidth]{legend.pdf}}
        \end{overpic}
    \end{subfigure}

    \vspace{3 mm}

    \caption{   
    Results on the \textit{Australian Tourism-M} (left column) and \textit{Australian Tourism-Q} (right column) datasets for \textit{MinT} (purple) and \textit{\tRec} (green).
   A relative score lower than 1 means improvement over the base forecasts.
   }
    \label{fig: rel scores aus}
\end{figure}

\subsection{Australian Tourism results}

We observe the same patterns
on  \textit{Australian tourism-M} and \textit{Australian tourism-Q}. 
In terms of the MSE, \textit{\tRec} and \textit{MinT} perform similarly, indicating comparable point forecast accuracy.
However, \textit{\tRec} yields better probabilistic forecasts, with a slight advantage in CRPS and ES. There are substantial improvements in MIS, particularly at the 95\% level, reflecting more reliable prediction intervals (Fig. \ref{fig: rel scores aus}).

We also compute the geometric mean of the relative prediction interval widths and the arithmetic mean of coverage rates across the time series.
We report the results for two different \trains\ in Table~\ref{tab:PI_width_cov_aus}.

\begin{table}[!h]
\centering
\begin{tabular}{l @{\hskip 20pt}l @{\hskip 20pt} c cccc c cccc}
\toprule
& & \multicolumn{4}{c}{\textbf{Australian Tourism - M}} & & \multicolumn{4}{c}{\textbf{Australian Tourism - Q}} \\
& & \multicolumn{2}{c}{\textbf{Train 55}} & \multicolumn{2}{c}{\textbf{Train 110}} & & \multicolumn{2}{c}{\textbf{Train 25}} & \multicolumn{2}{c}{\textbf{Train 40}} \\
\cmidrule(lr){3-6} \cmidrule(lr){8-11} 
 & \textbf{Level} & \textit{MinT} & \textit{\tRec} & \textit{MinT} & \textit{\tRec} & & \textit{MinT} & \textit{\tRec} & \textit{MinT} & \textit{\tRec} \\
\midrule
\multirow{2}{*}{\textbf{Relative Width}} & 80\% & 0.88 & 0.95 &0.92 & 0.97  & & 0.95 & 1.09 & 0.96 & 1.12 \\
 & 95\% & 0.88 & 0.95 & 0.92 & 0.98 & & 0.95 & 1.10 & 0.96 & 1.13 \\
 [1.8ex]
\multirow{2}{*}{\textbf{Coverage}} & 80\%& 0.74 & 0.77 & 0.79 & 0.81  & & 0.69 & 0.75 & 0.69  &  0.77 \\
 & 95\% & 0.88 & 0.91 & 0.91 & 0.93 & & 0.86 & 0.91 &  0.87 & 0.92 \\
\bottomrule
\end{tabular}
\caption{Relative width of the prediction intervals and coverage at the 80\% and 95\% confidence levels for the \textit{Australian Tourism-M} and \textit{Australian Tourism-Q} datasets.
Results are reported for various \trains.}
\label{tab:PI_width_cov_aus}
\end{table}

As for the \textit{Swiss Tourism} dataset,
\textit{\tRec} produces wider intervals than \textit{MinT}, resulting in average coverage rates closer to the nominal levels (80\% or 95\%).
The statistical analysis based on MCB confirms 
that \tRec~significantly improves the MIS score at both 80\% and 95\% levels
compared to both \textit{MinT} and the base forecasts.

\subsection{Ablation study}

To isolate the contributions of specific components of our method, we conduct an ablation study comparing \textit{\tRec} against three variants: 

\begin{itemize}
    \item \textit{\tRec-Diag}: the prior scale matrix $\Psibf_0$ is specified as in \textit{\tRec} but restricted to be diagonal, ignoring prior correlations;     
    \item \textit{\tRec-min\_$\nu_0$}: the $\nu_0$ parameter is fixed at its minimum admissible value ($n+2$), skipping the optimization procedure and making the prior weakly informative;
    \item \textit{\tRec-MAP} (see Section~\ref{sec:simulations}): 
    a maximum a posteriori point estimate is used for $\Wbf$, instead of marginalizing over its posterior distribution.
\end{itemize}
We present the results for the \textit{Swiss Tourism} dataset in Fig.~\ref{fig:swiss score_ablation}.
Additional experiments on the \textit{Australian Tourism-M} and \textit{Australian Tourism-Q} datasets yield similar results and are reported in \ref{app:datasets} (Fig.~\ref{fig: rel scores aus ablation}).
Point forecast accuracy is largely equivalent across methods, with the exception of \textit{\tRec-min\_$\nu_0$}, which suffers from insufficient shrinkage due to the minimal value of $\nu_0$.
However, significant differences emerge in probabilistic assessments, where \textit{\tRec} consistently achieves the best CRPS, MIS, and ES across different training lengths $T$.
Specifically,
\textit{\tRec-MAP} exhibits the poorest performance:
since it uses a point estimate of the covariance matrix, it does not account for its uncertainty, yielding prediction intervals that are systematically too short.
This highlights the fundamental role of the full Bayesian framework in incorporating covariance uncertainty. 
Furthermore, the suboptimal performance of \textit{\tRec-min\_$\nu_0$} validates the necessity of the optimization step for $\nu_0$.
Finally, while the improvement of \textit{\tRec} over \textit{\tRec-Diag} is the smallest among the observed effects, it nonetheless confirms that 
prior information on the correlations among residuals from naive methods is informative and contributes to improved predictive performance.

\begin{figure}[H]
    \centering
    \vspace{-2mm}

    \includegraphics[width=0.95\linewidth]{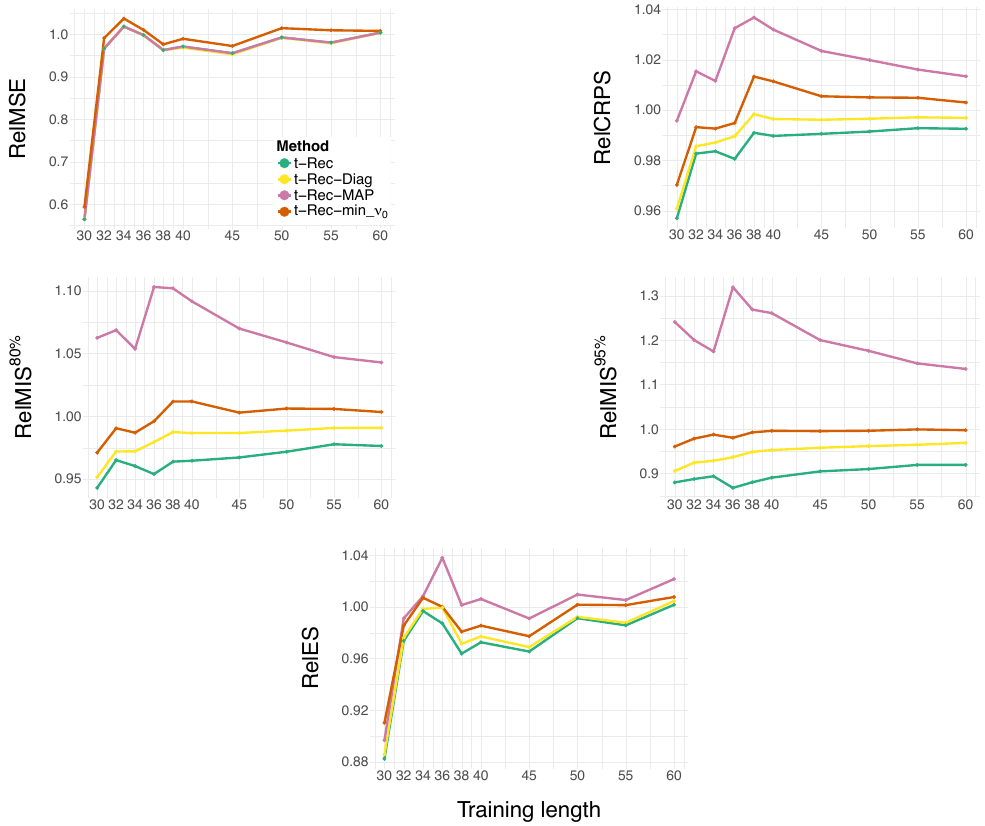}

    \vspace{0.2cm}

    \caption{Results on the \textit{Swiss Tourism} dataset for
    \textit{\tRec} (green), \textit{\tRec-Diag} (yellow),
    \textit{\tRec-MAP} (pink), and
    \textit{\tRec-min-$\nu_0$} (orange).
    A relative score lower than 1 indicates an improvement over the
    base forecasts.}
    \label{fig:swiss score_ablation}
\end{figure}

\section{Conclusion}\label{sec: conclusion}

We proposed \textit{\tRec}, a Bayesian method for probabilistic reconciliation that accounts for the uncertainty on the covariance matrix.
Given our choices of prior and likelihood,
we obtain a multivariate-t reconciled predictive distribution in closed form. 
Empirical results showed that \textit{\tRec} yielded comparable point forecast accuracy to \textit{MinT} but consistently outperformed it on prediction intervals.
Our results were robust across multiple datasets and various training set lengths.

Despite its robust performance, the effectiveness of \textit{\tRec} depends on the specification of the prior scale parameter $\Psi_0$.
Our approach based on the residuals of the naive and seasonal naive methods may fail for complex time series with multiple seasonalities or strong covariate dependencies.
In those cases, alternative definitions of $\Psi_0$ might be necessary to maintain a well-specified prior.
Moreover, we acknowledge the limitations of the Inverse-Wishart prior, particularly in high-dimensional settings \citep{gelman2006prior}. Future work will explore more modern priors for the covariance matrix.
Finally, our analysis focused on the cross-sectional hierarchical setting; we leave for future studies the extension to temporal and cross-temporal hierarchies.

\section*{Acknowledgments}
This work is partially funded by the Swiss National
Science Foundation (SNF), Switzerland, grant 200021\_
212164/1 awarded to G.C. and D.A. This project has received
funding from the European Unions Horizon Europe
Research and Innovation Framework under grant
agreement No. 101160720.

\bibliography{biblio}

\newpage

\appendix

\section{Reconciliation via conditioning of multivariate t}
\label{app: rec via cond t}

Similar to the Gaussian case \citep{zambon2024properties}, if the joint base forecast distribution is a multivariate t, then the reconciled distribution and the parameters can be computed in closed form. 
Theorem~\ref{theorem: t reconc via conditioning} is proved by applying the following result to Eq.~\eqref{eq: posterior predictive mt} with $\PsiH =  \frac{1}{\nu'-n+1} \Psibf'$ and $\nuhat = \nu'-n+1$.

\begin{proposition} \label{prop: t reconc via conditioning}
Let $\Abf \in \rr^{(n-n_b) \times n_b}$ be the aggregation matrix of a hierarchy with $n$ time series, and let the joint forecast distribution be a multivariate t:
$\YH \sim \mt\big(\yH, \PsiH, \nuhat\big) = \mt\left(\begin{bmatrix} \uH \\ \bH \end{bmatrix} , 
\begin{bmatrix} \Psibf'_{U} & \Psibf'_{UB}  \\ \Psibf'_{BU}  & \Psibf'_{B}  \end{bmatrix} , \nuhat\right)$.
Then, the reconciled distribution via conditioning of the $n_b$ bottom series is still a multivariate t:
\begin{equation}
\BT
\,\sim\, \mt\big(\bT,\, \SigmaT_B,\, \nutil \big),
\end{equation}
where
\begin{align}
\bT &= \bH + \left(\PsiH_{UB}^T - \PsiH_B \Abf^T\right) \Qbf^{-1} (\Abf \bH - \uH), \\
\SigmaT_B &= 
C
\left[\PsiH_B - \left(\PsiH_{UB}^T - \PsiH_B \Abf^T\right) \Qbf^{-1} \left(\PsiH_{UB}^T - \PsiH_B \Abf^T\right)^T \right], \\
\nutil &= \nuhat + n - n_b,
\end{align}
and 
\begin{align}
C &=  \frac{\nuhat +  (\Abf \bH - \uH)^T \Qbf^{-1} (\Abf \bH - \uH)}{\nuhat + (n-n_b)}, \nonumber \\  
\Qbf &= \PsiH_U - \PsiH_{UB} \Abf^T - \Abf \PsiH_{UB}^T + \Abf \PsiH_B \Abf^T.
\end{align}
\end{proposition}

To prove Proposition~\ref{prop: t reconc via conditioning}, we first state two technical lemmas about key properties of the multivariate t-distribution. 
We refer to \cite{Kotz_Nadarajah_2004} for the proof of Lemma~1 and to \cite{ding2016conditional} for the proof of Lemma~2.

\begin{lemma} \label{lemma: affine transform of multiv t}
Let $\Xbf$ be a $n$-dimensional random vector distributed as a multivariate t:
\[\Xbf \sim \mt\big(\mubf, \Sigmabf, \nu\big),\]
and let $\Vbf \in \rr^{n \times n}$ be full rank.
Then, the random vector $\Vbf \Xbf$ is distributed as a multivariate t:
\[\Vbf \Xbf \sim \mt\big(\Vbf \mubf, \Vbf \Sigmabf \Vbf^T, \nu\big).\]
\end{lemma}

\begin{lemma} \label{lemma: condit of multiv t}
Let $\Xbf = \begin{bmatrix} \Xbf_1 \\ \Xbf_2 \end{bmatrix} \sim 
\mt\left(\begin{bmatrix} \mubf_1 \\ \mubf_2 \end{bmatrix} , 
\begin{bmatrix} \Sigmabf_{11} & \Sigmabf_{12} \\ \Sigmabf_{21} & \Sigmabf_{22} \end{bmatrix}, \nu\right)$, where $\Xbf \in \rr^{n}$, $\Xbf_1 \in \rr^{n_1}$, $\Xbf_2 \in \rr^{n_2}$ and $n_1 + n_2 = n$.
Then, the conditional  distribution of $\Xbf_1$ given $\Xbf_2$ is still a multivariate t:
\begin{equation}
\Xbf_1\,|\,\Xbf_2 \sim \mt\big(\mubf_{1|2},\, \Sigmabf_{1|2},\, \nu_{1|2}\big),    
\end{equation}
where
\begin{align}
\mubf_{1|2} &= \mubf_1 + \Sigmabf_{12} \Sigmabf_{22}^{-1} (\xbf_2 - \mubf_2), \\
\Sigmabf_{1|2} &= \frac{\nu + (\xbf_2 - \mubf_2)^T \Sigmabf_{22}^{-1} (\xbf_2 - \mubf_2)}{\nu + n_2} 
\left[ \Sigmabf_{11} - \Sigmabf_{12} \Sigmabf_{22}^{-1} \Sigmabf_{21} \right], \\
\nu_{1|2} &= \nu + n_2.
\end{align}
\end{lemma}

\begin{proof}[Proof of Proposition~\ref{prop: t reconc via conditioning}]

Let us define $\Tbf \in \rr^{n \times n}$ as
\begin{equation*}
\Tbf = \begin{bmatrix} \textbf{0} & \Ibf_{n_b} \\
\Ibf_{n-n_b} & -\Abf \end{bmatrix},
\end{equation*}
and let $\Zbf := \Tbf\YH$.
Since $\Tbf$ is full rank, from Lemma~\ref{lemma: affine transform of multiv t}, $\Zbf$ is a multivariate t:
\begin{equation}
\Zbf \sim \mt\left( \Tbf \yH, \, \Tbf\PsiH\Tbf^T,\, \nuhat \right),
\end{equation}
where
\begin{align}\label{eq: T Sigma T}
\Tbf \yH &= \begin{bmatrix} \bH \\ \uH - \Abf \bH \end{bmatrix}, \nonumber \\
\Tbf \PsiH \Tbf^T &= \begin{bmatrix} \PsiH_B & \PsiH_{UB}^T - \PsiH_B \Abf^T \\
\PsiH_{UB} - \Abf \PsiH_B & \Qbf \end{bmatrix}, \end{align}
and $\Qbf= \PsiH_U - \PsiH_{UB} \Abf^T - \Abf \PsiH_{UB}^T + \Abf \PsiH_B \Abf^T$.
Since 
\[\Zbf = \begin{bmatrix} \BH \\ \UH - \Abf\BH \end{bmatrix} =: \begin{bmatrix} \Zbf_1 \\ \Zbf_2 \end{bmatrix},\]
the reconciled bottom distribution, specified by Eq.~\eqref{eq:rec_via_cond}, is given by the conditional distribution of $\Zbf_1$ given $\Zbf_2=0$.
From Lemma~\ref{lemma: condit of multiv t}, we have that 
\begin{equation*}
\Zbf_1 \,|\, \Zbf_2 = 0 \sim \mt\left(\mubf_{1|2}, \,\Sigmabf_{1|2}, \nu_{1|2}\right), 
\end{equation*}
where
\begin{align*}
\mubf_{1|2} &= \bH + \left(\PsiH_{UB}^T - \PsiH_B \Abf^T\right) \Qbf^{-1} (\Abf \bH - \uH) = \bT, \\
\nu_{1|2} &= \nuhat + (n-n_b) = \nutil, \\
\Sigmabf_{1|2} &= \frac{\nuhat + (\Abf \bH - \uH)^T \Qbf^{-1} (\Abf \bH - \uH)}{\nuhat + (n-n_b)} \; \left[\PsiH_B - \left(\PsiH_{UB}^T - \PsiH_B \Abf^T\right) \Qbf^{-1} \left(\PsiH_{UB}^T - \PsiH_B \Abf^T\right)^T \right]  \\
&= C \; \left[\PsiH_B - \left(\PsiH_{UB}^T - \PsiH_B \Abf^T\right) \Qbf^{-1} \left(\PsiH_{UB}^T - \PsiH_B \Abf^T\right)^T \right] 
= \SigmaT_B.
\end{align*}
\end{proof}

\subsection{Reconciliation of a minimal hierarchy}
\label{app:rec_minimal}

We explicitly write the expression of the reconciled t-distribution on a minimal hierarchy with one upper and two bottom series (Fig.~\ref{fig: min_hiearchy}) using Theorem~\ref{theorem: t reconc via conditioning}.
The reconciled bottom and upper distributions are multivariate t and univariate t:
\begin{align}\label{eq: bot-up-t-stud-small-hierarchy}
    \BT \sim \mt\left(\bT,\, \SigmaT_B,\,\nutil\right), && \Util \sim \text{t}\left(\util,\, \sigmatil_u,\, \nutil\right),
\end{align}
where
\begin{align}
&\nutil = \nu' - 1, \notag \\
&\bT = \begin{bmatrix}
        \left( 1 - \frac{g_1}{Q}\right) \bhat_1 + \frac{g_1}{Q} \left(\uhat - \bhat_2\right) \\
        \left( 1 - \frac{g_2}{Q}\right) \bhat_2 + \frac{g_2}{Q} \left(\uhat - \bhat_1\right) 
    \end{bmatrix}, && \util = \Big( 1 - \frac{g_u}{Q}\Big)\, \uhat + \frac{g_u}{Q} \left(\bhat_1 + \bhat_2\right), \label{eq: u_tilde} \\
    \label{eq: sigma_u_tilde}
&\SigmaT_B = C \begin{bmatrix} \Psi_1' - \frac{g_1^2}{Q} & \Psi_{1,2}' - \frac{g_1 g_2}{Q}  \\ \Psi_{1,2}' - \frac{g_1 g_2}{Q} & \Psi_2' - \frac{g_2^2}{Q}
\end{bmatrix}, && \sigmatil_u^2 = C \left(\Psi_u' - \frac{g_u^2}{Q} \right),
\end{align}
and
\[
\begin{array}{l@{\hskip 6em}l}
C   = \dfrac{1}{\nutil} \left(1 + \dfrac{(\bhat_1 + \bhat_2 - \uhat)^2}{Q}\right), 
& Q = g_1 + g_2 + g_u, \\
g_1 = (\Psi_1' + \Psi'_{1,2}) - \Psi'_{u,1}, 
& \multirow{3}{*}{\(\Psibf' = \begin{bmatrix} 
\Psi'_u & \Psi'_{u,1} & \Psi'_{u,2} \\ 
\Psi'_{u,1} & \Psi'_1 & \Psi'_{1,2} \\ 
\Psi'_{u,2} & \Psi'_{1,2} & \Psi'_2 
\end{bmatrix}.\)} \\
g_2 = (\Psi_2' + \Psi'_{1,2}) - \Psi'_{u,2}, 
 \\
g_u = \Psi'_{u} - \Psi'_{u,1} - \Psi'_{u,2}, 
\end{array}
\]

\section{Additional details and results for the simulations}\label{app:simulation}

\paragraph{Simulation setup}

For the experiments in Section~\ref{sec: simulations}, we adopt a simulation framework similar to that described in Section~3.4 of \cite{wickramasuriya2019optimal}.
We consider the minimal hierarchy of Fig.~\ref{fig: min_hiearchy}.
The length of the time series is set to 12.
The bottom-level time series are simulated using a basic structural time series model defined as
$$
\bbf_t = \mubf_t + \gammabf_t + \etabf_t,
$$
where $\mubf_t$ is the trend component, $\gammabf_t$ is the seasonal component, and $\etabf_t$ is the error component.
The trend component evolves according to a local linear trend model:
\begin{align*}
\mubf_t &= \mubf_{t-1} + \nubf_t + \epsilonbf_t, \qquad \epsilonbf_t \sim \mathcal{N}(0,\, 2 I_2),   \\
\nubf_t &= \nubf_{t-1} + \zetabf_t, \qquad\qquad\;\; \zetabf_t \sim \mathcal{N}(0,\, 0.007 I_2),
\end{align*}
where $\epsilonbf_t$, $\zetabf_t$, and $\omegabf_t$ are mutually independent and also independent over time.
The seasonal component is defined by
$$
\gammabf_t = -\sum_{i=1}^{s-1} \gammabf_{t-i} + \omegabf_t, \quad \omegabf_t \sim \mathcal{N}(0,\, 7 I_2),
$$
with the number of seasons per year set to $s = 4$ to reflect quarterly data.
The initial states $\mubf_0$, $\nubf_0$, $\gammabf_0$, $\gammabf_1$, and $\gammabf_2$ are independently drawn from a multivariate normal distribution with mean zero and identity covariance matrix.
The error component $\etabf_t$ for each series is generated from an ARIMA$(1,0,1)$ model with $\phi_1 = 0.3$ and $\theta_1 = 0.5$;
we use  
$\begin{bmatrix}
    5 & 3\\
    3 & 4
\end{bmatrix}$
as contemporaneous error covariance matrix, the same as that of series AA and AB in \cite{wickramasuriya2019optimal}.
Finally, the upper series is obtained by summing the bottom series.



\paragraph{Simulation study for different training lengths}
In Fig.~\ref{fig:sim incoherence VS T}, we show the distribution of the relative widths (i.e., reconciled PI width/base forecast PI width) of \textit{\tRec} and \textit{MinT} for different values of $T$.

\begin{figure}[H]
    \centering

\makebox[\textwidth]{%
        \parbox[b]{0.465\textwidth}{\centering \textbf{\textit{T = 5}}}%
        \hfill
        \parbox[b]{0.465\textwidth}{\centering \textbf{\textit{T = 20}}}%
    }

    \begin{subfigure}[b]{0.45\textwidth}
        \centering
        \includegraphics[width=\textwidth]{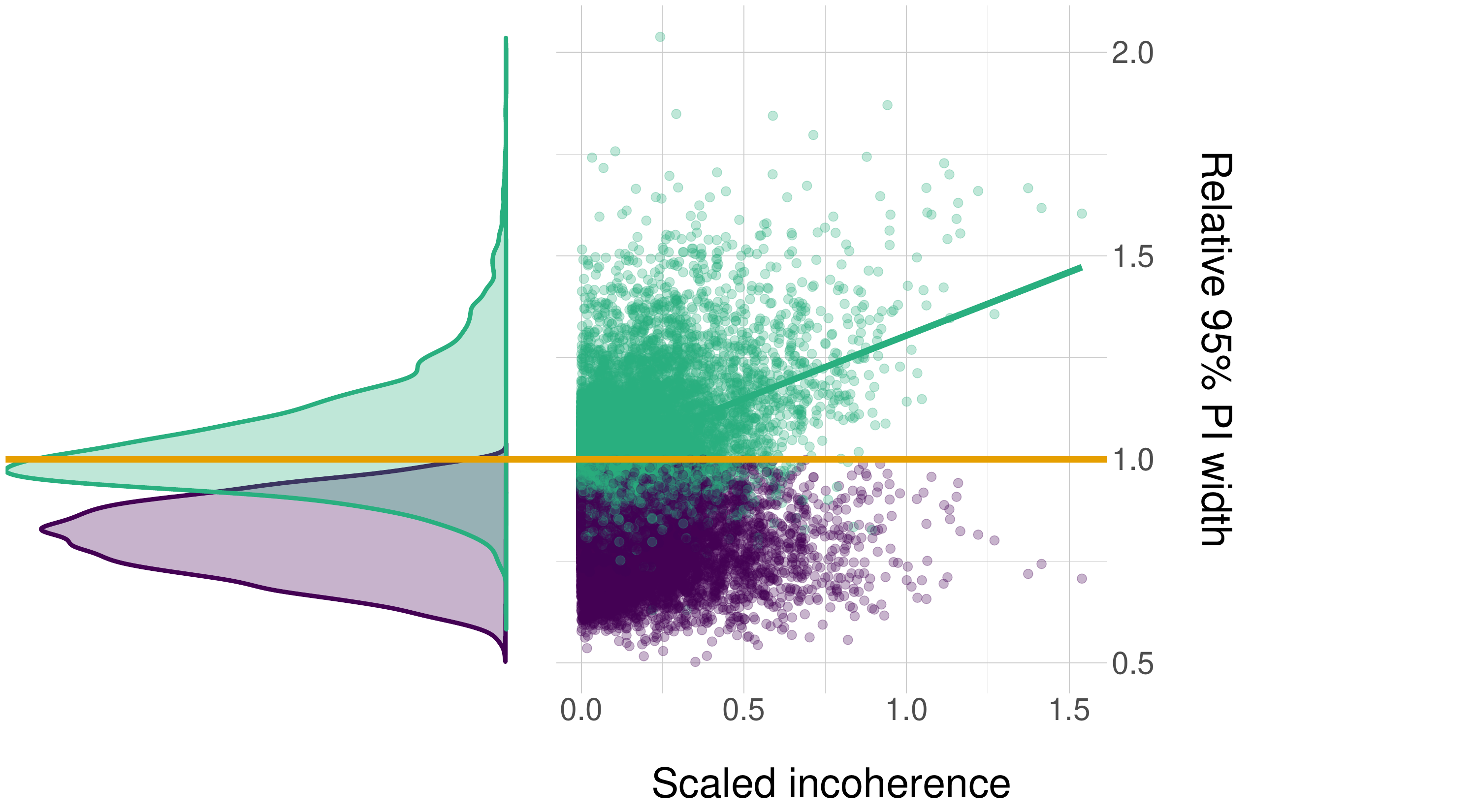}
    \end{subfigure}
\hfill
    \begin{subfigure}[b]{0.45\textwidth}
        \centering
        \includegraphics[width=\textwidth]{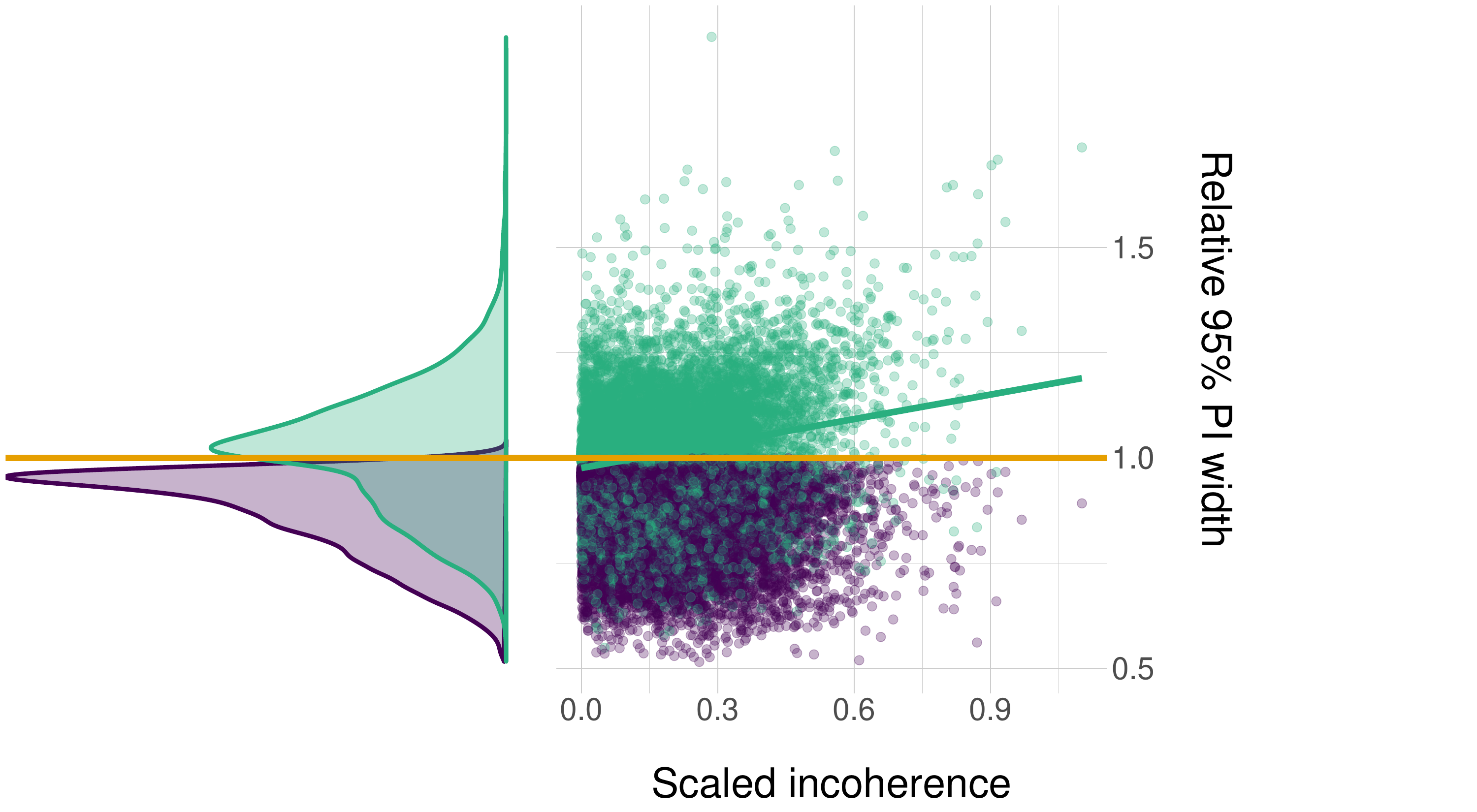}
    \end{subfigure}

    \makebox[\textwidth]{%
        \parbox[b]{0.465\textwidth}{\centering \textbf{\textit{T = 30}}}%
        \hfill
        \parbox[b]{0.465\textwidth}{\centering \textbf{\textit{T = 55}}}%
    }

    \begin{subfigure}[b]{0.45\textwidth}
        \centering
        \includegraphics[width=\textwidth]{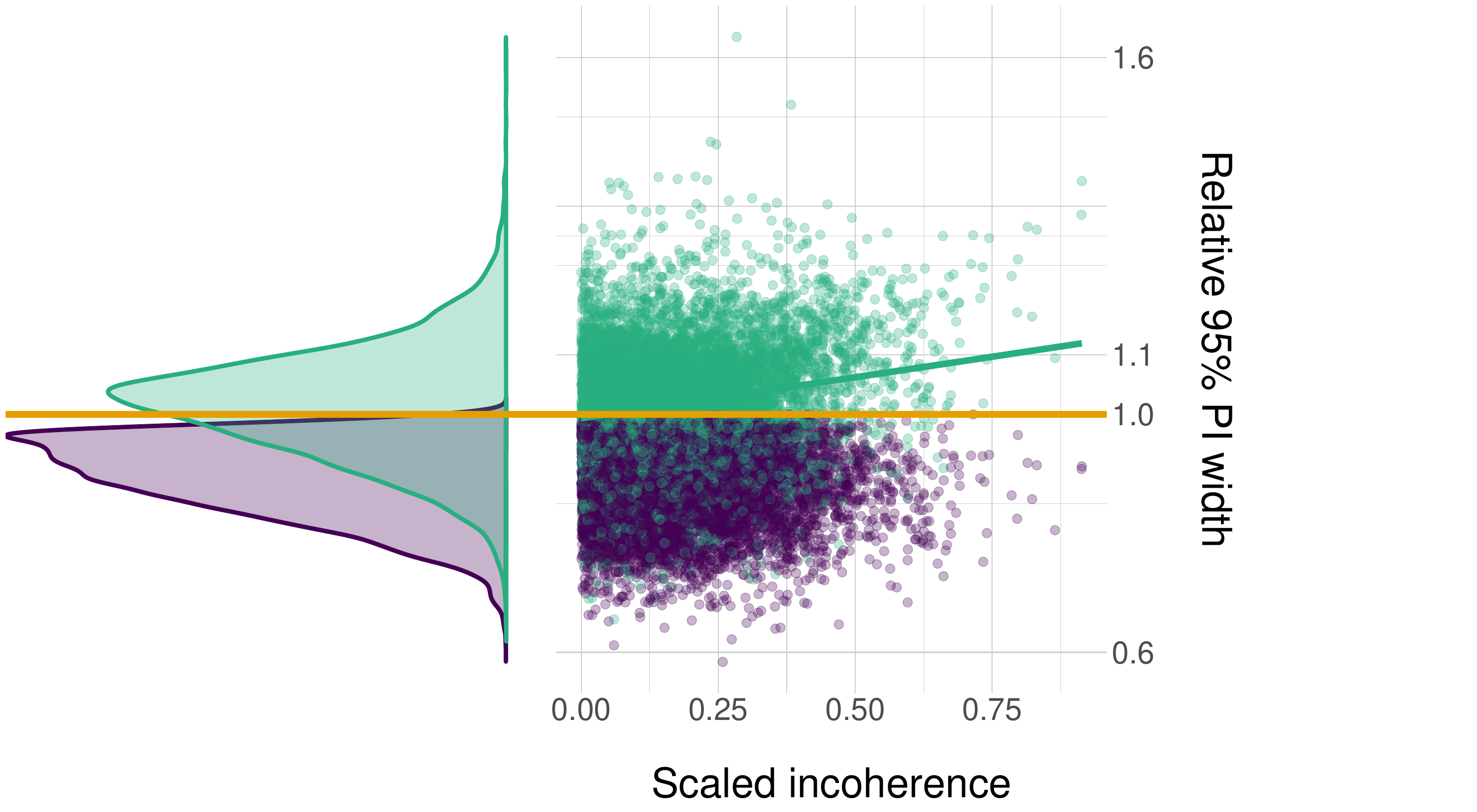}
    \end{subfigure}
\hfill
    \begin{subfigure}[b]{0.45\textwidth}
        \centering
        \includegraphics[width=\textwidth]{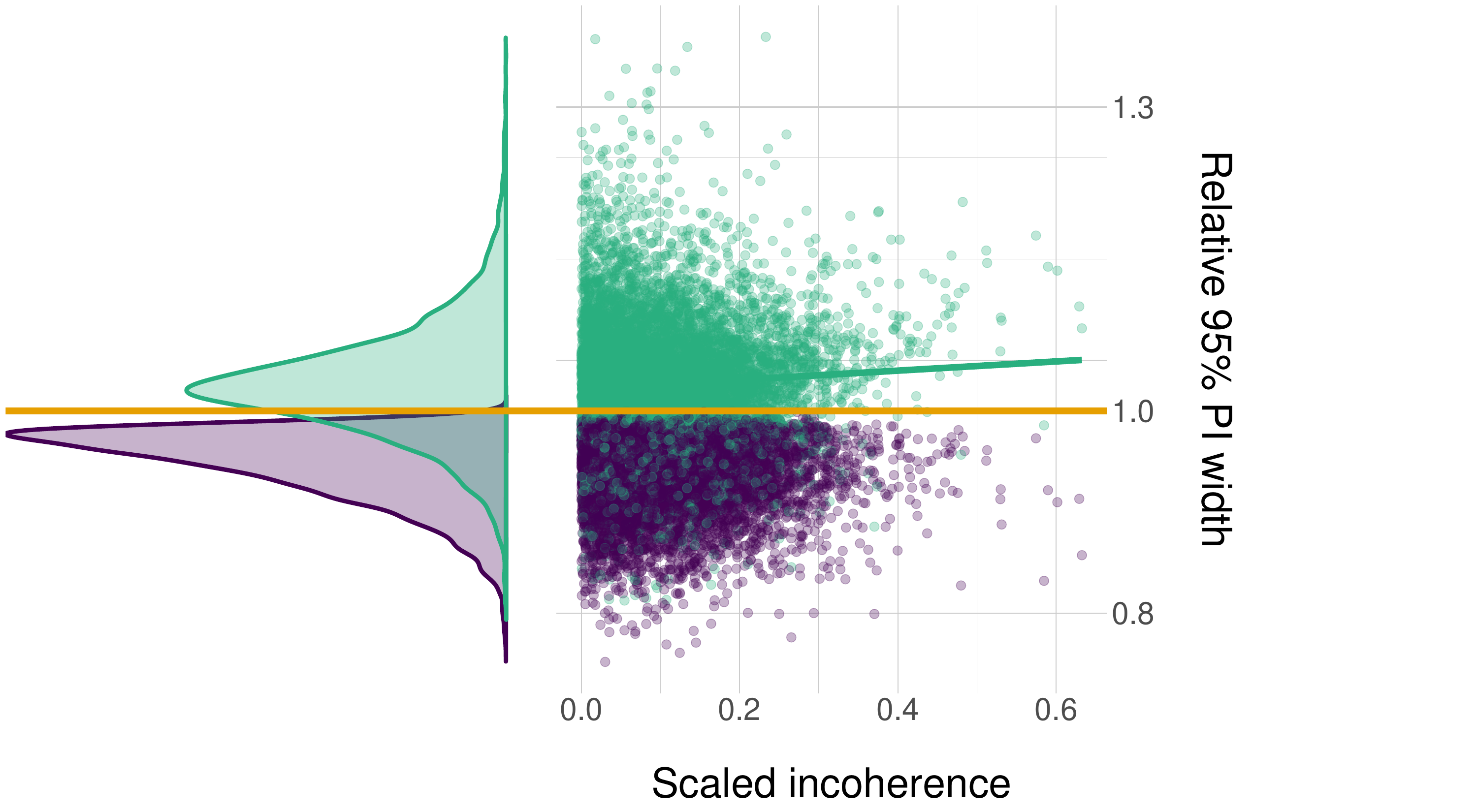}
    \end{subfigure}

    \caption{Results of the simulations for different values of the training length $T$. In each plot:
Left: distribution of the relative widths of the 95\% PIs of the upper series for \textit{MinT} (purple) and \textit{\tRec} (green), across 10,000 simulations.
Right: for each simulation, the relative widths of \textit{MinT} and \textit{\tRec} are plotted against the scaled incoherence $|\uhat -\bhat_1 - \bhat_2| \,/\, Q'^{1/2}$.
Notably, for \textit{\tRec} the dependence between the relative PI width and the scaled incoherence diminishes as $T$ grows. This happens because, as $T$ becomes large, the uncertainty on $\Wbf$ decreases and the predictive $t$-distribution converges to a Gaussian. Since, in the Gaussian case, incoherence does not affect interval widths, the influence of the incoherence progressively vanishes.
}
    \label{fig:sim incoherence VS T}
\end{figure}

\paragraph{Scores on simulations}

We assess the point and probabilistic accuracy of \textit{MinT} and \textit{\tRec} using relative scores, defined in Eq.~\eqref{eq: relMSE}, Eq.~\eqref{eq: relCRPS}, and Eq.~\eqref{eq:rel_ES}. The averages are taken over the $10,000$ independent experiments.
In Table~\ref{tab:combined_scores_sim}, we report the results obtained for different values of $T$. A relative score $<1$ means improvement over the base forecasts.
While the two methods behave similarly on the point forecasts (MSE),
\textit{\tRec}  outperforms \textit{MinT} on the probabilistic scores. 
The performance gains are  more pronounced for smaller values of $T$, for which the uncertainty on $\Wbf$ is larger.
%

\begin{table}[H]
\centering
\begin{tabular}{llccccc}
\toprule
\textbf{Score} & \textbf{Method } & $T = 5$ & $T = 12$ & $T = 20$ & $T = 30$ & $T = 55$ \\
\midrule

\multirow{2}{*}{\textbf{RelMSE}}
 & \textit{MinT}  & \textbf{0.96} & \textbf{0.98} & \textbf{0.9} & \textbf{0.93} & \textbf{0.97} \\
 & \textit{\tRec} & 0.97 & 0.99 & \textbf{0.9} & \textbf{0.93} & \textbf{0.97} \\

\midrule
\multirow{2}{*}{\textbf{RelCRPS}}
 & \textit{MinT}  & 1.01 & 1 & 0.96 & 0.97 & \textbf{0.99} \\
 & \textit{\tRec} & 0.98 & \textbf{0.99} & \textbf{0.94} & \textbf{0.96} & \textbf{0.99} \\

\midrule
\multirow{2}{*}{\textbf{$\text{RelMIS}^{80\%}$}}
 & \textit{MinT}  & 1.05 & 1.02 & 0.98 & 0.98 & \textbf{0.99} \\
 & \textit{\tRec} & \textbf{0.97} & \textbf{0.99} & \textbf{0.93} & \textbf{0.95} & \textbf{0.98} \\

\midrule
\multirow{2}{*}{\textbf{$\text{RelMIS}^{95\%}$}}
 & \textit{MinT}  & 1.14 & 1.07 & 1.03 & 1.00 & 0.99 \\
 & \textit{\tRec} & \textbf{0.89} & \textbf{0.95} & \textbf{0.87} & \textbf{0.90} & \textbf{0.96} \\

\midrule
 \multirow{2}{*}{\textbf{RelES}}
 & \textit{MinT}  & 0.99 & 0.99 & 0.94 & 0.94 & \textbf{0.97} \\
 & \textit{\tRec} & \textbf{0.97} & \textbf{0.98} & \textbf{0.92} & \textbf{0.93} & \textbf{0.97} \\

\bottomrule
\end{tabular}
\caption{Relative scores of \textit{MinT} and \textit{\tRec} across different training lengths $T$. 
A relative score $<1$ means improvement over the base forecasts.
For each score and $T$, the best-performing method is indicated in bold.
}
\label{tab:combined_scores_sim}
\end{table}

\section{Additional results on real datasets}
\label{app:datasets}

See Table~\ref{tab:cantons} and Figs.~\ref{fig:cov_swiss}, \ref{fig: rel scores aus ablation}.


\begin{table}[H]
\centering
\small
\begin{tabular}{ll@{\hspace{1cm}}ll@{\hspace{1cm}}ll}
\toprule
\textbf{Abbr.} & \textbf{Canton} &
\textbf{Abbr.} & \textbf{Canton} &
\textbf{Abbr.} & \textbf{Canton} \\
\midrule
AG & Aargau                 & NW & Nidwalden              & AI & Appenzell Innerrhoden \\
OW & Obwalden              & AR & Appenzell Ausserrhoden & SG & St. Gallen \\
BE & Bern                  & SH & Schaffhausen           & SO & Solothurn \\
BL & Basel-Landschaft      & BS & Basel-Stadt            & FR & Fribourg \\
GE & Geneva                & GL & Glarus                 & GR & Graubünden \\
JU & Jura                  & LU & Lucerne                & NE & Neuchâtel \\
SZ & Schwyz                & TG & Thurgau                & TI & Ticino \\
UR & Uri                   & VD & Vaud                   & VS & Valais \\
ZG & Zug                   & ZH & Zurich                 &     &         \\
\bottomrule
\end{tabular}
\caption{List of the 26 Swiss cantons.
}
\label{tab:cantons}
\end{table}

\begin{figure}[H]
    \centering
    \vspace{-0.2cm}
    \includegraphics[width=0.8\linewidth]{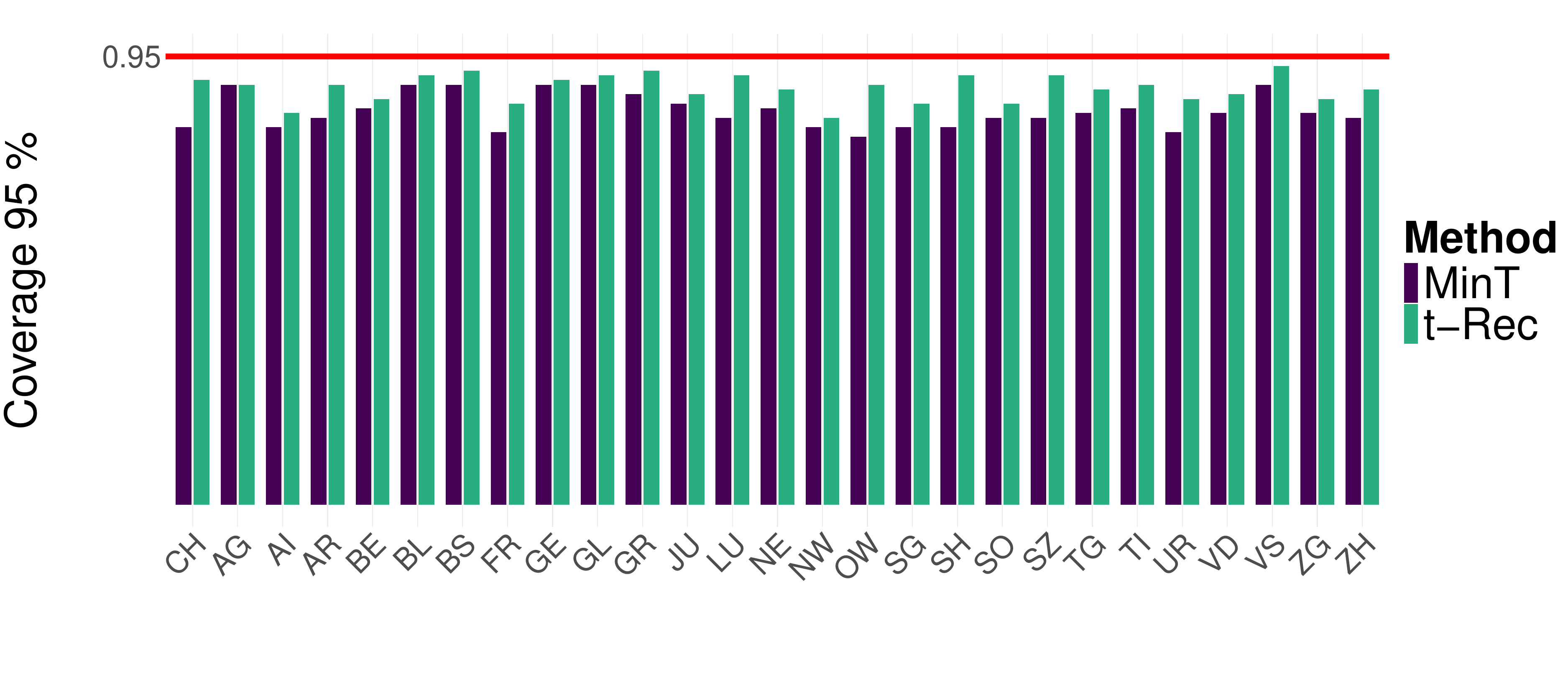}
    \caption{Coverage of the 95\% prediction intervals across 100 rolling origins (\train\ $= 40$) on the \textit{Swiss tourism} dataset. For each series, \textit{\tRec} achieves coverage closer to the target (red line) than \textit{MinT}.
    }
    \label{fig:cov_swiss}
\end{figure}




\begin{figure}[h!]
    \centering

    \vspace{-0.6cm}
    
    \makebox[\textwidth]{%
        \parbox[b]{0.465\textwidth}{\centering \textbf{\textit{Australian Tourism-M}}}%
        \hfill
        \parbox[b]{0.465\textwidth}{\centering \textbf{\textit{Australian Tourism-Q}}}%
    }
    
    \begin{subfigure}[b]{0.42\textwidth}
        \centering
        \begin{overpic}[width=\linewidth]{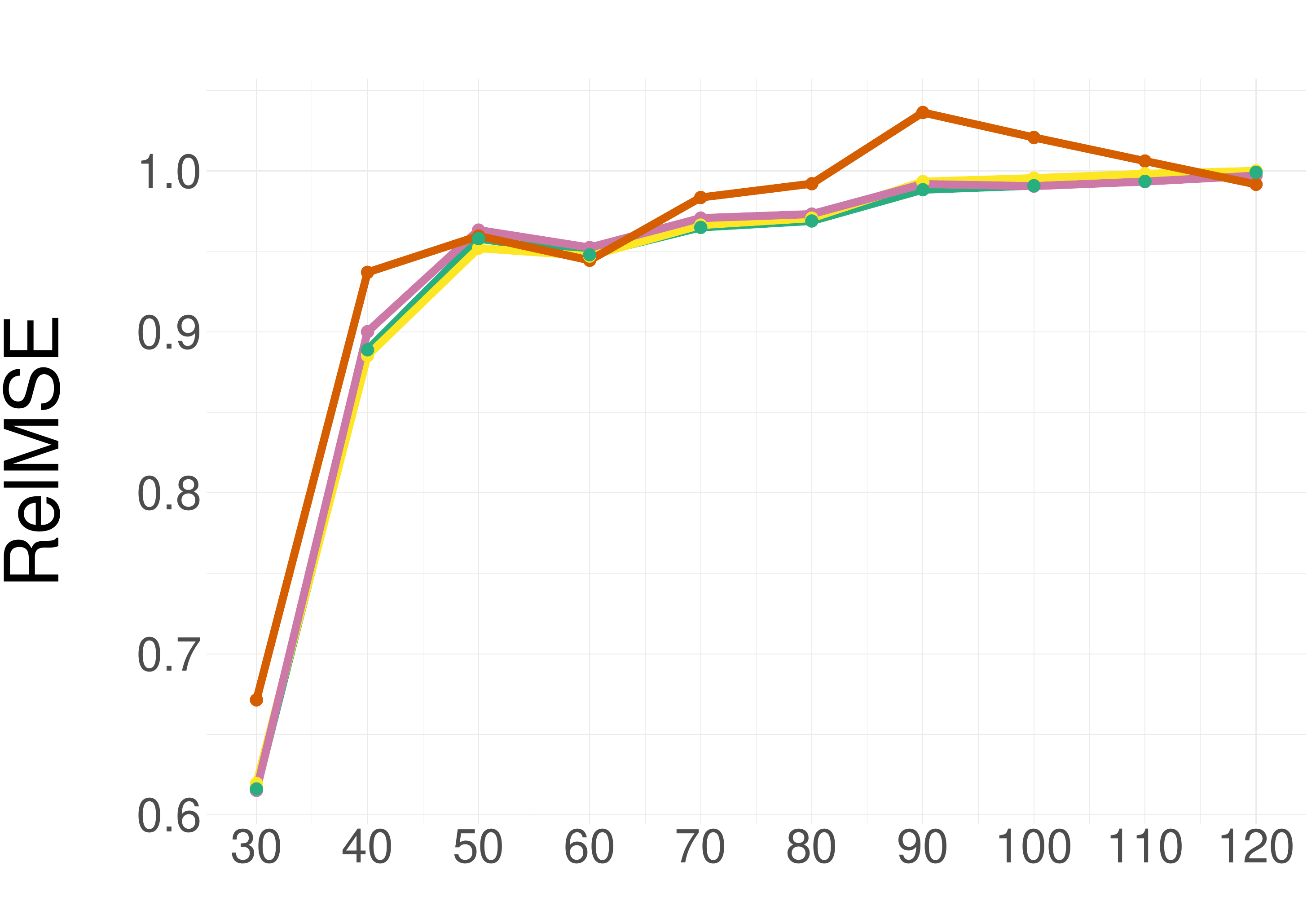}
            \put(68,20){\includegraphics[width=0.3\textwidth]{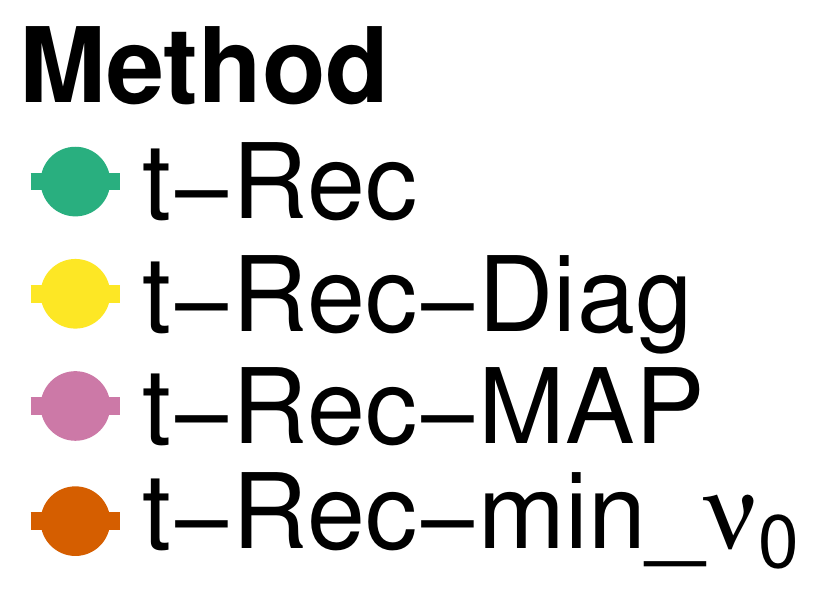}}
        \end{overpic}
    \end{subfigure}
    \hfill
    \begin{subfigure}[b]{0.42\textwidth}
        \centering
        \begin{overpic}[width=\linewidth]{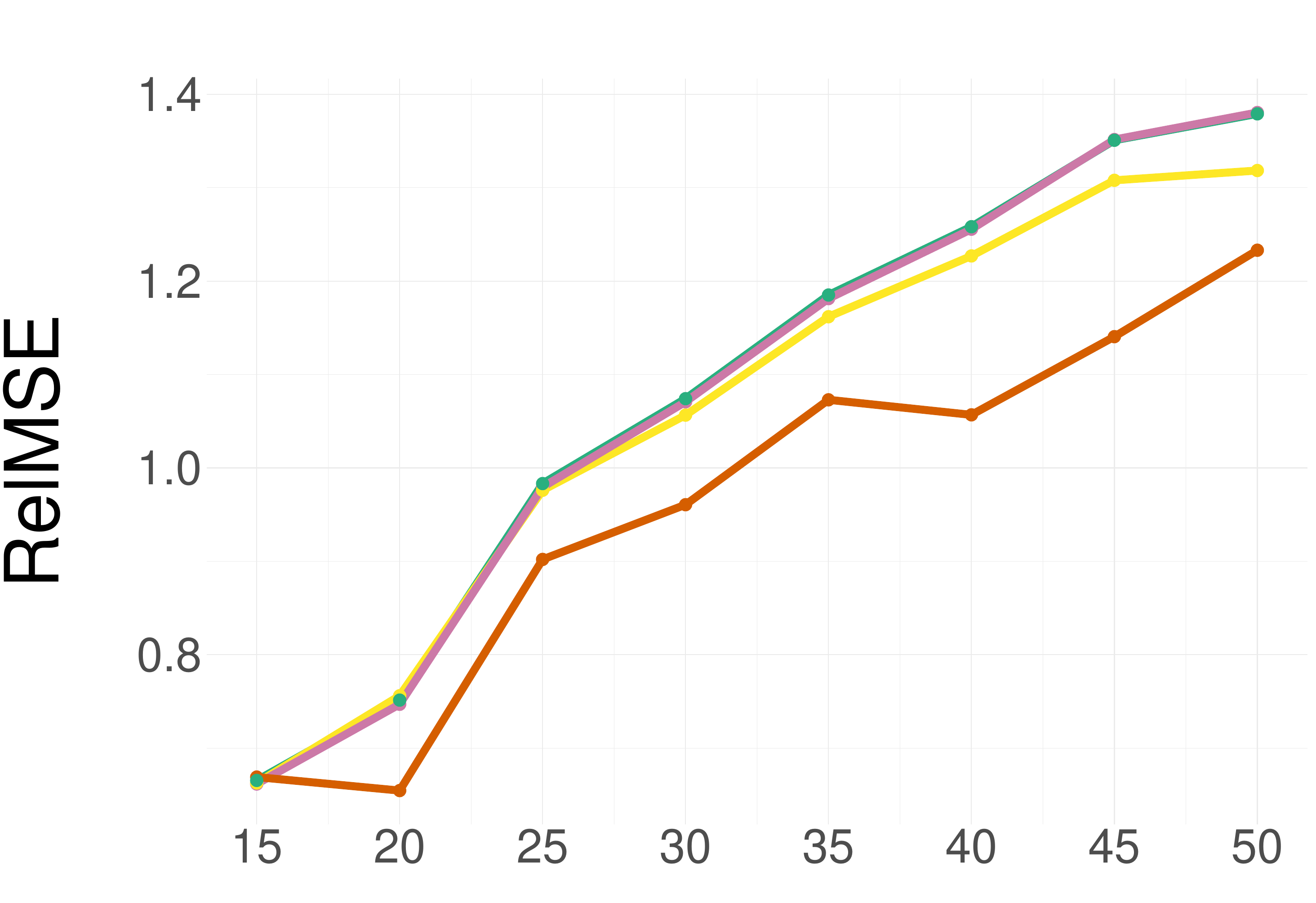}
        \end{overpic}
    \end{subfigure}

    \vspace{-0.2cm}
    
    \begin{subfigure}[b]{0.42\textwidth}
        \centering
        \begin{overpic}[width=\linewidth]{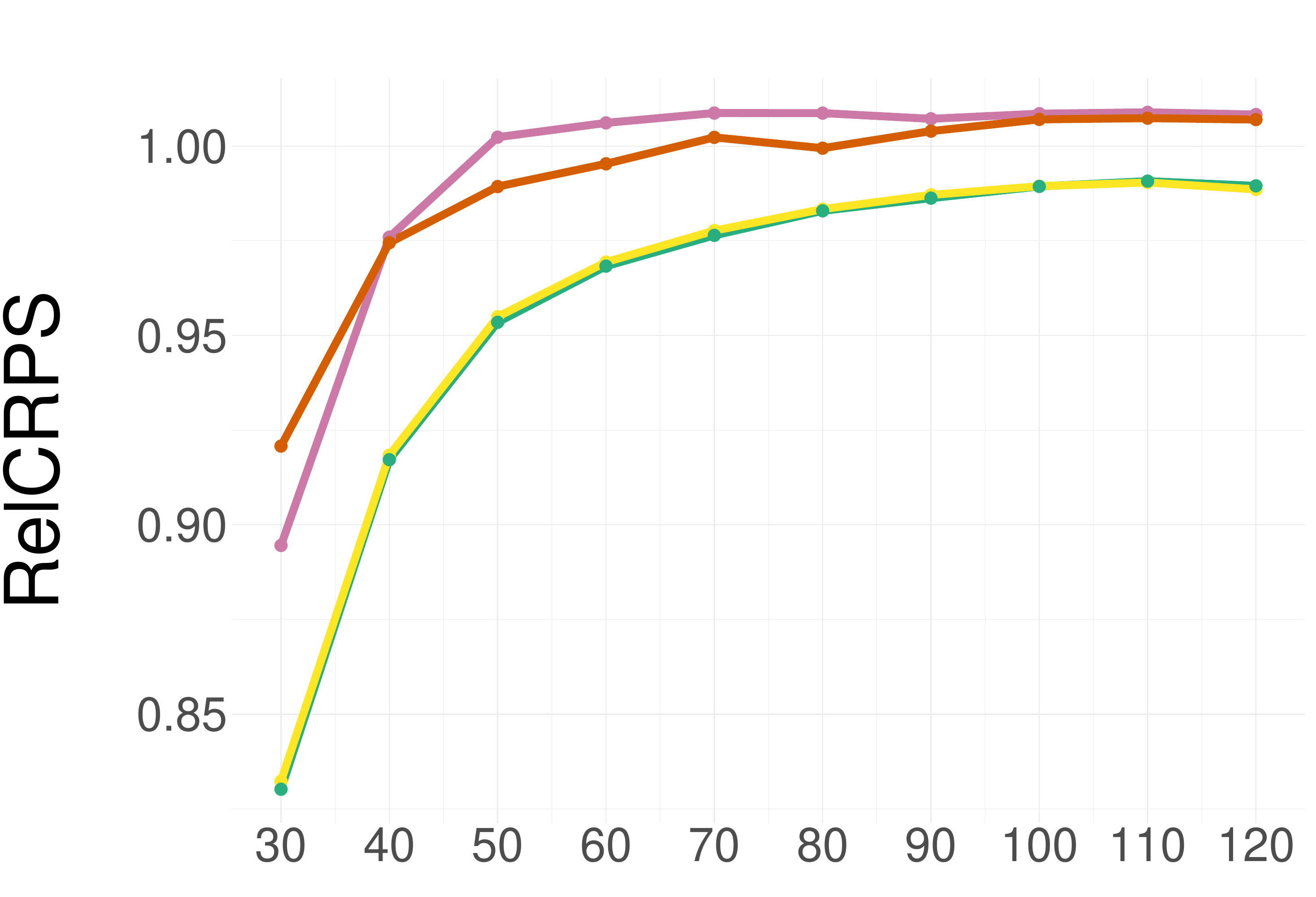}
        \end{overpic}
    \end{subfigure}
    \hfill
    \begin{subfigure}[b]{0.42\textwidth}
        \centering
        \begin{overpic}[width=\linewidth]{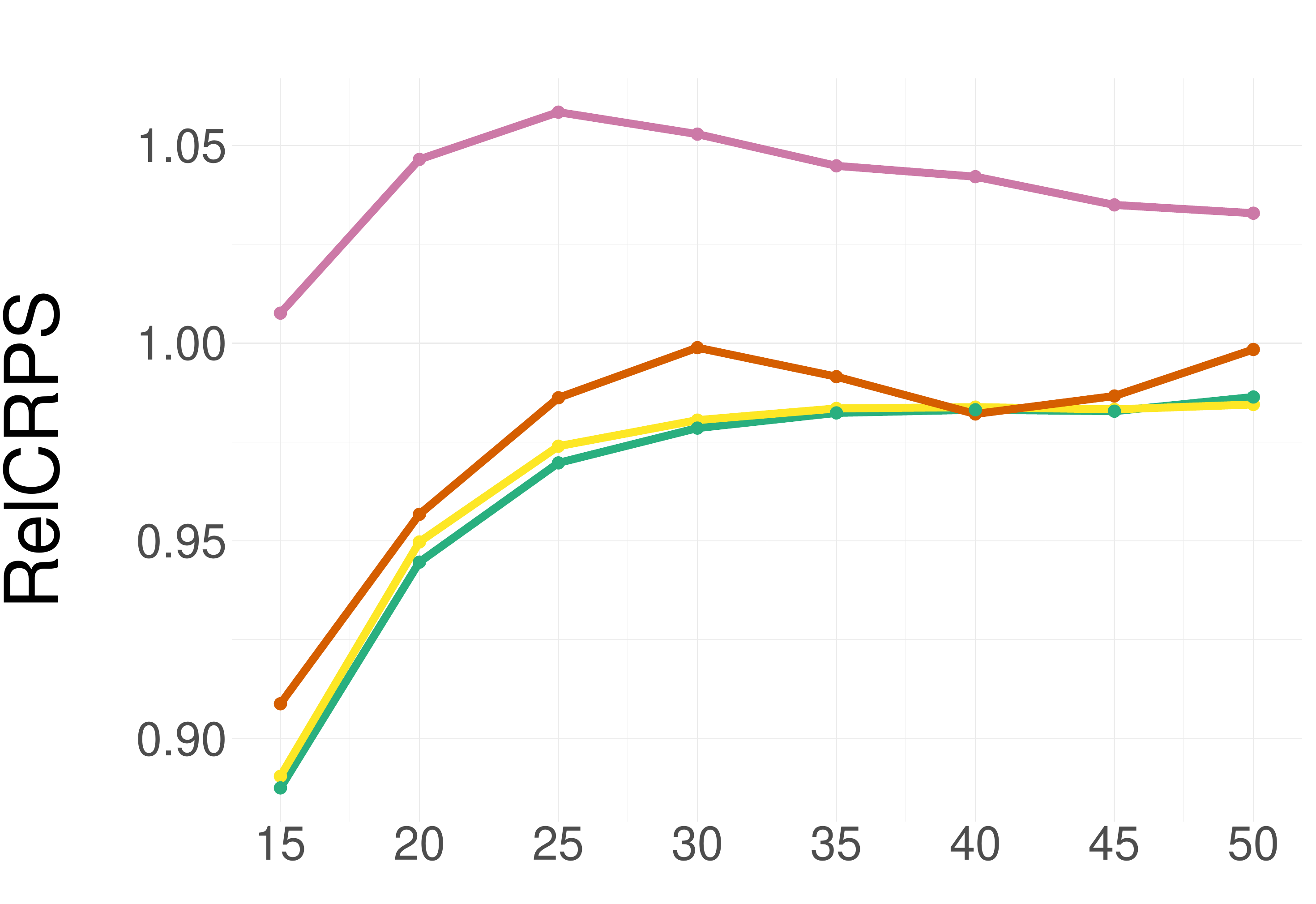}
        \end{overpic}
    \end{subfigure}

    \vspace{-0.2cm}  

    \begin{subfigure}[b]{0.42\textwidth}
        \centering
        \begin{overpic}[width=\linewidth]{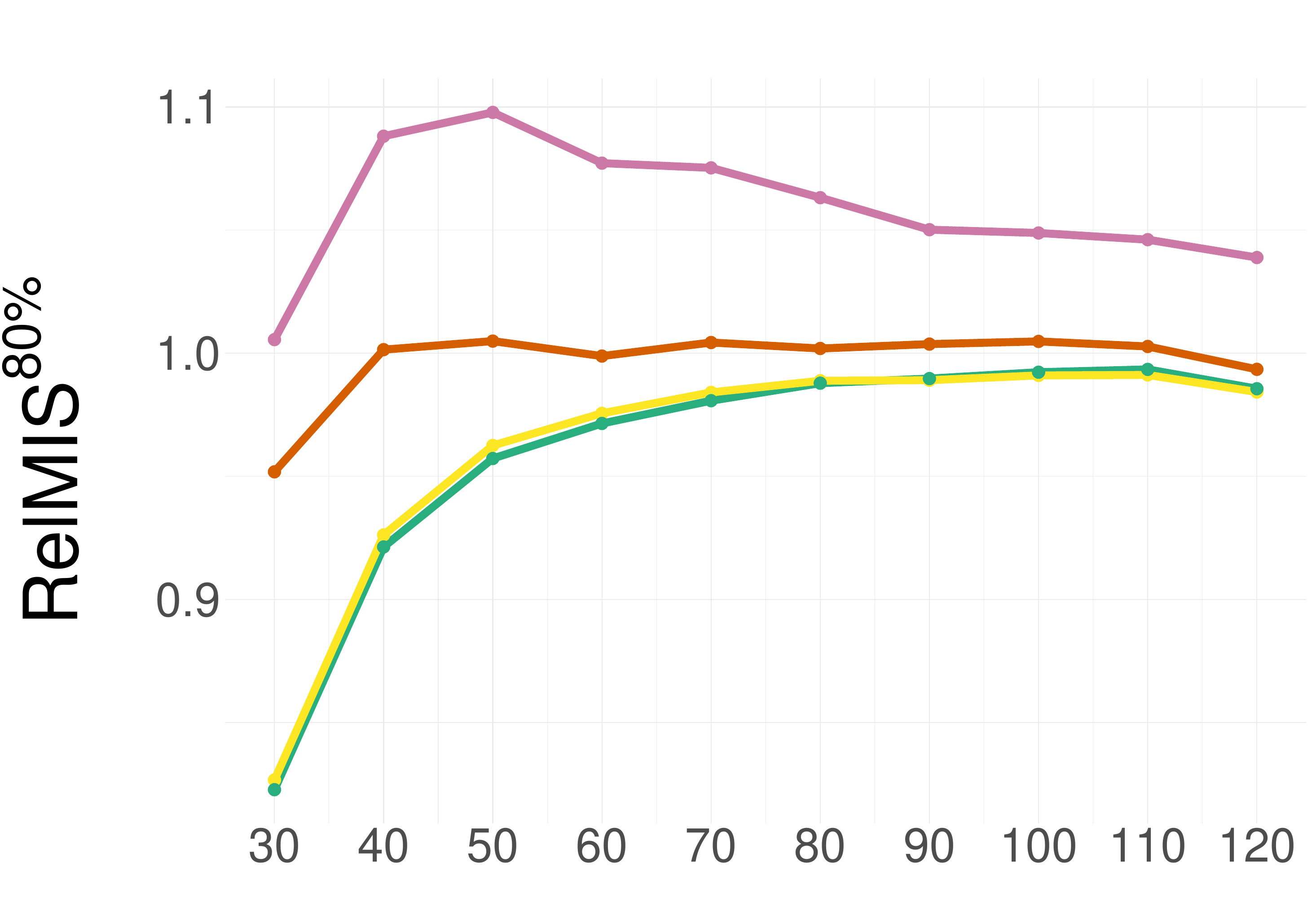}
        \end{overpic}
    \end{subfigure}
    \hfill
    \begin{subfigure}[b]{0.42\textwidth}
        \centering
        \begin{overpic}[width=\linewidth]{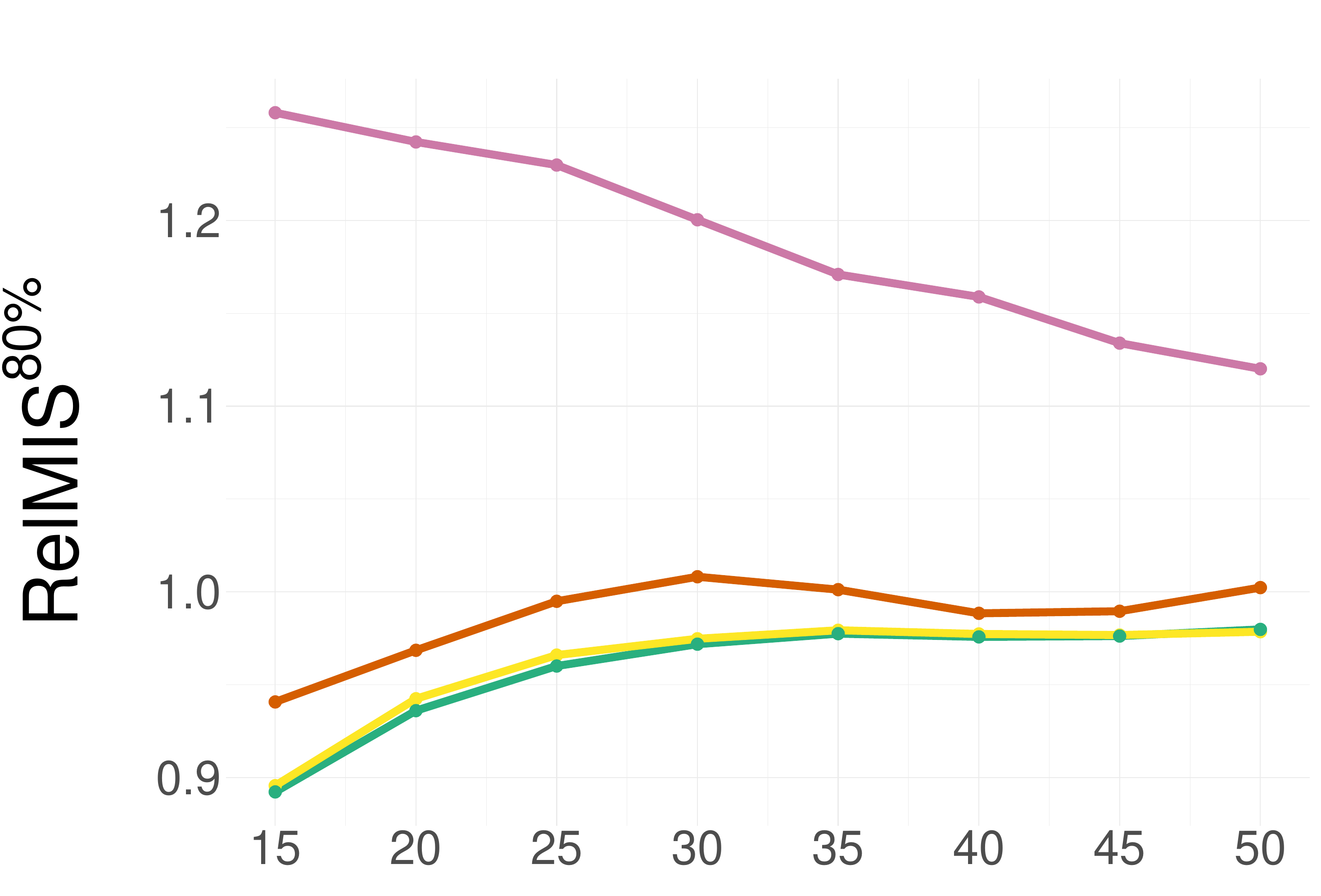}
        \end{overpic}
    \end{subfigure}

    \vspace{-0.2cm} 

    \begin{subfigure}[b]{0.42\textwidth}
        \centering
        \begin{overpic}[width=\linewidth]{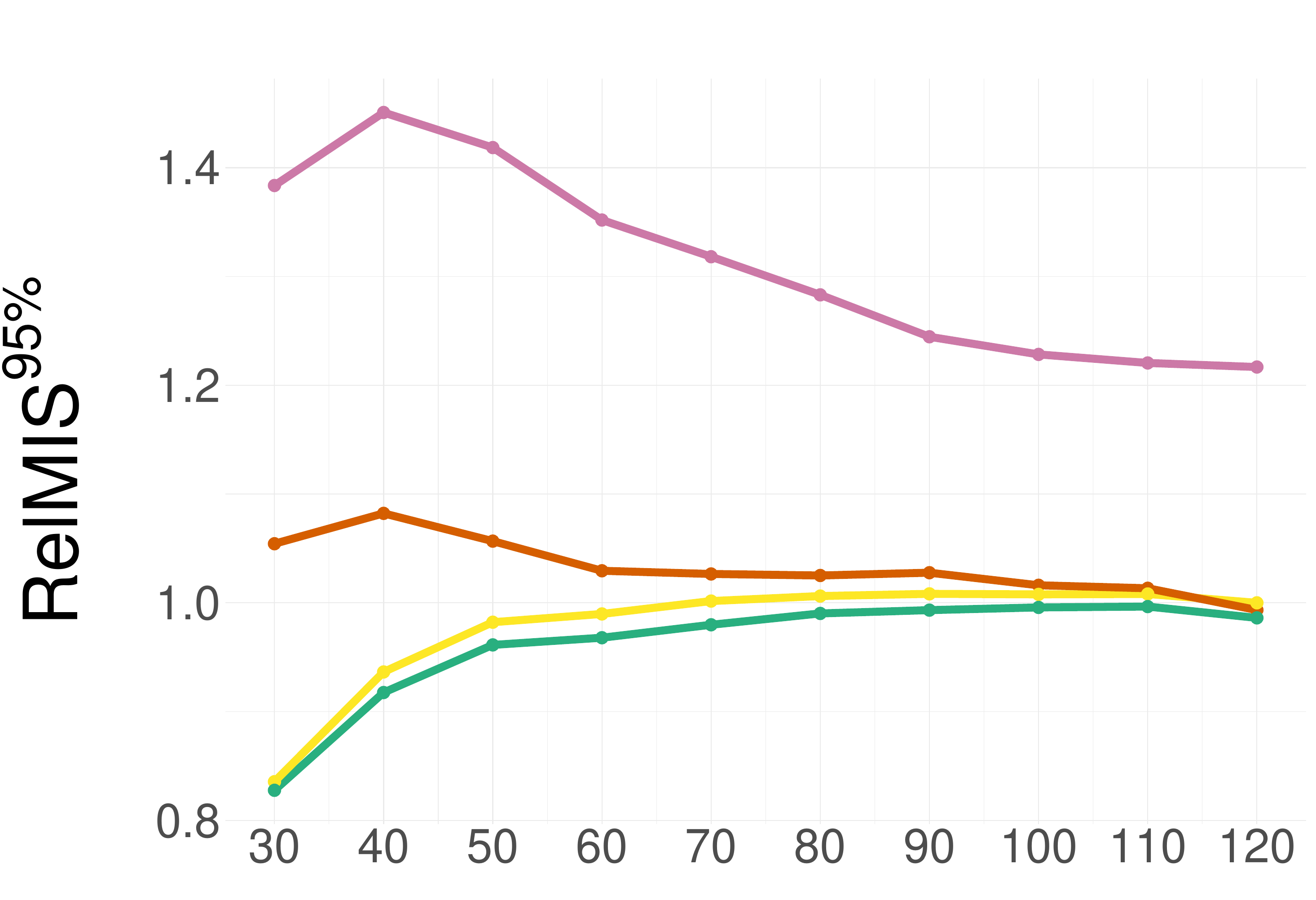}
        \end{overpic}
    \end{subfigure}
    \hfill
    \begin{subfigure}[b]{0.42\textwidth}
        \centering
        \begin{overpic}[width=\linewidth]{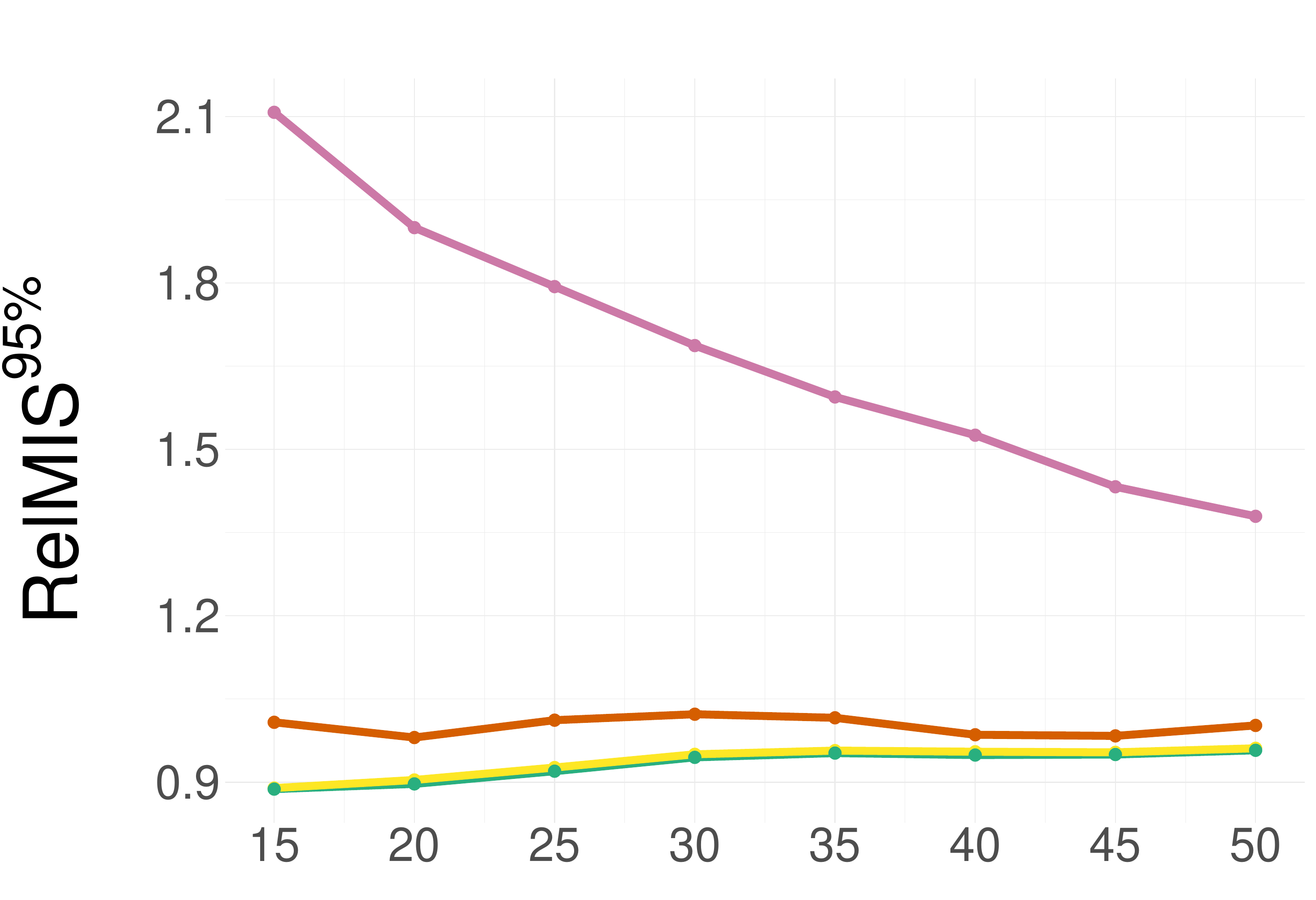}
        \end{overpic}
    \end{subfigure}

    \vspace{-0.2cm} 

    \begin{subfigure}[b]{0.42\textwidth}
        \centering
        \begin{overpic}[width=\linewidth]{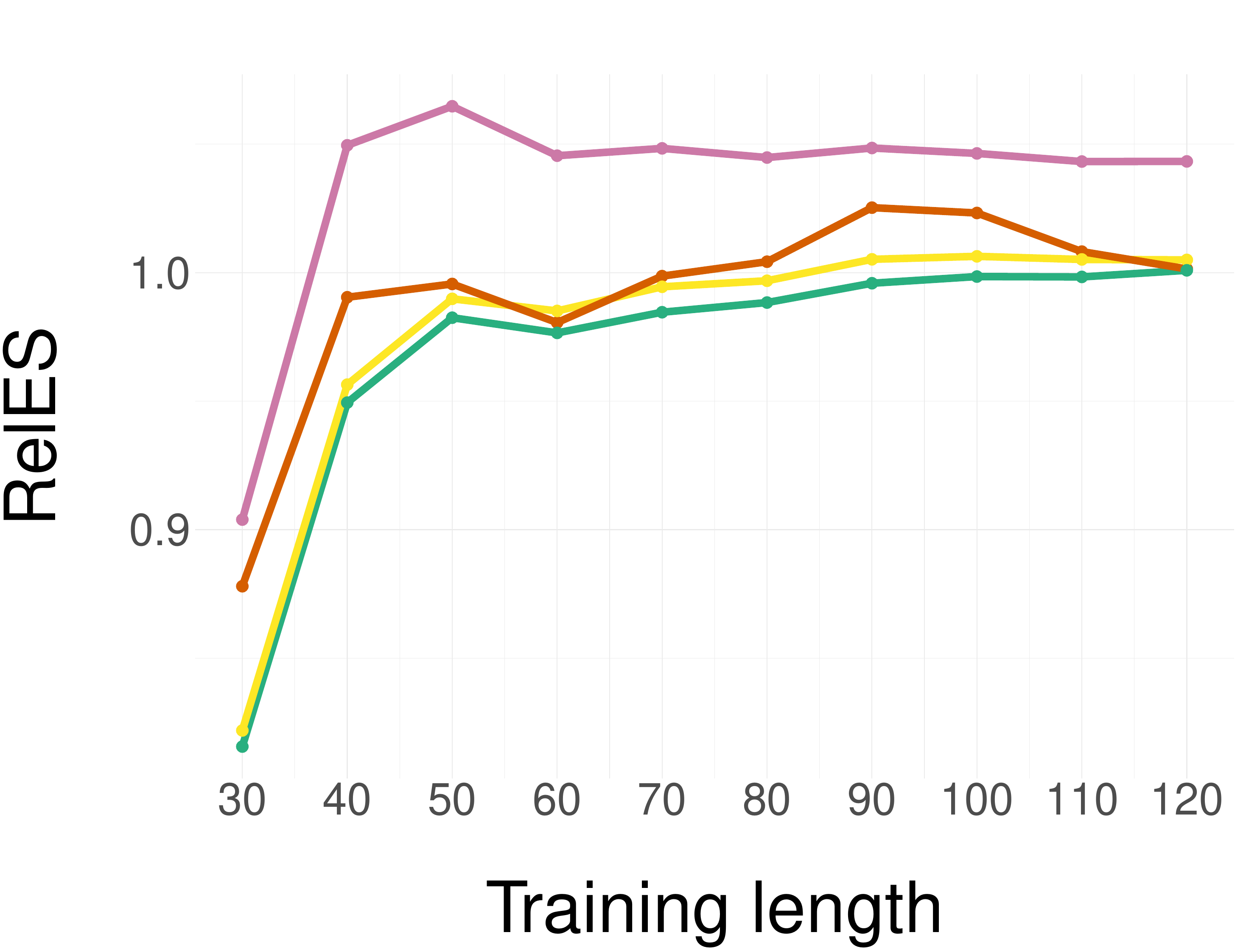}
        \end{overpic}
    \end{subfigure}
    \hfill
    \begin{subfigure}[b]{0.42\textwidth}
        \centering
        \begin{overpic}[width=\linewidth]{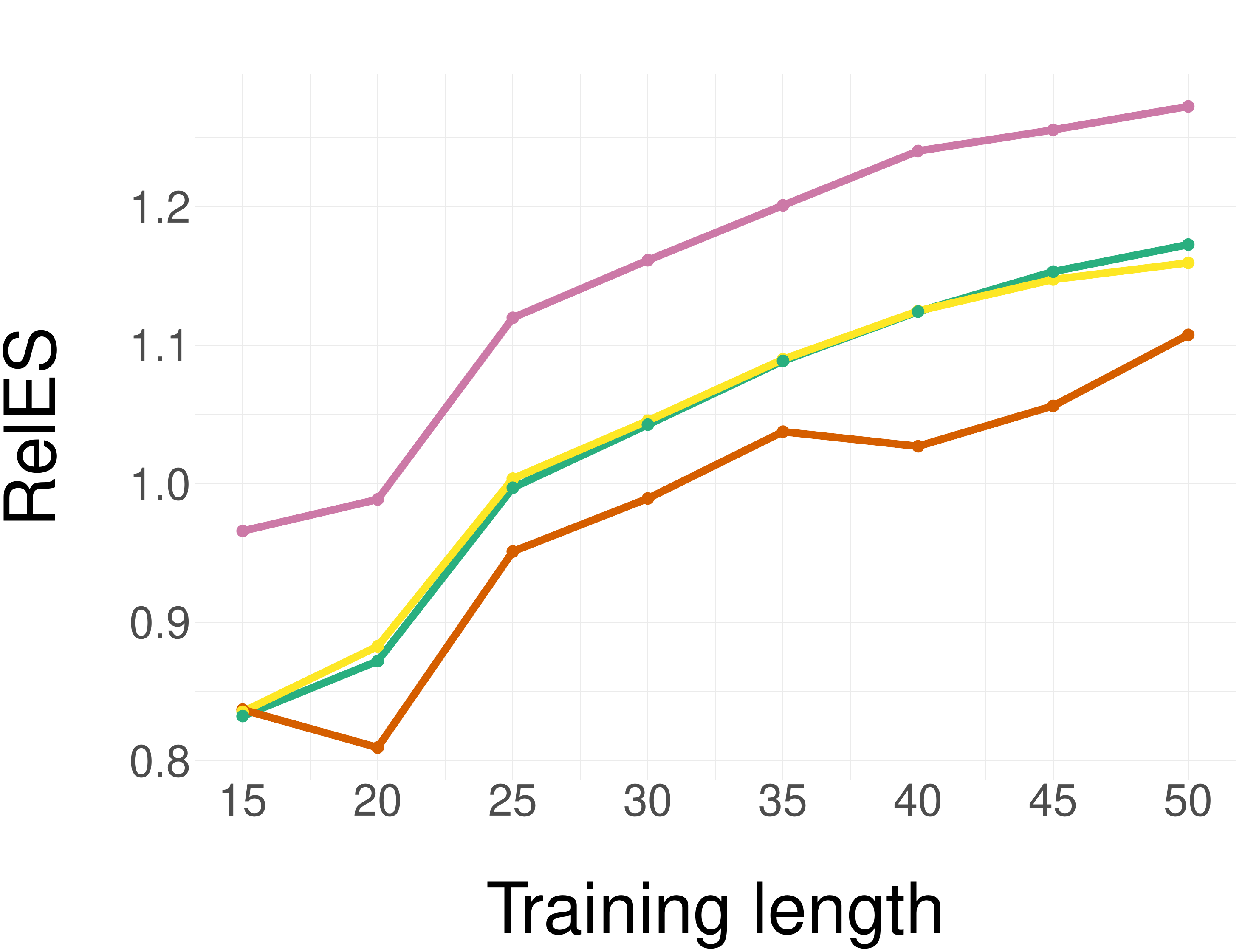}
        \end{overpic}
    \end{subfigure}

    \vspace{0.25cm} 

    \caption{   
    Results on the \textit{Australian Tourism-M} (left column) and \textit{Australian Tourism-Q} (right column) datasets for \textit{\tRec} (green), \textit{\tRec-Diag} (yellow), \textit{\tRec-MAP} (pink), and \textit{\tRec-min\_$\nu_0$} (orange).
   A relative score lower than 1 means improvement over the base forecasts.}
    \label{fig: rel scores aus ablation}
\end{figure}

\end{document}